\documentclass[onecolumn,notitlepage,floatfix,superscriptaddress]{revtex4-2}

\usepackage[lmargin=.7in,rmargin=.7in,tmargin=.6in,bmargin=0.8in]{geometry}
\usepackage{graphicx}
\usepackage{amsmath}
\usepackage{amsthm}
\usepackage{amssymb}
\usepackage{latexsym}
\usepackage{array}
\usepackage{hyperref}
\usepackage{float}
\usepackage{amsfonts}
\usepackage{dsfont}
\usepackage{mathrsfs}
\usepackage{verbatim}
\usepackage{bbold}
\usepackage{lipsum}
\usepackage[normalem]{ulem}
\usepackage{upgreek}
\usepackage{makecell}
\usepackage{adjustbox,lipsum}

\usepackage{algorithm}
\usepackage{algpseudocode}

\usepackage{textcomp}
\usepackage[dvipsnames]{xcolor}

\usepackage{bm}
\usepackage{times}

\newcommand{\bra}[1]{\left<#1\right|}
\newcommand{\ket}[1]{\left|#1\right>}
\newcommand{\abs}[1]{\left|#1\right|}
\newcommand{\norm}[1]{\left\lVert#1\right\rVert}

\newcommand{\braket}[2]{\left<{#1}|{#2}\right>}
\newcommand{\ketbra}[2]{\ket{#1}\!\!\bra{#2}}

\newcommand{\tr}[1]{\mathrm{Tr}{#1}}

\newtheorem{theorem}{Theorem}
\newtheorem{proposition}{Proposition}
\newtheorem{lemma}{Lemma}
\newtheorem{corollary}{Corollary}
\newtheorem{definition}{Definition}
\newtheorem{remark}{Remark}

\makeatletter
\renewcommand\part[1]{%
  \clearpage
  \onecolumngrid
  \section*{#1}
  
}
\makeatother

\makeatletter
\let\origaddcontentsline\addcontentsline
\renewcommand{\addcontentsline}[3]{}
\makeatother

\begin{document}

\part{}

\title{Sharpening Worst-Case Error Assessment for Fault-Tolerant Quantum Computing: Fidelity and Its Deviation}

\author{Kyoungho~Cho}
\affiliation{Institute for Convergence Research and Education in Advanced Technology, Yonsei University, Seoul 03722, Republic of Korea}
\affiliation{Department of Statistics and Data Science, Yonsei University, Seoul 03722, Republic of Korea}

\author{Ilkwon~Sohn}
\affiliation{Quantum Network Research Center, Korea Institute of Science and Technology Information, Daejeon 34141, Korea}

\author{Yongsoo~Hwang}\email{yhwang@etri.re.kr}
\affiliation{Electronics and Telecommunications Research Institute, Daejeon 34129, Korea}

\author{Jeongho~Bang}\email{jbang@yonsei.ac.kr}
\affiliation{Institute for Convergence Research and Education in Advanced Technology, Yonsei University, Seoul 03722, Republic of Korea}
\affiliation{Department of Quantum Information, Yonsei University, Incheon 21983, Republic of Korea}

\date{\today}

\begin{abstract}
Gate fidelity---an average fidelity over all possible input states---is the workhorse metric for benchmarking quantum gates or circuits, yet fault-tolerant quantum computing ultimately depends on the worst-case behavior, typically quantifiable by so-called the diamond distance. In the low-error regime, the coherent errors can inflate the worst-case error even when the reported gate fidelity is high, making the gate fidelity alone an unreliable proxy for fault-tolerance readiness. To capture the missing information, we introduce a companion observable---what we dub the fidelity deviation---that quantifies how strongly the state-dependent fidelities fluctuate across input states. Adopting such fluctuations in assessing the fault-tolerance is physically natural because some input directions are nearly unaffected while others form narrow ``valleys'' that dominate adversarial circuit behavior. For coherent (unitary) gate errors on two or more qubits, we show that the gate fidelity together with the fidelity deviation constrains the relevant spectral moments of the error unitary, enabling an explicit and tight certificate of the worst-case error. Both quantities are estimated directly from the same randomized input-measurement experiment, without full process tomography. We show that the fidelity and its deviation can provide an economical, operationally meaningful, and accurate standard for assessing the fault tolerance of the engineered quantum gates and circuits.
\end{abstract}

\maketitle


The characterization of gate or circuit performance is a central task in quantum computing, and it becomes especially critical in fault-tolerant quantum computation (FTQC)~\cite{Steane1999,Martinis2015,Postler2022,Zhou2025}. Typically, FTQC is governed by the worst-case behavior: gates are embedded deep inside large circuits and used repeatedly, and a small subset of adversarial input directions can dominate the failure probability of an error-correcting gadget. The circuit-level metric aligned with this viewpoint is the diamond distance~\cite{Kitaev2002,Gilchrist2005}. For channels $\mathcal{E}$, $\mathcal{F}$ acting on $\mathcal{H}$, the (normalized) diamond distance is
\begin{eqnarray}
d_{\diamond}(\mathcal{E},\mathcal{F}) := \frac{1}{2}\norm{\mathcal{E}-\mathcal{F}}_{\diamond} = \frac{1}{2}\max_{\hat{\rho}}\norm{(\mathcal{E}\otimes \mathcal{I})(\hat{\rho})-(\mathcal{F}\otimes \mathcal{I})(\hat{\rho})}_{1},
\label{eq:intro_diamond_def}
\end{eqnarray}
where the maximization is over density operators on the system plus an ancilla of dimension $d=\dim\mathcal{H}$. However, it is not directly accessible without essentially tomographic information.

Thus, in practice, the gate (in)fidelity is widely used~\cite{Nielsen2002}. For the state-dependent single fidelity $f_{\mathcal{E}}(\psi) := \bra{\psi}\mathcal{E}\bigl(\ketbra{\psi}{\psi}\bigr)\ket{\psi}$ and the interaction-picture error channel $\mathcal{E}$, the gate fidelity $F$ is defined as the Haar average of the single fidelities, i.e., 
\begin{eqnarray}
F := F_{\rm avg}(\mathcal{E}) := \int d\psi\, f_{\mathcal{E}}(\psi).
\label{eq:intro_f_def}
\end{eqnarray}
The gate fidelity $F$ can be estimated efficiently and is robust against state-preparation-and-measurement (SPAM) errors using scalable protocols, such as randomized benchmarking (RB)~\cite{Knill2008,Magesan2011}. Concretely, given an ideal unitary gate $\hat{U}$ on $\mathbb{C}^d$ and an implemented completely-positive trace-preserving (CPTP) map $\mathcal{G}$, it is convenient to work in the interaction picture and define the error channel $\mathcal{E} := \mathcal{U}^{\dagger}\circ \mathcal{G}$, where $\mathcal{U}(\hat{\rho})=\hat{U}\hat{\rho}\hat{U}^{\dagger}$, so that $\mathcal{E}=\mathcal{I}$ corresponds to perfect implementation. However, average fidelity data do \emph{not} by themselves provide sharp worst-case guarantees. In the low-error regime, coherent (systematic/unitary) errors make the worst-case behavior parametrically larger than what the average (in)fidelity suggests (the familiar square-root amplification), and unitarity-based refinements lose resolution once the noise becomes predominantly coherent~\cite{Gilchrist2005,Sanders2016,Kueng2016}. 

A major refinement proposed in prior work introduces the unitarity $u$ and derives $(r, u)$-based bounds for unital noise without leakage~\cite{Kueng2016}. This refinement is powerful when the noise has a significant incoherent component because deviations of $u$ from unity certify the favorable linear scaling $d_{\diamond}=O(r)$. However, in the coherence-dominated regime that is most dangerous for FTQC, unitarity necessarily saturates to $u \simeq 1$ (and equals $1$ for purely unitary errors), so the $(r, u)$ strategy loses resolution within the coherent manifold.

In this work, we develop a complementary strategy. Our starting point is to treat gate quality as a landscape over input states rather than a single number. We formalize our idea by introducing the \emph{fidelity deviation} $D$, 
\begin{eqnarray}
D := D(\mathcal{E}) := \sqrt{\int d\psi\, f_{\mathcal{E}}(\psi)^2 - F^2}.
\label{eq:intro_D_def}
\end{eqnarray}
Physically, $D$ quantifies how non-uniformly the gate acts across input states. This is precisely the information that becomes decisive for coherent control errors, where narrow ``valleys'' in the fidelity landscape can drive worst-case risk even when the gate (averaged) fidelity is high~\cite{Sanders2016,Kueng2016}.

\begin{figure}[t]
\centering
\includegraphics[width=1.00\linewidth]{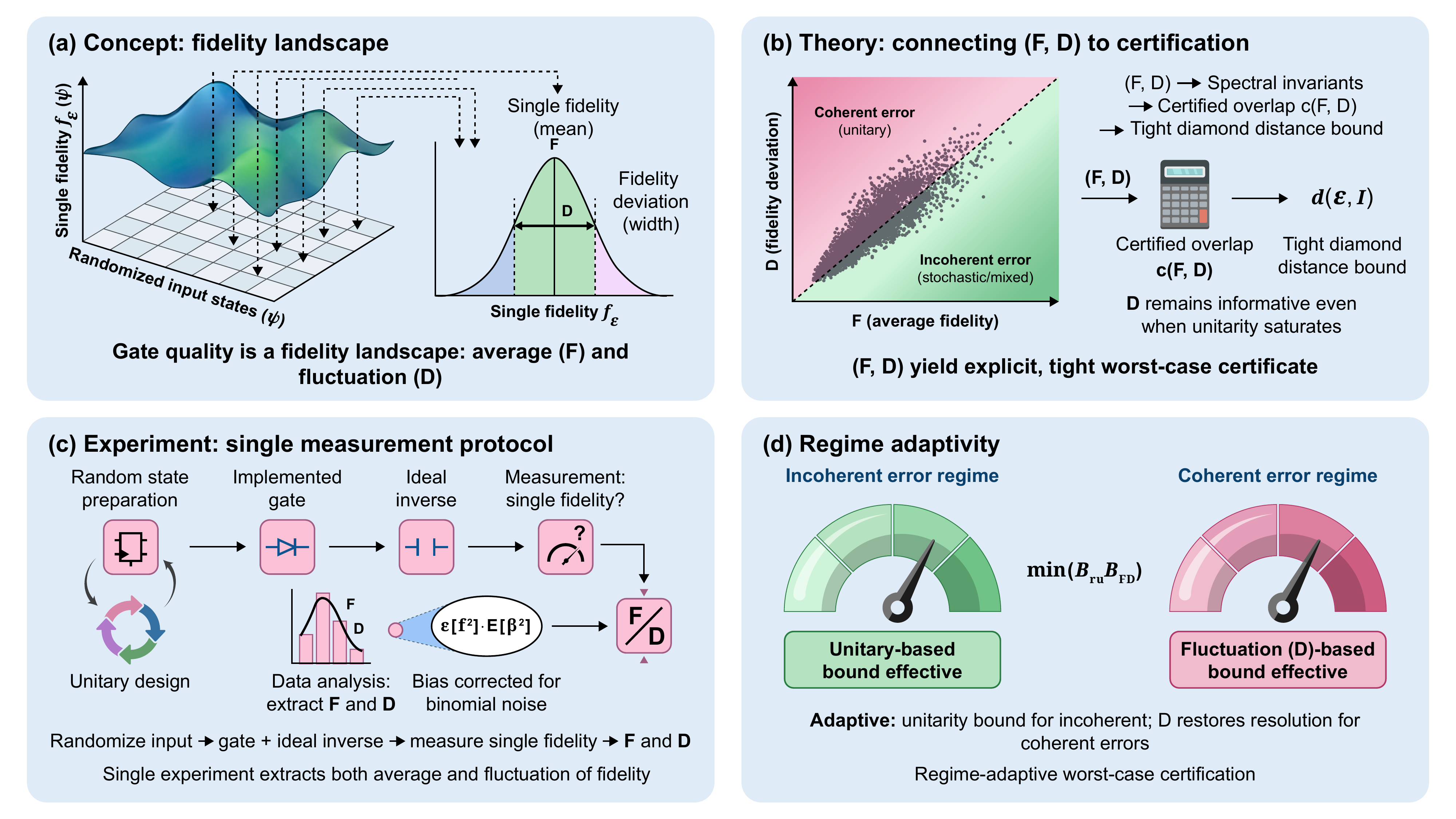}
\caption{\textbf{Fluctuation-assisted worst-case certification from $(F,D)$.} \textbf{(a) Concept} The quality of a quantum gate should be treated as a fidelity landscape rather than a single number. From the interaction-picture error channel $\mathcal{E}$, one samples the state-dependent single fidelities $f_{\mathcal{E}}(\psi)$ over randomized input states $\ket{\psi}$. The mean of this distribution is the gate fidelity $F$, while its width is the fidelity deviation $D$. \textbf{(b) Theory.} In the coherent (unitary) regime, $D$ remains informative even when unitarity saturates. For two-qubit and larger coherent errors, the pair $(F, D)$ fixes the spectral moment invariants of the effective error unitary and yields an explicit certified overlap $c(F, D)$ ({\bf Theorem~\ref{thm:main}}), which directly translates into a tight worst-case certificate on the diamond distance. \textbf{(c) Experiment.} Both $F$ and $D$ are extracted from the same randomized input-measurement experiment: one prepares random states (for example from a unitary design), applies the implemented gate followed by the ideal inverse, and performs a projective ``identity'' test. A simple factorial-moment correction removes the bias from finite-shot (binomial) noise when estimating the second moment, enabling an unbiased estimate of $D$ from the same data stream used for $F$ (Methods). \textbf{(d) Regime adaptivity.} When the incoherence dominates, the unitarity-based bounds can certify linear-in-infidelity worst-case behavior; when the coherence dominates and unitarity saturates, the fluctuation information in $D$ restores high-resolution worst-case certification.}
\label{fig:overview}
\end{figure}

Our main result shows that for coherent (unitary) errors on two or more qubits, the pair $(F,D)$ fixes spectral moment invariants of the effective error unitary and yields a tight upper bound on the diamond distance. Importantly, $D$ can be estimated from the same randomized input-measurement experiment used for fidelity estimation, using a simple shot-noise correction. We illustrate the tightening on three coherent examples (a CZ-like two-qubit phase error, a decomposed Toffoli gate, and a $10$-qubit QFT circuit) and discuss a regime-adaptive certification workflow that combines the strengths of our $(F, D)$ and unitarity-based measure proposed in Ref.~\cite{Kueng2016}. Schematic overview of our work is illustrated in Fig.~\ref{fig:overview}.

\section*{Results}\label{sec:results}

\medskip
\paragraph*{\bf Gate fidelity is not a worst-case certificate.}

A key reason that the gate fidelity alone can be misleading for FTQC is that $F$ is an average-over-input overlap-squared quantity for a single use of the gate, whereas $d_{\diamond}$ quantifies the worst-case distinguishability under arbitrary (possibly entangled) inputs. In a fault-tolerant circuit, gates are composed and nested inside gadgets, and adversarial input directions can be selected implicitly by the preceding noisy evolution. In such settings, the relevant question becomes: what is the largest deviation a gate can induce on any state that may be entangled with the rest of the computer?

This mismatch already appears in the best-known general conversion between the average infidelity and diamond distance. If $r:=1-F$ is the average infidelity, then one only has inequalities of the form~\cite{Gilchrist2005,Sanders2016}
\begin{eqnarray}
\frac{d+1}{d} r  \le d_{\diamond}(\mathcal{E}, \mathcal{I}) \le \sqrt{d(d+1) r},
\label{eq:intro_conversion}
\end{eqnarray}
so even when $r \ll 1$ the worst-case error can scale as $\Theta(\sqrt{r})$. This square-root amplification is structural: it is saturated by coherent miscalibrations, where the gate error is a small unitary over-rotation. A simple illustration is a single-qubit coherent rotation error $\hat{X}=e^{-i\delta \hat{\sigma}_z}$, for which
\begin{eqnarray}
r(\delta)=1-F(\delta)=\frac{2}{3}\sin^2(\delta)\ \simeq\ \frac{2}{3}\delta^2,
\quad
d_{\diamond}(\mathcal{X},\mathcal{I})=\left|\sin(\delta)\right|\ \simeq\ |\delta|.
\label{eq:results_sqrt_r_example}
\end{eqnarray}
Thus, a gate can report a very small infidelity (quadratic in $\delta$) while already having a percent-level worst-case distinguishability (linear in $\delta$). In multi-qubit settings, this effect is amplified by the increased freedom of the error spectrum.

A widely used refinement supplements $r$ with the unitarity $u$~\cite{Kueng2016}, which acts as a regime witness: when incoherence is present, deviations of $u$ from $1$ can certify the favorable linear scaling $d_{\diamond}=O(r)$. Concretely, for unital noise without leakage one has an upper bound of the form
\begin{eqnarray}
d_{\diamond}(\mathcal{E},\mathcal{I}) \le d^2 c_d\sqrt{u+\frac{2dr}{d-1}-1},
\quad
c_d := \frac{1}{2}\sqrt{1-\frac{1}{d^2}}.
\label{eq:results_ru_bound}
\end{eqnarray}
However, in the fully coherent limit $u=1$ identically, so $u$ cannot resolve different coherent errors. In fact, Eq.~(\ref{eq:results_ru_bound}) reduces to a dimension-amplified version of the fidelity-only conversion,
\begin{eqnarray}
d_{\diamond}(\mathcal{E},\mathcal{I}) \le \frac{d}{\sqrt{2}}\sqrt{d(d+1)\,r} \quad (u=1),
\end{eqnarray}
which can be substantially looser as the system size grows. This motivates an additional experimentally accessible observable that remains sensitive within the coherent manifold. For more details, see Sec.~II of the Supplementary Information.

\medskip
\paragraph*{\bf Fidelity deviation as a fluctuation observable.}

The fidelity deviation $D$ defined in Eq.~(\ref{eq:intro_D_def}) is the standard deviation of the single fidelities $f_{\mathcal{E}}(\psi)$ whose Haar average is $F$. Conceptually, the pair $(F, D)$ treats the gate quality as a distribution (or landscape) rather than a single number: $F$ measures the mean height, while $D$ measures the width (state dependence). This fluctuation information is precisely what the gate fidelity discards, and it becomes decisive when a small set of low-fidelity ``valleys'' determines the worst-case behavior.

Two basic sanity checks are immediate. First, $D=0$ if and only if the single fidelity is constant almost everywhere on the sphere, meaning that the gate acts uniformly well (or uniformly badly) on all inputs. Second, because $0\le f_{\mathcal{E}}(\psi)\le 1$, one always has the trivial variance bound
\begin{eqnarray}
0 \le D^2 \le F(1-F),
\label{eq:results_D_trivial_bound}
\end{eqnarray}
so $D$ is naturally of the same order as the infidelity in typical low-error regimes. As a concrete reference point, for an isotropic depolarizing error channel $\mathcal{E}_{\rm dep}(\hat{\rho})=(1-p)\hat{\rho} + p\hat{\mathds{1}}/d$, the single fidelity is state-independent, i.e., $f_{\mathcal{E}_{\rm dep}}(\psi)=(1-p)+p/d$, and hence $D(\mathcal{E}_{\rm dep})=0$ even when $F<1$. In contrast, coherent unitary mis-rotations generically create strong state dependence: at fixed average infidelity $r$, the coherent errors can have markedly different valley structure. Here, $D$ provides exactly the extra ``roughness'' information needed to distinguish them. To see this, let $\hat{U}_{\rm ideal}$ and $\hat{U}_{\rm exp}$ be the ideal and implemented unitaries, and define the effective error unitary as
\begin{eqnarray}
\hat{X} := \hat{U}_{\rm ideal}^{\dagger}\hat{U}_{\rm exp} \in U(d),
\label{eq:results_X_def}
\end{eqnarray}
so that the interaction-picture error channel is $\mathcal{X}(\hat{\rho})=\hat{X}\hat{\rho}\hat{X}^{\dagger}$. Then, the state-dependent single fidelity simplifies to
\begin{eqnarray}
f_{\mathcal{X}}(\psi) = \abs{\bra{\psi}\hat{X}\ket{\psi}}^2,
\end{eqnarray}
and therefore $F$ and $D$ become Haar moments of the same overlap amplitude $\abs{\bra{\psi}\hat{X}\ket{\psi}}$.
For $d=2$ (i.e., for the single-qubit gates), these moments are not independent (indeed, $D$ is fixed once $F$ is known), but for two-qubit and larger gates ($d\ge 4$) the deviation provides genuinely new information about the unitary spectrum that is invisible to the mean fidelity alone.

\medskip
\paragraph*{\bf From $(F, D)$ to a ``tight'' worst-case certificate.}

The relation between $D$ and the worst-case error becomes especially direct. For a unitary error channel $\mathcal{X}(\hat{\rho})=\hat{X}\hat{\rho}\hat{X}^{\dagger}$, the worst-case error is governed by the smallest overlap that $\hat{X}$ can have with any input state. This is captured by the minimum overlap amplitude
\begin{eqnarray}
m(\hat{X})=\min_{\ket{\psi}}\abs{\bra{\psi}\hat{X}\ket{\psi}},
\label{eq:min_X}
\end{eqnarray}
which admits both an operational and a geometric interpretation.
Define the numerical range
\begin{eqnarray}
W(\hat{X}) := \Bigl\{\bra{\psi}\hat{X}\ket{\psi}\ :\ \ket{\psi}\in \mathbb{C}^d,\ \braket{\psi}{\psi}=1\Bigr\}\subset\mathbb{C}.
\label{eq:results_numerical_range}
\end{eqnarray}
For a unitary $\hat{X}$, $W(\hat{X})$ is the convex hull of the eigenvalues of $\hat{X}$ on the unit circle, so $m(\hat{X})=\min_{z\in W(\hat{X})}|z|$ is simply the distance from the origin to this spectral polygon. The worst-case behavior becomes large precisely when this polygon approaches (or encloses) the origin---a geometric way to visualize the coherent ``valleys'' of the fidelity landscape.

A standard fact then turns the diamond-distance problem into the problem of lower-bounding $m(\hat{X})$:
\begin{lemma}[Diamond distance for unitary errors]\label{lem:unitary_diamond}
Let $\mathcal{X}(\hat{\rho})=\hat{X}\hat{\rho}\hat{X}^{\dagger}$ with $\hat{X}\in U(d)$. Then,
\begin{eqnarray}
d_{\diamond}(\mathcal{X},\mathcal{I}) = \sqrt{1-m(\hat{X})^2}.
\label{eq:lem_unitary_diamond}
\end{eqnarray}
\end{lemma}

\begin{proof}[Proof sketch.]---By convexity of the trace norm in the definition of $d_{\diamond}$, the maximization can be restricted to pure states on the system--ancilla space, $\hat{\rho}_{SA}=\ketbra{\Psi}{\Psi}$. For two pure states $\ket{\Psi}$ and $\ket{\Phi}$, one has $\frac{1}{2}\|\ketbra{\Psi}{\Psi}-\ketbra{\Phi}{\Phi}\|_1=\sqrt{1-\abs{\braket{\Psi}{\Phi}}^2}$. Taking $\ket{\Phi}=(\hat{X}\otimes\hat{\mathds{1}}_A)\ket{\Psi}$ therefore yields
\begin{eqnarray}
d_{\diamond}(\mathcal{X},\mathcal{I})=\max_{\ket{\Psi}}\sqrt{1-\abs{\bra{\Psi}(\hat{X}\otimes\hat{\mathds{1}}_A)\ket{\Psi}}^2}.
\end{eqnarray}
The overlap depends only on the reduced system state $\hat{\rho}_S=\tr{{}_A\ketbra{\Psi}{\Psi}}$ via $\bra{\Psi}(\hat{X}\otimes\hat{\mathds{1}}_A)\ket{\Psi}=\tr{(\hat{X}\hat{\rho}_S)}$. As $\hat{\rho}_S$ ranges over all density operators, $\tr{(\hat{X}\hat{\rho}_S)}$ ranges over the numerical range $W(\hat{X})$, which is convex; hence allowing mixed $\hat{\rho}_S$ (equivalently, an ancilla) does not enlarge the achievable overlap set beyond $\{\bra{\psi}\hat{X}\ket{\psi}\}$. Because $x\mapsto \sqrt{1-x^2}$ decreases on $[0,1]$, the maximization reduces to the minimum overlap amplitude $m(\hat{X})$, giving Eq.~(\ref{eq:lem_unitary_diamond}).
\end{proof}

{\bf Lemma~\ref{lem:unitary_diamond}} shows that the coherent worst-case certification reduces to a spectral-geometric question: how close can $W(\hat{X})$ get to the origin? The central observation of this work is that the pair $(F,D)$ provides exactly the right amount of the moment information to constrain this geometry for two-qubit and larger coherent errors. To make this explicit, define two spectral trace invariants
\begin{eqnarray}
P := \abs{\tr{(\hat{X})}},
\quad
E_2 := \int d\psi\, f_{\mathcal{X}}(\psi)^2 = D^2+F^2,
\label{eq:results_P_E2}
\end{eqnarray}
and define $Q$ as a fourth-moment invariant (see Methods for the derivation):
\begin{eqnarray}
Q := \abs{ \tr{(\hat{X}^2)} + \tr{(\hat{X})}^2 }.
\label{eq:results_Q_def}
\end{eqnarray}
The key point is that $F$ fixes $P$ through a second-moment Haar identity, while $D$ fixes $Q$ through a fourth-moment identity. In particular,
\begin{eqnarray}
F = \frac{d+P^2}{d(d+1)},
\quad
E_2=\frac{2d(d+3)+4(d+2)P^2+Q^2}{d(d+1)(d+2)(d+3)}.
\label{eq:results_moments_closed}
\end{eqnarray}
For $d\ge 4$ (two qubits and above), Eq.~(\ref{eq:results_moments_closed}) implies that the experimentally accessible pair $(F,D)$ uniquely determines the spectral invariants $(P,Q)$ (equivalently, $P^2$ and $Q^2$). For completeness, we summarize this moment-to-spectrum conversion as the following lemma.

\begin{lemma}[Spectral invariants from $(F,D)$ for $d\ge 4$]\label{lem:PQ_from_FD}
Let $d \ge 4$ and let $\hat{X}\in U(d)$ be the effective error unitary. Then, $(F,D)$ uniquely determines $P$ and $Q$. In particular,
\begin{eqnarray}
P^2 &=& d(d+1)F-d, \nonumber \\
Q^2 &=& d(d+1)(d+2)(d+3)(D^2+F^2) - 2d(d+3) - 4(d+2)P^2.
\end{eqnarray}
\end{lemma}

\begin{proof}[Proof sketch.]---For a unitary error channel, $f_{\mathcal{X}}(\psi)=\abs{\bra{\psi}\hat{X}\ket{\psi}}^2$. The Haar integrals of $(\ketbra{\psi}{\psi})^{\otimes 2}$ and $(\ketbra{\psi}{\psi})^{\otimes 4}$ reduce to contractions with the symmetric projectors (Methods), yielding Eq.~(\ref{eq:results_moments_closed}). Expanding $\abs{\tr{(\hat{X}^2)}+\tr{(\hat{X})}^2}^2$ collects the fourth-order trace terms, so solving Eq.~(\ref{eq:results_moments_closed}) for $P^2$ and $Q^2$ (using $E_2=D^2+F^2$) gives the stated expressions.
\end{proof}

One can therefore ask and answer an extremal question: among all unitaries $\hat{X}\in U(d)$ consistent with the observed moments, what is the smallest possible $m(\hat{X})$? The solution yields an explicit certified minimum overlap $c(F,D)$.

\begin{theorem}[Moment-assisted coherent certificate for $d\ge 4$]\label{thm:main}
Let $d \ge 4$ and let $\hat{X} \in U(d)$ be the effective error unitary. Let $F$ and $D$ be defined by Eq.~(\ref{eq:intro_f_def}) and Eq.~(\ref{eq:intro_D_def}) for $\mathcal{X}(\hat{\rho})=\hat{X}\hat{\rho}\hat{X}^{\dagger}$. Then, the minimum overlap obeys
\begin{eqnarray}
m(\hat{X}) \ge c(F,D) := \left[\frac{P}{d}-\frac{\sqrt{(d-2)\bigl(dQ+d^2-(d+2)P^2\bigr)}}{2d}\right]_{+},
\label{eq:thm_c_def}
\end{eqnarray}
where $[x]_{+}:=\max\{x, 0\}$. Consequently, we have
\begin{eqnarray}
d_{\diamond}(\mathcal{X},\mathcal{I})  \le \sqrt{1-c(F,D)^2}.
\label{eq:thm_diamond_bound}
\end{eqnarray}
The bound is tight for admissible $(F,D)$ arising from unitary errors.
\end{theorem}

\begin{proof}[Proof sketch.]---By {\bf Lemma~\ref{lem:unitary_diamond}}, it suffices to lower bound $m(\hat{X})=\min_{\psi}\abs{\bra{\psi}\hat{X}\ket{\psi}}$. By {\bf Lemma~\ref{lem:PQ_from_FD}}, the data $(F,D)$ fixes the two spectral invariants $(P,Q)$, so one considers the extremal problem of minimizing $m(\hat{X})$ over all $\hat{X}\in U(d)$ consistent with these invariants. Because $\hat{X}$ is normal, $W(\hat{X})=\mathrm{Conv}(\mathrm{Spec}(\hat{X}))$ is the convex hull of eigenvalues on the unit circle, and $m(\hat{X})$ is the distance from the origin to this spectral polygon.
If $0\in W(\hat{X})$ then $m(\hat{X})=0$ and the bound is trivial. Otherwise, the closest point of $W(\hat{X})$ to the origin lies on a boundary chord between two extreme eigenvalues. After a global phase rotation, one can assume this chord is symmetric with endpoints $e^{\pm i\beta}$ and that all eigenvalues lie on the arc $\theta\in[-\beta,\beta]$, so $m(\hat{X})=\cos\beta$.
Writing $b=\cos\beta$ and $x_k=\cos\alpha_k$ for the remaining eigenvalue angles, the constraints imposed by $(P,Q)$ become algebraic constraints on $\sum_k x_k$ and $\sum_k x_k^2$. A Cauchy--Schwarz inequality lower bounds $\sum_k x_k^2$ at fixed $\sum_k x_k$, yielding a quadratic inequality for $b$ whose smaller root is exactly the bracketed expression in Eq.~(\ref{eq:thm_c_def}). Hence, $b \ge c(F,D)$, proving Eq.~(\ref{eq:thm_c_def}) and therefore Eq.~(\ref{eq:thm_diamond_bound}). Tightness follows by saturating the Cauchy--Schwarz step, which is achieved by an explicit two-angle spectrum (one bulk angle and one extremal conjugate pair), so that $m(\hat{X})=c(F,D)$ for a compatible unitary. The proof details are provided in Sec.~V of the Supplementary Information.
\end{proof}

{\bf Theorem~\ref{thm:main}} is a ``moment upgrade'' of coherent worst-case certification. The gate fidelity alone fixes essentially one spectral constraint through $P=\abs{\tr{(\hat{X})}}$, leaving enough freedom for the spectral polygon of $\hat{X}$ to approach the origin and collapse $m(\hat{X})$. The deviation $D$ injects an additional fourth-moment constraint (through $Q$) that controls the coherent spectral spread and thereby sharpens the worst-case guarantees. The tightness follows because for $d \ge 4$ the two moment constraints are sufficient to pin down the extremal geometry.

Operationally, {\bf Theorem~\ref{thm:main}} gives a direct recipe: (i) estimate $(F,D)$ from randomized fidelity data; (ii) compute $(P,Q)$; (iii) evaluate $c(F,D)$ in Eq.~(\ref{eq:thm_c_def}); and (iv) report $\sqrt{1-c(F,D)^2}$ as an explicit data-driven upper bound on the coherent worst-case error. It is also useful to interpret the certified overlap $c(F,D)$. Because $m(\hat{X})=\min_{z\in W(\hat{X})}|z|$, $c(F,D)$ is a data-driven lower bound on the distance from the origin to the numerical range (spectral polygon) of the coherent error. When $c(F,D)$ is close to $1$, the spectrum of $\hat{X}$ is forced to lie in a narrow arc of the unit circle and the gate is certified to be uniformly close to the ideal on all input states. Conversely, if statistical noise or genuinely broad fluctuations make $c(F,D)$ small, then the moment data are compatible with a spectral polygon that approaches the origin, and no nontrivial worst-case guarantee can be asserted from moments alone (the resulting bound correctly returns a value close to $1$). Finally, the tightness statement in {\bf Theorem~\ref{thm:main}} should be read operationally: there exist coherent error unitaries whose worst-case error attains the bound for the same observed $(F,D)$, so additional certification beyond $\sqrt{1-c(F,D)^2}$ would require additional experimentally accessible information beyond these two moments.

\medskip
\paragraph*{\bf Direct estimation of $(F,D)$ from randomized input-measurement data.}

One key practical point is that $D$ is experimentally accessible without full process tomography~\cite{OBrien2004,Blume2017}. Both $F$ and $D^2$ are moments of the same random variable $f_{\mathcal{E}}(\psi)$. Hence, once an experiment can sample the single fidelities over randomized input states, it can estimate $F$ and additionally track the empirical second moment to estimate $D$.

In our direct protocol, one samples $\ket{\psi_i}=\hat{V}_i\ket{0}$ from a unitary design: a $2$-design suffices for estimating $F$, while a $4$-design suffices for estimating the second moment entering $D$~\cite{Dankert2009,Harrow2009,Brandao2016}. For each sampled state, one implements the interaction-picture error channel $\mathcal{E}$ (for example, by applying the implemented gate followed by the ideal inverse), and performs the projective ``identity'' test onto $\ket{\psi_i}$ by applying $\hat{V}_i^{\dagger}$ and measuring in the computational basis. Repeating this test $N$ times yields binomial counts $K_i \sim {\rm Bin}(N, f_{\mathcal{E}}(\psi_i))$.

Crucially, estimating $D$ does \emph{not} require any additional experimental configuration beyond what is already needed for estimating $F$: one simply retains the per-state shot statistics. To remove the upward bias in the naive estimator of $f_{\mathcal{E}}(\psi)^2$ caused by finite-shot (binomial) noise, we use a factorial-moment correction (Methods). Concretely, letting $\hat{f}_i:=K_i/N$, an unbiased estimator of $f_{\mathcal{E}}(\psi_i)^2$ is
\begin{eqnarray}
\widehat{f_i^2}:=\frac{K_i(K_i-1)}{N(N-1)}=\frac{N\hat{f}_i^2-\hat{f}_i}{N-1},
\end{eqnarray}
and the unbiased estimators of the Haar moments are obtained by averaging $\hat{f}_i$ and $\widehat{f_i^2}$ over $M$ random inputs and by using an unbiased cross-average for $F^2$ (Methods). The resulting estimator $(\hat{F},\hat{D})$ can be inserted directly into Eq.~(\ref{eq:thm_diamond_bound}) to produce a data-driven worst-case certificate. Because the protocol reuses the same data stream as fidelity estimation, the additional experimental cost for extracting $D$ is minimal.
\begin{figure}[t]
\centering
\includegraphics[width=1.00\linewidth]{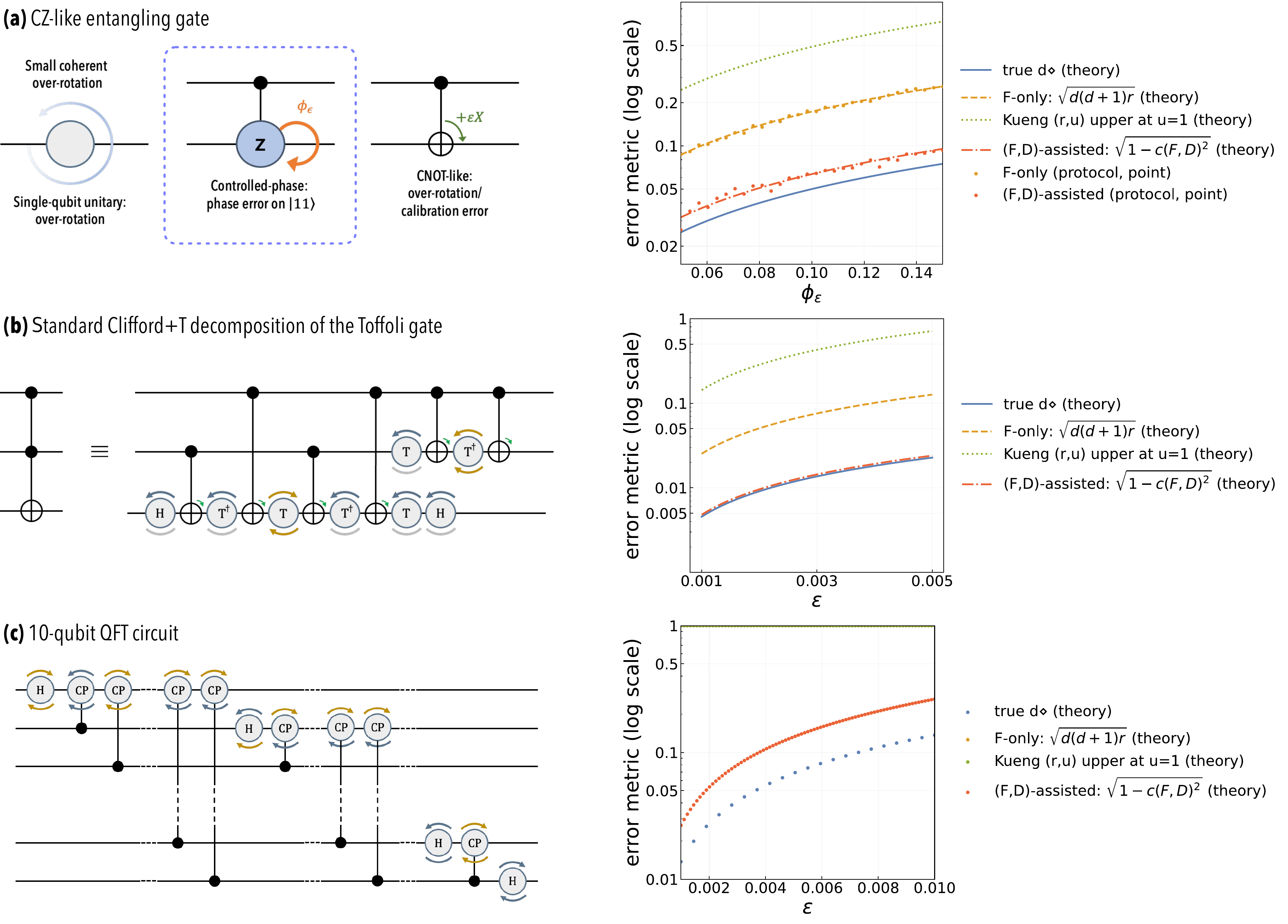}
\caption{\textbf{Three coherent examples and consolidated worst-case certification.} Single-qubit primitives carry a small coherent over-rotation about their native unitary axis (for example, a calibration error on the generator of a $T$ or Hadamard gate). Two-qubit primitives carry a coherent miscalibration, such as, a controlled-phase over-rotation or a controlled-$X$ over-rotation on a CNOT-like interaction. Here, we consider three examples: (a) a CZ-like entangling gate with a coherent phase error on the $\ket{11}$ component, (b) a standard Clifford+$T$ decomposition of the Toffoli gate, where each primitive (single-qubit gates and CNOTs) is affected by the same coherent over-rotation parameter, and (c) the $10$-qubit QFT circuit built from Hadamard and controlled-phase gates, with a uniform coherent over-rotation applied to each primitive. For each example, we compare the exact diamond distance of the resulting coherent error (solid line) with three computable upper bounds obtained from experimentally accessible data: the fidelity-only conversion bound, Kueng et al.'s $(r,u)$-based bound evaluated in the coherent limit, and our proposed $(F, D)$-assisted bound from {\bf Theorem~\ref{thm:main}}. All plots use a logarithmic vertical axis to highlight the low-error regime relevant to fault tolerance. Across all three examples, incorporating the fidelity deviation $D$ yields a substantially tighter worst-case certificate in the coherent regime, while the unitarity-assisted bound becomes increasingly loose with system size because unitarity necessarily saturates for unitary errors.}
\label{fig:examples}
\end{figure}

\medskip
\paragraph*{\bf Examples.}

We test our moment-assisted certificate on three coherent settings: (i) a CZ-like coherent phase miscalibration of a two-qubit entangling gate, (ii) a three-qubit Toffoli gate implemented via a standard decomposition with a uniform coherent over-rotation on each single- and two-qubit primitive, and (iii) a $10$-qubit quantum Fourier transform (QFT) circuit under a coherent over-rotation model on Hadamard and controlled-phase primitives. These examples span a primitive two-qubit entangling gate, a compiled three-qubit non-Clifford gate, and a large ($n=10$) algorithmic circuit module.

For each example, the overall interaction-picture error channel is unitary, so the true diamond distance can be computed exactly from the effective error unitary $\hat{X}$ using {\bf Lemma~\ref{lem:unitary_diamond}}. Numerically, we evaluate
\begin{eqnarray}
d_{\diamond}(\mathcal{X},\mathcal{I})=\sqrt{1-m(\hat{X})^2},
\label{eq:results_exact_diamond_convex_hull}
\end{eqnarray}
where $m(\hat{X})=\min_{z\in \mathrm{Conv}(\mathrm{Spec}(\hat{X}))}\abs{z}$, i.e., we compute the eigenvalues of $\hat{X}$ on the unit circle, take their convex hull in the complex plane, and evaluate the Euclidean distance from the origin to this polygon.
In parallel, we compute three computable upper bounds from experimentally accessible summary data: (i) the fidelity-only conversion $\sqrt{d(d+1)r}$ with $r=1-F$;  (ii) Kueng~et~al.'s $(r,u)$-based bound evaluated in the coherent limit $u=1$, which reduces to $(d/\sqrt{2})\sqrt{d(d+1)r}$ [Eq.~(\ref{eq:results_ru_bound})]; and (iii) our proposed $(F, D)$-assisted bound $\sqrt{1-c(F,D)^2}$ from {\bf Theorem~\ref{thm:main}}. The theory curves in Fig.~\ref{fig:examples} are obtained deterministically from the trace moments of $\hat{X}$ via Eq.~(\ref{eq:results_moments_closed}) (Methods). In addition, for the CZ example, we overlay finite-sample protocol-based point estimates obtained by the Monte Carlo simulation described above (here, $M=500$ random inputs and $N=1000$ shots per input).

(i) {\em Two-qubit CZ-like coherent phase rotation.} Here, we take
\begin{eqnarray}
\hat{X}_{\mathrm{CZ}}(\phi_{\epsilon}) = \mathrm{diag}(1,1,1,e^{i\phi_{\epsilon}}),\quad (d=4).
\label{eq:results_CZ_model}
\end{eqnarray}
This simplest coherent model already exhibits the $\sqrt{r}$ gap: $r(\phi_{\epsilon})$ is quadratic in $\phi_{\epsilon}$ whereas $d_{\diamond}(\phi_{\epsilon})=\abs{\sin(\phi_{\epsilon}/2)}$ is linear (Methods). Geometrically, $\mathrm{spec}(\hat{X}_{\mathrm{CZ}})$ consists of two distinct points on the unit circle, so $m(\hat{X}_{\mathrm{CZ}})$ is the distance from the origin to the chord connecting $1$ and $e^{i\phi_{\epsilon}}$.

(ii) {\em Decomposed Toffoli under uniform coherent primitive over-rotations.} Let qubits $1$, $2$ be controls and qubit $3$ be the target. We use the standard Clifford+$T$ decomposition~\cite{Selinger2013}
\begin{eqnarray}
\hat{U}_{\mathrm{Tof}} = \hat{\mathrm{CNOT}}_{1\rightarrow 2} \hat{T}_{2}^{\dagger} \hat{T}_{1} \hat{H}_{3} \hat{\mathrm{CNOT}}_{1\rightarrow 2} \hat{T}_{3} \hat{T}_{2} 
\hat{\mathrm{CNOT}}_{1\rightarrow 3} \hat{T}_{3}^{\dagger} \hat{\mathrm{CNOT}}_{2\rightarrow 3} \hat{T}_{3} \hat{\mathrm{CNOT}}_{1\rightarrow 3} \hat{T}_{3}^{\dagger} \hat{\mathrm{CNOT}}_{2\rightarrow 3} \hat{H}_{3},
\label{eq:results_toffoli_decomposition}
\end{eqnarray}
where $\hat{H}$ is the Hadamard gate and $\hat{T}=\mathrm{diag}(1,e^{i\pi/4})$. We assume that every primitive in Eq.~(\ref{eq:results_toffoli_decomposition}) is implemented with the same coherent over-rotation parameter $\epsilon$: the single-qubit gates are 
\begin{eqnarray}
\hat{T}^{(\epsilon)} := \exp\left(-i \epsilon \frac{\hat{\sigma}_{z}}{2}\right)\hat{T},
\quad
\hat{H}^{(\epsilon)} := \exp\left(-i \epsilon \frac{\hat{H}}{2}\right)\hat{H},
\label{eq:results_1q_overrot}
\end{eqnarray}
and CNOT gate is
\begin{eqnarray}
\hat{\mathrm{CNOT}}_{c \rightarrow t}^{(\epsilon)} := \exp\left(-i\epsilon\hat{\Pi}_{1}^{(c)}\otimes \hat{\sigma}_{x}^{(t)}\right)\hat{\mathrm{CNOT}}_{c\rightarrow t}.
\label{eq:results_CNOT_overrot}
\end{eqnarray}
where $\hat{\Pi}_{1} := \ket{1}\bra{1} = \frac{1}{2}\bigl(\hat{\mathds{1}}-\hat{\sigma}_{z}\bigr)$. The implemented Toffoli is obtained by replacing every primitive with its $(\epsilon)$-version, and the effective error unitary is
\begin{eqnarray}
\hat{X}_{\mathrm{Tof}}(\epsilon)=\hat{U}_{\mathrm{Tof}}^{\dagger}\hat{U}_{\mathrm{Tof}}^{(\epsilon)}\in U(8).
\end{eqnarray}
We compute $\mathrm{spec}(\hat{X}_{\mathrm{Tof}}(\epsilon))$ and apply Eq.~(\ref{eq:results_exact_diamond_convex_hull}) to obtain the true diamond distance, while $(F(\epsilon),D(\epsilon))$ follow from the trace moments $\tr{(\hat{X}_{\mathrm{Tof}}(\epsilon))}$ and $\tr{(\hat{X}_{\mathrm{Tof}}(\epsilon)^2)}$ (Methods). This example illustrates a key practical point: compiled multi-qubit gates can inherit coherent miscalibrations from many primitives, and the dimension-amplified $u=1$ saturation makes unitarity-based worst-case bounds increasingly conservative as $d$ grows.

(iii) {\em $n$-qubit QFT circuit under uniform coherent over-rotations.}
We consider the standard $n$-qubit QFT decomposition into Hadamard and controlled-phase (CP) gates (omitting the final bit-reversal swaps). Let $d=2^n$ and define
\begin{eqnarray}
\hat{U}_{\mathrm{QFT},n} = \hat{G}_L \cdots \hat{G}_2 \hat{G}_1,
\label{eq:results_qft_gate_list}
\end{eqnarray}
with gate sequence
\begin{eqnarray}
(\hat{G}_1,\ldots,\hat{G}_L) = \left( \hat{H}_1, \hat{CP}_{2,1}(\theta_{2,1}), \ldots, \hat{CP}_{n,1}(\theta_{n,1}), \hat{H}_2, \hat{CP}_{3,2}(\theta_{3,2}), \ldots, \hat{H}_n \right),
\label{eq:results_qft_gate_sequence}
\end{eqnarray}
where $\theta_{k,j}=\pi/2^{k-j}$ for $1\le j<k\le n$ and
\begin{eqnarray}
\hat{CP}_{k,j}(\theta) = \hat{\mathds{1}}+(e^{i\theta}-1)\ketbra{11}{11}_{k,j}.
\label{eq:results_qft_cp_def}
\end{eqnarray}
We impose a uniform coherent miscalibration parameter $\epsilon$ on every primitive: the Hadamard gates follow Eq.~(\ref{eq:results_1q_overrot}) and CP gates are
\begin{eqnarray}
\hat{CP}_{k,j}^{(\epsilon)}(\theta) = \hat{CP}_{k,j}((1+\epsilon)\theta),
\label{eq:results_qft_overrot_model}
\end{eqnarray}
so that the implemented QFT is $\hat{U}_{\mathrm{QFT},n}^{(\epsilon)}=\hat{G}_L^{(\epsilon)}\cdots\hat{G}_1^{(\epsilon)}$ and the effective error unitary is $\hat{X}_{n}(\epsilon)=\hat{U}_{\mathrm{QFT},n}^{\dagger}\hat{U}_{\mathrm{QFT},n}^{(\epsilon)}\in U(d)$. We evaluate up to $n=10$ ($d=1024$) by explicit matrix products, compute the eigenvalues of $\hat{X}_n(\epsilon)$, and obtain the exact $d_{\diamond}$ from Eq.~(\ref{eq:results_exact_diamond_convex_hull}). In this large-$d$ coherent setting, the saturated $(r,u)$ bound becomes extremely loose due to its explicit dimension factor, while the fluctuation information in $D$ continues to constrain the spectral polygon and yields a substantially sharper certificate.

Figure~\ref{fig:examples} shows that the tightening provided by $D$ is not specific to a single gate: it persists from a two-qubit entangling primitive to a multi-gate Toffoli construction and to a large-scale ($10$-qubit) QFT circuit module. Across all three coherent examples, incorporating the fidelity deviation $D$ yields a substantially tighter worst-case certificate in the coherent regime, while the unitarity-assisted bound becomes increasingly loose with system size because unitarity necessarily saturates for unitary errors. Further worked details and additional finite-sample simulations of the estimation protocol are provided in Sec.~VI of the Supplementary Information.

\medskip
\paragraph*{\bf Regime-adaptive certification by combining $(r,u)$ and $(F,D)$.}

The unitarity $u$ is a powerful regime witness when incoherent components are present: in decoherence-dominated regimes it enables linear-in-infidelity certificates~\cite{Wallman2015,Kueng2016}. However, once the noise becomes predominantly coherent, the unitarity saturates (it equals one for perfectly unitary errors) and loses the resolution within the coherent manifold. In that regime, the fluctuation information encoded in $D$ remains nontrivial and restores high-resolution worst-case certification.

This complementarity suggests a practical workflow that is useful for experimental calibration and gate qualification. One first measures $(F, D)$ from randomized fidelity data, and (when applicable) also extracts $(r, u)$ from standard RB-type characterizations. If $u$ is significantly below one, the noise has a meaningful incoherent component and the $(r,u)$-based bound can certify a favorable $O(r)$ worst-case scaling; if $u$ is close to one and the error appears coherence-dominated, then $(F,D)$ can resolve how ``valley-like'' the coherent landscape is and yield a much tighter coherent certificate. Formally, this idea is captured by the following hybrid min-bound.

\begin{theorem}[Hybrid regime-adaptive certification]
\label{thm:hybrid}
Let $\mathcal{E}$ be the interaction-picture error channel on $\mathbb{C}^d$. Let $r := 1-F$ be the average infidelity and let $u$ be the unitarity. Assume (i) $\mathcal{E}$ is unital and has no leakage so that Kueng~et~al.'s $(r, u)$ bound applies, and (ii) the coherent model hypotheses needed for {\bf Theorem~\ref{thm:main}} apply (for example, $\mathcal{E}(\hat{\rho})=\hat{X}\hat{\rho}\hat{X}^{\dagger}$ for a unitary $\hat{X}$). Define two computable upper bounds,
\begin{eqnarray}
B_{ru}(r, u) &:=& d^2 c_d\sqrt{u+\frac{2dr}{d-1}-1}, \nonumber \\
B_{FD}(F,D) &:=& \sqrt{1-c(F,D)^2},
\end{eqnarray}
where $c_d := \frac{1}{2}\sqrt{1-\frac{1}{d^2}}$ and $c(F,D)$ is given in Eq.~(\ref{eq:thm_c_def}).
Then, the diamond distance obeys the hybrid bound
\begin{eqnarray}
d_{\diamond}(\mathcal{E},\mathcal{I}) \le B_{\mathrm{hyb}}(F,D,r,u) := \min\{B_{ru}(r,u),\,B_{FD}(F,D)\}.
\label{eq:disc_hybrid}
\end{eqnarray}
\end{theorem}

{\bf Theorem~\ref{thm:hybrid}} formalizes the regime-adaptive strategy: use $u$ as a ``coherence thermometer'' to exploit favorable linear scaling when incoherence is present, and use $(F,D)$ to regain tight worst-case control when unitarity saturates. In practice, tracking both $u$ and $D$ can also help diagnose model mismatch: when the noise is not perfectly unitary, these quantities provide complementary diagnostics for how and why the average fidelity may fail to reflect worst-case risk.

\section*{Discussion and outlook}\label{sec:discussion}

This work proposes a practical sharpening of fault-tolerance assessment for quantum computing by elevating the average fidelity from a single mean value to a landscape characterization. The fidelity deviation $D$ is the standard deviation of the state-dependent fidelities underlying the average fidelity, and it captures the state-to-state non-uniformity that coherent control errors naturally generate. From a physical perspective, a coherent error is a unitary mis-rotation: it cannot affect all input states equally, and the resulting ``valleys'' in the fidelity landscape are precisely what can dominate adversarial circuit behavior. Our main technical result makes this intuition quantitative: for unitary errors on two or more qubits, $(F, D)$ fixes the spectral moment invariants sufficient to compute an explicit certified overlap $c(F, D)$, which in turn yields a tight upper bound on the diamond distance without full process tomography.

Beyond improving bounds in theory, the main significance of fidelity deviation is operational. Once an experiment can sample fidelities to estimate $F$, the second-moment information needed for $D$ can be harvested from the same data stream with a simple shot-noise correction (Methods). This provides an economical way to upgrade coherent worst-case certification using only randomized input-state preparation and projective measurements, and it remains informative precisely in the coherence-dominated regime where the unitarity saturates and fidelity-only conversions exhibit the $\sqrt{r}$ catastrophe. From the viewpoint of experimental reporting, this suggests a simple practice: together with the usual average fidelity (or infidelity), report the deviation $D$ and, when available, the unitarity $u$; the resulting data can be translated into a conservative and interpretable circuit-level certificate via Eqs.~(\ref{eq:thm_diamond_bound}) and (\ref{eq:disc_hybrid}).

It is also worth noting that the proposed deviation observable is compatible with existing benchmarking pipelines. Many experimental platforms already estimate average fidelities (or infidelities) using randomized protocols. Whenever such a protocol produces a collection of fidelity estimates across randomized settings---rather than only a single aggregated fit parameter---the same dataset can be post-processed to obtain $D$ using the unbiased second-moment estimator described here. Thus, extracting $D$ often does not require new quantum resources so much as retaining and analyzing information that is typically discarded when reporting only a mean. At the same time, because $D$ is defined from the fidelities, practical implementations must treat SPAM carefully; one can either embed the $(F, D)$ estimation in SPAM-robust randomized schemes or regard $(F, D)$ as a complementary diagnostic used alongside RB when calibrating coherent control errors.

The fluctuation viewpoint also suggests a calibration lesson. Two gates with the same average fidelity can have very different fidelity landscapes: one may be uniformly good, while the other may contain narrow low-fidelity directions that are invisible to the mean yet dangerous for worst-case circuit execution. The quantity $D$ detects this difference directly. In principle, this enables a refined control objective: beyond maximizing $F$, one can aim to suppress $D$ to ``flatten'' the fidelity landscape and reduce the risk of adversarial coherent accumulation.

In summary, our methodology provides a direct and economical route to the evaluation, since it produces a worst-case certificate from experimentally accessible data without full tomography, and it remains accurate in the coherent regime that is often dangerous for FTQC. We therefore suggest that reporting $(F, D)$---together with $u$ when available---can serve as a default and standard way to benchmark the engineered quantum gates and circuits from a fault-tolerance viewpoint.


\section*{Methods}

\medskip
\paragraph*{\bf Closed-form Haar moments for unitary errors.}

In the coherent setting $\mathcal{X}(\hat{\rho})=\hat{X}\hat{\rho}\hat{X}^{\dagger}$, the state-dependent single fidelity is $f_{\mathcal{X}}(\psi)=\abs{\bra{\psi}\hat{X}\ket{\psi}}^2$. The average fidelity $F$ and the second Haar moment $E_2 := \int d\psi\, f_{\mathcal{X}}(\psi)^2$ (hence $D^2=E_2-F^2$) can be evaluated in closed form via second and fourth Haar moments~\cite{Collins2006,Mele2024,Cho2026}. A convenient identity is the symmetric-projector formula
\begin{eqnarray}
\int d\psi\, \bigl(\ketbra{\psi}{\psi}\bigr)^{\otimes k} = \frac{\hat{\Pi}^{(k)}_{\rm sym}}{\tr{\hat{\Pi}^{(k)}_{\rm sym}}},
\label{eq:methods_sym_proj}
\end{eqnarray}
where $\hat{\Pi}^{(k)}_{\rm sym}$ is the projector onto the fully symmetric subspace of $(\mathbb{C}^d)^{\otimes k}$. The normalization is explicit: $\tr{\hat{\Pi}^{(k)}_{\rm sym}}=\binom{d+k-1}{k}$, so for $k=2$ and $k=4$, we have $\tr{\hat{\Pi}^{(2)}_{\rm sym}}=d(d+1)/2$ and $\tr{\hat{\Pi}^{(4)}_{\rm sym}}=d(d+1)(d+2)(d+3)/24$.

Using the trace identity
\begin{eqnarray}
\abs{\bra{\psi}\hat{X}\ket{\psi}}^2 = \tr{(\hat{X}\otimes \hat{X}^{\dagger}) \bigl(\ketbra{\psi}{\psi}\bigr)^{\otimes 2}},
\end{eqnarray}
Eq.~(\ref{eq:methods_sym_proj}) for $k=2$ yields the well-known formula
\begin{eqnarray}
F = \frac{d+\abs{\tr{(\hat{X})}}^2}{d(d+1)}.
\label{eq:methods_F_closed}
\end{eqnarray}
Thus, in the coherent unitary setting, the average fidelity depends only on the single spectral invariant $P:=\abs{\tr{(\hat{X})}}$, which is a global measure of how strongly the eigenphases of $\hat{X}$ cluster on the unit circle.

The second moment $E_2$ requires fourth-order Haar contractions. Evaluating Eq.~(\ref{eq:methods_sym_proj}) at $k=4$ produces two independent spectral combinations, which we package as $P=\abs{\tr{(\hat{X})}}$ and
\begin{eqnarray}
Q:=\abs{ \tr{(\hat{X}^2)} + \tr{(\hat{X})}^2 }.
\end{eqnarray}
In terms of these invariants one obtains
\begin{eqnarray}
E_2 = D^2+F^2 = \frac{2d(d+3)+4(d+2)P^2+Q^2}{d(d+1)(d+2)(d+3)}.
\label{eq:methods_E2_compact}
\end{eqnarray}
Thus, $(F, D)$ determine $(P, Q)$ for $d\ge 4$. Intuitively, $P$ fixes a second-moment spectral constraint (the average overlap), while $Q$ injects a fourth-moment constraint that controls the coherent spectral spread. This is the technical reason that $D$ remains informative even when the unitarity saturates at $u=1$. For more details of the fourth-moment evaluation, see Sec.~III of the Supplementary Information.

\medskip
\paragraph*{\bf Direct estimation protocol and shot-noise correction.}

Let $\ket{\psi_i}=\hat{V}_i\ket{0}$ be random input states obtained by sampling $\hat{V}_i$ from a unitary $4$-design~\cite{Dankert2009,Harrow2009,Brandao2016}. (In practice, an approximate design produced by sufficiently random Clifford or local random circuits can be used; the only requirement is that the resulting ensemble reproduces Haar moments up to order four to the desired accuracy.) For each $i$, estimate the single fidelity $f_i:=f_{\mathcal{E}}(\psi_i)$ by repeating the projective test $N$ times and recording the number of ``pass'' outcomes $K_i \sim {\rm Bin}(N,f_i)$. Define the per-state frequency $\hat{f}_i=K_i/N$, so $\mathbb{E}[\hat{f}_i|f_i]=f_i$~\cite{Hoeffding1992}. Averaging over $M$ random inputs then estimates the Haar mean $F$.

To estimate $D$, one must also estimate the second Haar moment $E_2=\mathbb{E}_{\psi}[f_{\mathcal{E}}(\psi)^2]$. A naive approach would average $\hat{f}_i^2$, but this is biased upward by shot noise:
\begin{eqnarray}
\mathbb{E}[\hat{f}_i^2|f_i]=f_i^2+\frac{f_i(1-f_i)}{N}.
\label{eq:methods_bias_naive_second_moment}
\end{eqnarray}
This bias is small when $N$ is very large, but in the moderate-shot regime it would systematically overestimate $E_2$ and hence distort $D$. We therefore use a simple factorial-moment correction. Using the identity $\mathbb{E}[K_i(K_i-1)|f_i]=N(N-1)f_i^2$, define
\begin{eqnarray}
\widehat{f_i^2} := \frac{K_i(K_i-1)}{N(N-1)} = \frac{N\hat{f}_i^2-\hat{f}_i}{N-1},
\quad
\mathbb{E}[\widehat{f_i^2}|f_i]=f_i^2.
\end{eqnarray}
Then, the unbiased estimators of the first and second Haar moments are
\begin{eqnarray}
\hat{F} := \frac{1}{M}\sum_{i=1}^{M}\hat{f}_i,
\quad
\widehat{E_2} := \frac{1}{M}\sum_{i=1}^{M}\widehat{f_i^2}.
\label{eq:methods_estimators_FE2}
\end{eqnarray}

Finally, since $\widehat{E_2}-\hat{F}^2$ is still biased (because $\hat{F}^2$ is biased as an estimator of $F^2$), we use an unbiased cross-average:
\begin{eqnarray}
\widehat{F^2} := \frac{1}{M(M-1)}\sum_{i\neq j}\hat{f}_i\hat{f}_j,
\quad
\mathbb{E}[\widehat{F^2}]=F^2.
\end{eqnarray}
The unbiased estimator of the deviation is then
\begin{eqnarray}
\widehat{D^2}:=\widehat{E_2}-\widehat{F^2},
\quad
\hat{D}:=\sqrt{\max\{\widehat{D^2},0\}}.
\end{eqnarray}
The truncation at zero is only to handle rare finite-sample fluctuations that make $\widehat{D^2}$ slightly negative; it does not affect the large-sample limit.

The total experimental cost is set by two knobs: $M$ (how many random inputs are sampled) and $N$ (how many shots are spent per input). Increasing $M$ reduces the sampling noise associated with drawing $\psi$ from the design ensemble, while increasing $N$ reduces the binomial shot noise for each sampled state. Because $D$ depends on a second moment, it is typically beneficial to allocate enough shots so that $N\ge 2$ (to enable the factorial correction) and to increase $M$ until the empirical distribution of $f_{\mathcal{E}}(\psi)$ is well resolved. Importantly, however, no additional circuits or measurement settings are required beyond those already used to estimate $F$.

\medskip
\paragraph*{\bf Numerical evaluation for Fig.~\ref{fig:examples}.}

The ``theory'' curves in Fig.~\ref{fig:examples}(c) are obtained deterministically from the effective error unitary $\hat{X}$, rather than by Monte Carlo integration over $\ket{\psi}$. Concretely, for each error parameter value (e.g., $\phi_{\epsilon}$ or $\epsilon$) we construct the implemented unitary $\hat{U}_{\rm exp}$ by replacing each primitive gate in the chosen decomposition with its coherently miscalibrated version and multiplying the resulting matrices in the computational basis. We then form $\hat{X}=\hat{U}_{\rm ideal}^{\dagger}\hat{U}_{\rm exp}$ and compute $\tr{(\hat{X})}$ and $\tr{(\hat{X}^2)}$, which determine $(F,D)$ via Eqs.~(\ref{eq:methods_F_closed})--(\ref{eq:methods_E2_compact}) (equivalently, Eq.~(\ref{eq:results_moments_closed})).

To compute the ``true'' diamond distance shown in Fig.~\ref{fig:examples}, we diagonalize $\hat{X}$ to obtain its eigenvalues $\{\lambda_j\}_{j=1}^{d}$, take their convex hull in the complex plane, and evaluate the minimum distance from the origin to this polygon to obtain $m(\hat{X})$, hence $d_{\diamond}=\sqrt{1-m(\hat{X})^2}$ via Eq.~(\ref{eq:results_exact_diamond_convex_hull}). For the CZ panel, we additionally overlay finite-sample, protocol-based point estimates obtained by simulating the direct $(F,D)$ estimation procedure described above (see Sec.~VI of the Supplementary Information).

\medskip
\paragraph*{\bf Two-qubit coherent phase rotation example (closed forms).}

For the CZ-like coherent phase miscalibration model, i.e.,
\begin{eqnarray}
\hat{X}(\phi_{\epsilon})={\rm diag}(1,1,1,e^{i\phi_{\epsilon}})\qquad (d=4),
\end{eqnarray}
Eq.~(\ref{eq:methods_F_closed}) yields a simple closed form. Here, $\tr{(\hat{X})}=3+e^{i\phi_{\epsilon}}$ so that $P^2=\abs{3+e^{i\phi_{\epsilon}}}^2=10+6\cos(\phi_{\epsilon})$, and therefore
\begin{eqnarray}
F(\phi_{\epsilon}) = 1-\frac{3}{5}\sin^2\left(\frac{\phi_{\epsilon}}{2}\right),
\quad
r(\phi_{\epsilon}) = \frac{3}{5}\sin^2\left(\frac{\phi_{\epsilon}}{2}\right).
\end{eqnarray}
A direct fourth-moment evaluation (equivalently, Eq.~(\ref{eq:methods_E2_compact}) with $d=4$) gives
\begin{eqnarray}
D(\phi_{\epsilon}) = \frac{1}{5}\sqrt{\frac{17}{7}}\ \sin^2\left(\frac{\phi_{\epsilon}}{2}\right),
\end{eqnarray}
so in the small-error regime both $r$ and $D$ scale as $\Theta(\phi_{\epsilon}^2)$.

For the worst-case metric, the numerical range of $\hat{X}$ is the chord between $1$ and $e^{i\phi_{\epsilon}}$ (the repeated eigenvalues at $1$ do not change the convex hull), so the distance from the origin to this chord is $\cos(\phi_{\epsilon}/2)$ and thus $m(\hat{X})=\cos(\phi_{\epsilon}/2)$. Using {\bf Lemma~\ref{lem:unitary_diamond}}, we obtain the exact diamond distance
\begin{eqnarray}
d_{\diamond}(\mathcal{X},\mathcal{I})= \left| \sin\left(\frac{\phi_{\epsilon}}{2}\right) \right|.
\end{eqnarray}
This explicitly illustrates the coherent $\sqrt{r}$ phenomenon: $r\simeq (3/20)\phi_{\epsilon}^2$ while $d_{\diamond} \simeq \abs{\phi_{\epsilon}}/2$.

Finally, this example provides a concrete check of tightness: inserting the closed-form $(F(\phi_{\epsilon}),D(\phi_{\epsilon}))$ into {\bf Theorem~\ref{thm:main}} reproduces $c(F,D)=\cos(\phi_{\epsilon}/2)$, and hence $\sqrt{1-c(F,D)^2}=\abs{\sin(\phi_{\epsilon}/2)}$, i.e., the bound is saturated by the model itself. In other words, the moment constraints encoded in $(F,D)$ are sufficient to reconstruct the relevant spectral geometry of $\hat{X}$ exactly in this coherent two-qubit setting.

\section*{Declarations}

\medskip\noindent
{\bf Data Availability.}---The data that support the findings of this study are available from the corresponding author upon request.

\medskip\noindent
{\bf Code Availability.}---The code that supports the findings of this study is available from the corresponding author upon request.

\medskip\noindent
{\bf Author contributions.}---JB and YH conceived the original idea of the work. IS performed the initial numerical simulations to further develop and validate the idea, and the project was subsequently advanced. JB carried out the theoretical proofs together with KC and YH. The manuscript was written by KC and JB, and all authors jointly reviewed and refined the structure, logic, and presentation of the paper. Correspondence and requests for materials should be addressed to JB and YH.

\medskip\noindent
{\bf Competing interests.}---The authors declare no competing interests.

\section*{Acknowledgements}
This work was supported by the Ministry of Science, ICT and Future Planning (MSIP) through the National Research Foundation of Korea (RS-2024-00432214, RS-2025-03532992, and RS-2025-18362970) and the Institute of Information and Communications Technology Planning and Evaluation grant funded by the Korean government (RS-2019-II190003, ``Research and Development of Core Technologies for Programming, Running, Implementing and Validating of Fault-Tolerant Quantum Computing System''), the Korean ARPA-H Project through the Korea Health Industry Development Institute (KHIDI), funded by the Ministry of Health \& Welfare, Republic of Korea (RS-2025-25456722), and the Ministry of Trade, Industry and Resources (MOTIR), Korea, under the project ``Industrial Technology Infrastructure Program'' (RS-2024-00466693). This work was also supported by Korea Institute of Science and Technology Information (KISTI) ((KISTI)K26L1M3C5-01). We acknowledge the Yonsei University Quantum Computing Project Group for providing support and access to the Quantum System One (Eagle Processor), which is operated at Yonsei University.



%


\clearpage
\newpage
\onecolumngrid

\makeatletter
\let\addcontentsline\origaddcontentsline
\makeatother

\part{Supplementary Information}

This supplementary information is written as a self-contained and complete companion manuscript, so that the readers can fully understand the calculations, logical flow, and results of our main research on ``Sharpen Worst-Case Error Assessment for Fault-Tolerant Quantum Computing: Fidelity and Its Deviation.''

\tableofcontents

\section{Introduction}\label{Sec:Intro}

The characterization of gate performance is a central task in quantum computer development, and particularly essential for fault-tolerant quantum computation (FTQC)~\cite{fowler2004constructing,aharonov2008,Gilchrist2005,aliferis2006,Kueng2016}.  In practice, a widely adopted single-number figure of merit is the average gate fidelity, or simply, gate fidelity~\cite{nielsen2002simple}, because it can be estimated efficiently and robustly against the state-preparation-and-measurement (SPAM) errors by scalable protocols, for example, randomized benchmarking (RB)~\cite{PhysRevA.77.012307,magesan2011scalable}. Concretely, given an ideal unitary gate $\hat{U}$ and an implemented CPTP map $\mathcal{G}$ on $\mathbb{C}^d$ ($d=2^n$ for an $n$-qubit gate), it is convenient to work in the interaction picture and define the error channel $\mathcal{E}=\mathcal{U}^{\dagger}\circ\mathcal{G}$ so that $\mathcal{E}=\mathcal{I}$ corresponds to perfect implementation. Here, the state-dependent survival probability is $f_{\mathcal{E}}(\psi) = \bra{\psi}\mathcal{E}(\ketbra{\psi}{\psi})\ket{\psi}$, and the gate fidelity (and gate infidelity) is given by $F = F_{\mathrm{avg}}(\mathcal{E}) = \int d\psi f_{\mathcal{E}}(\psi)$ ($r = 1-F$). Because $F$ (equivalently $r$) is experimentally attractive and has become a standard reporting metric, it is often used as a proxy for the gate quality when comparing the devices and calibrating control.

FTQC, however, ultimately poses a worst-case question: when a noisy gate is embedded deep inside a largely-scaled circuit and is used repeatedly, how badly can it behave in the most adversarial scenario?  This viewpoint is fundamentally different from the ``average-case'' metrics, because a small subset of input directions can dominate the failure probability of an error-correcting gadget.  It is known that the standard circuit-relevant metric capturing this worst-case behavior is the diamond distance~\cite{kitaev1997,Gilchrist2005}. For any two channels $\mathcal{E}$, $\mathcal{F}$ acting on $\mathbb{C}^d$, the normalized diamond distance is defined as
\begin{eqnarray}
d_{\diamond}(\mathcal{E},\mathcal{F}) := \frac{1}{2}\|\mathcal{E}-\mathcal{F}\|_{\diamond} = \frac{1}{2}\max_{\hat{\rho}} \norm{(\mathcal{E}\otimes\mathcal{I})(\hat{\rho})-(\mathcal{F}\otimes\mathcal{I})(\hat{\rho})}_{1},
\label{eq:intro_diamond_def}
\end{eqnarray}
where the maximization is over the density operators $\hat{\rho}$ on system plus an ancilla of dimension $d$. Operationally, $d_{\diamond}$ is the optimal single-shot distinguishability between the two channels under arbitrary entangled inputs and measurements, and it is stable under adding the idle ancillas and subadditive under composition, which is why it naturally enters threshold theorems (see Ref.~\cite{Kueng2016}). However, the practical difficulty is that $d_{\diamond}$ is not directly accessible without essentially tomographic information or costly optimization, so one seeks sharp certificates of $d_{\diamond}$ from experimentally accessible data, such as, $F$.

A central obstacle is that $r=1-F$ can drastically underestimate the worst-case risk in low-error regime. In general, one has the conversion bounds of the form $d_{\diamond}(\mathcal{E},\mathcal{I}) \le \sqrt{d(d+1)r}$, implying that even when $r \ll 1$, the diamond distance can scale as $\Theta(\sqrt{r})$.  This square-root amplification is not merely a loose artifact; it is structurally saturated by coherent (systematic/unitary) miscalibrations, where the gate fidelity is second order in a small coherent angle while the worst-case distinguishability is first order.  A major refinement of this conversion problem was proposed by Kueng {\em et al.}, who introduced the \emph{unitarity} $u$ and derived $(r,u)$-based worst-case bounds for unital noise without leakage~\cite{Kueng2016}. This refinement is powerful precisely when the noise has a significant incoherent component, as the deviations of $u$ from unity certify the favorable linear scaling $d_{\diamond}=O(r)$.  However, in the coherent-dominated regime that is most dangerous for FTQC, unitarity inevitably saturates to $u \simeq 1$ and becomes nearly constant. In particular, for purely unitary errors $u=1$ exactly, so the $(r, u)$ bounds collapse and can even become provably looser than the best-known fidelity-only conversion by a dimension factor. Thus, while $u$ is an excellent regime witness for incoherence, it does not resolve fine structure within the fully coherent manifold where the $\sqrt{r}$ phenomenon is most pronounced.

In this work, we develop a complementary strategy: instead of asking only ``how coherent is the noise?'' (which unitarity answers), we ask ``how non-uniform is the coherent action across input states?''  We formalize this idea by suggesting a notion of, dubbed as, fidelity deviation $D$, defined as the standard deviation of the same experimentally measurable $f_{\mathcal{E}}(\psi)$: namely,
\begin{eqnarray}
D = D(\mathcal{E}) := \sqrt{\int d\psi f_{\mathcal{E}}(\psi)^2 - F^2}.
\end{eqnarray}
While $F$ captures only the mean of the fidelity landscape, $D$ captures its fluctuation and therefore injects higher-moment information that remains nontrivial even when $u=1$. Our main technical result shows that in the coherent (i.e., unitary/systematic) error setting $\mathcal{E}(\hat{\rho})=\hat{X}\hat{\rho}\hat{X}^{\dagger}$, the pair $(F, D)$ determines the relevant spectral moment invariants and yields an explicit, computable, and tight lower bound on the minimum overlap $m(\hat{X})=\min_{\psi}|\langle\psi|\hat{X}|\psi\rangle|$.  Consequently, we can achieve a remarkably
sharp $(F, D)$-assisted worst-case certificate of the form
\begin{eqnarray}
d_{\diamond}(\mathcal{E}, \mathcal{I}) \le \sqrt{1-c(F,D)^2},
\end{eqnarray}
where $c(F,D)$ is an explicit function computed from the observed ($F$, $D$). Importantly, $D$ is experimentally accessible without full tomography: we propose a direct $t$-design sampling protocol in which one measures the same survival probabilities used to estimate $F$ and additionally tracks their second moment to estimate $D$. We then demonstrate the practical impact of the resulting certification by worked examples, including a CZ-like coherent phase error (two-qubit, $d=4$), a decomposed Toffoli gate under coherent primitive over-rotations (three-qubit, $d=8$), and an $n$-qubit quantum Fourier transform (up to $n=10$, $d=2^n$), where the $(F, D)$-assisted bound tracks the true diamond distance far more closely than the bounds derived from $F$ alone in the coherent regime. Finally, because $(r, u)$ and $(F, D)$ are informative in complementary noise regimes, we advocate a regime-adaptive hybrid strategy: use $(r, u)$ to exploit linear scaling when incoherence is present, and/or use $(F, D)$ to recover sharp resolution when unitarity saturates near $u=1$.

\section{Background: Average fidelity vs.\ diamond distance (fault-tolerance viewpoint)}\label{sec:background}

Fault-tolerant quantum computation (FTQC) ultimately asks a worst-case question: when a noisy gate is embedded deep inside a large-circuit quantum computation and is used repeatedly, how badly can it behave in the most adversarial input scenario? This perspective is fundamentally different from that of many experimentally efficient
characterization tools---most notably randomized benchmarking (RB)---which naturally return the average-case figures of merit. This mismatch is not merely a minor conversion issue: in the low-error regime of our practical interest, the average- and worst-case errors can differ by orders of magnitude, especially when coherent (systematic/unitary) effects dominate.

The goal of this section is to set a precise conceptual and mathematical stage for the main theme of this paper: i.e., the average fidelity alone can severely underestimate worst-case fault-tolerance risk, and one needs additional experimentally accessible information to sharpen the worst-case assessment. We therefore review:
\begin{itemize}
\item[(i)] the gate-fidelity-based error metric and why it is experimentally attractive;
\item[(ii)] the diamond distance as the canonical FTQC-relevant worst-case metric;
\item[(iii)] the best-known general conversion bounds linking the two (highlighting the coherent $\sqrt{r}$ catastrophe); and
\item[(iv)] the unitarity-assisted refinement from prior work and a crucial limitation that becomes acute in the fully coherent regime.
\end{itemize}

\subsection{Gate fidelity and the average error rate}\label{subsec:avg_fid}

Let $\mathcal{H} \cong \mathbb{C}^d$ be the system Hilbert space with $d=2^n$ for an $n$-qubit gate. Let $\mathcal{U}(\hat{\rho}) := \hat{U}\hat{\rho}\hat{U}^\dagger$ denote the ideal unitary gate and $\mathcal{G}$ the experimentally implemented quantum channel (i.e., a completely-positive trace-preserving map). It is convenient to work in the interaction picture and define the error channel such that
\begin{eqnarray}
\mathcal{E} := \mathcal{U}^\dagger \circ \mathcal{G},
\label{eq:def_error_channel}
\end{eqnarray}
where $\mathcal{E} = \mathcal{I}$ corresponds to the perfect implementation. Here, $\mathcal{I}$ is the identity channel.

A standard average-case metric is the gate fidelity~\cite{nielsen2002simple}
\begin{eqnarray}
F_{\rm avg}(\mathcal{E}) := \int d\psi \bra{\psi} \mathcal{E}\bigl(\ketbra{\psi}{\psi}\bigr) \ket{\psi},
\label{eq:def_Favg}
\end{eqnarray}
where the integral is over the pure states $\ket{\psi} \in \mathcal{H}$ with respect to the Haar measure. The corresponding average error rate (gate infidelity) is
\begin{eqnarray}
r(\mathcal{E}) := 1 - F_{\rm avg}(\mathcal{E}).
\label{eq:def_r}
\end{eqnarray}
The practical appeal of $r$ is substantial: it can be estimated efficiently and robustly to SPAM errors using RB-type protocols, and thus $r$ is widely reported as a single-number gate-quality metric~\cite{magesan2011scalable}.

However, Eq.~(\ref{eq:def_Favg}) is an average over the input states for a single use of the gate. FTQC, by contrast, demands the control of the error behavior uniformly over all inputs and in the presence of entanglement with other subsystems. This is the origin of the central conversion problem: how can we relate a measured $r$ to a worst-case noise strength that enters threshold theorems?

\subsection{Diamond distance as a fault-tolerance-relevant worst-case metric}\label{subsec:diamond}

A worst-case error metric for quantum circuits should satisfy, beyond the axioms of a metric, the two operational requirements:
\begin{itemize}
\item {\bf Stability under idle ancillas}: A gate in a circuit acts on a subsystem while many other qubits are present. The error assigned to that gate should not change if we tensor with an arbitrary idle identity on an ancilla.
\item {\bf Chaining under composition}: FTQC is compositional: gates are applied sequentially and in parallel. An error metric should not behave pathologically under composition; ideally, it should accumulate at most additively.
\end{itemize}

The standard measure satisfying these circuit-level desiderata is the \emph{diamond distance}~\cite{Gilchrist2005}. For channels $\mathcal{E}$, $\mathcal{F}$ acting on $\mathcal{H}$, the diamond norm distance is defined as
\begin{eqnarray}
\norm{\mathcal{E}-\mathcal{F}}_{\diamond} := \max_{\hat{\rho}} \norm{(\mathcal{E}\otimes\mathcal{I})(\hat{\rho}) - (\mathcal{F}\otimes\mathcal{I})(\hat{\rho})}_{1},
\label{eq:def_diamond_norm}
\end{eqnarray}
where the maximization is over the density operators $\hat{\rho}$ on $\mathcal{H} \otimes \mathcal{H}$ (it suffices to take an ancilla of the dimension $d$), and $\norm{\cdot}_1$ is the trace norm. The associated (normalized) diamond distance is
\begin{eqnarray}
d_{\diamond}(\mathcal{E}, \mathcal{F}) := \frac{1}{2}\norm{\mathcal{E} - \mathcal{F}}_{\diamond}.
\label{eq:def_diamond_dist}
\end{eqnarray}
This quantity $d_{\diamond}(\mathcal{E},\mathcal{F})$ equals the optimal single-shot bias for distinguishing $\mathcal{E}$ from $\mathcal{F}$ when one is allowed to use an entangled input and arbitrary measurement. Equivalently, it quantifies the maximum possible change of any experiment’s outcome statistics caused by replacing $\mathcal{F}$ with $\mathcal{E}$ in a larger circuit.

Now we address the two key properties (stability and chaining) explicitly, because they explain why the diamond distance is the natural worst-case error metric in threshold theorems.
\begin{proposition}[Stability under idle ancillas]\label{prop:diamond_stability}
For any channels $\mathcal{E}$, $\mathcal{F}$ on $\mathcal{H}$ and any ancillary system $\mathcal{K}$,
\begin{eqnarray}
d_{\diamond}(\mathcal{E}\otimes \mathcal{I}_{\mathcal{K}},\; \mathcal{F}\otimes \mathcal{I}_{\mathcal{K}}) = d_{\diamond}(\mathcal{E},\mathcal{F}).
\label{eq:diamond_stability}
\end{eqnarray}
\end{proposition}

\begin{proof}---By definition, the diamond norm already optimizes over arbitrary input states on a system plus an ancilla. By tensoring both channels with an additional idle identity, we can simply enlarge the optimization domain without changing the attainable maximum; conversely, restricting the ancilla to a subsystem recovers the original value. Therefore, the value is invariant under adding an idle ancilla.
\end{proof}

\begin{proposition}[Chaining / subadditivity under composition]\label{prop:diamond_chaining}
For CPTP maps $\mathcal{E}_{1},\mathcal{E}_{2}$ and $\mathcal{F}_{1},\mathcal{F}_{2}$ on the same space,
\begin{eqnarray}
d_{\diamond}(\mathcal{E}_{2}\circ\mathcal{E}_{1},\; \mathcal{F}_{2}\circ\mathcal{F}_{1}) \le d_{\diamond}(\mathcal{E}_{2},\mathcal{F}_{2}) + d_{\diamond}(\mathcal{E}_{1},\mathcal{F}_{1}).
\label{eq:diamond_chaining}
\end{eqnarray}
Consequently, for an ideal circuit $\mathcal{C}_{\rm id}=\mathcal{U}_{L}\circ\cdots\circ\mathcal{U}_{1}$ and its implementation $\mathcal{C}=\mathcal{E}_{L}\circ\cdots\circ\mathcal{E}_{1}$,
\begin{eqnarray}
d_{\diamond}(\mathcal{C},\mathcal{C}_{\rm id}) \le \sum_{j=1}^{L} d_{\diamond}(\mathcal{E}_{j},\mathcal{U}_{j}).
\label{eq:diamond_circuit_bound}
\end{eqnarray}
\end{proposition}

\begin{proof}---Using the triangle inequality, we can verify that
\begin{eqnarray}
d_\diamond(\mathcal{E}_{2}\circ\mathcal{E}_{1}, \mathcal{F}_{2}\circ\mathcal{F}_{1}) \le d_\diamond(\mathcal{E}_{2}\circ\mathcal{E}_{1}, \mathcal{E}_{2}\circ\mathcal{F}_{1}) + d_\diamond(\mathcal{E}_{2}\circ\mathcal{F}_{1}, \mathcal{F}_{2}\circ\mathcal{F}_{1}).
\end{eqnarray}
For CPTP post-processing $\mathcal{E}_{2}$, the diamond distance is contractive: $d_\diamond(\mathcal{E}_{2}\circ\mathcal{E}_{1}, \mathcal{E}_{2}\circ\mathcal{F}_{1})
\le d_\diamond(\mathcal{E}_{1},\mathcal{F}_{1})$, and similarly for CPTP pre-processing by $\mathcal{F}_{1}$. By combining these, we can yield Eq.~(\ref{eq:diamond_chaining}). By iterating the same argument over $L$ layers, Eq.~(\ref{eq:diamond_circuit_bound}) is obtained.
\end{proof}

Because of {\bf Propositions~\ref{prop:diamond_stability}} and {\bf \ref{prop:diamond_chaining}}, $d_\diamond(\mathcal{E},\mathcal{I})$ is the canonical per-location worst-case error strength used in FTQC. The main difficulty is experimental: $d_\diamond$ is not directly measurable without essentially tomographic information. Hence we seek sharp bounds or proxies in terms of experimentally accessible quantities such as $r$.

\subsection{From average infidelity to diamond distance: the coherent $\sqrt{r}$ catastrophe}\label{subsec:rb_diamond_bounds}

The threshold theorems are phrased in the worst-case terms (diamond distance or closely related norms), whereas standard benchmarking yields an average-case quantity. A crucial point is that this mismatch is not merely an inessential constant factor: in the regime $r \ll 1$, $d_\diamond$ can scale like $\sqrt{r}$, making the worst-case risk far larger than what the average infidelity suggests~\cite{sanders2015bounding}. To argue this point, we provide the following theorem:
\begin{proposition}[Dimension-dependent conversion bounds]
\label{prop:rb_bounds_proof}
For any CPTP map $\mathcal{E}$ on $\mathbb{C}^d$, the following holds:
\begin{eqnarray}
\frac{d+1}{d}\,r(\mathcal{E}) \le d_{\diamond}(\mathcal{E},\mathcal{I}) \le \sqrt{d(d+1)\,r(\mathcal{E})}.
\label{eq:rb_best_known_bounds}
\end{eqnarray}
\end{proposition}

\begin{proof}---Firstly, we introduce a useful intermediate, i.e., entanglement fidelity. Firstly, let us define the maximally entangled state on $\mathcal{H}\otimes\mathcal{H}$ such that
\begin{eqnarray}
\ket{\Phi} := \frac{1}{\sqrt{d}}\sum_{j=1}^d \ket{j}\otimes\ket{j},
\end{eqnarray}
 and the (normalized) Choi state $\hat{\rho}_{\mathcal{E}} := (\mathcal{E}\otimes\mathcal{I})(\ketbra{\Phi}{\Phi})$. Then, the entanglement fidelity is
\begin{eqnarray}
F_{\rm e}(\mathcal{E}) := \bra{\Phi}\hat{\rho}_{\mathcal{E}}\ket{\Phi} = \bra{\Phi}(\mathcal{E}\otimes\mathcal{I})(\ketbra{\Phi}{\Phi})\ket{\Phi}.
\label{eq:def_ent_fid}
\end{eqnarray}
It is standard that
\begin{eqnarray}
F_{\rm avg}(\mathcal{E}) = \frac{d F_{\rm e}(\mathcal{E}) + 1}{d+1},
\label{eq:Favg_Fe_relation}
\end{eqnarray}
so the entanglement infidelity $r_{\rm e}:=1-F_{\rm e}$ is also related to the average infidelity $r$ by
\begin{eqnarray}
r_{\rm e}(\mathcal{E}) = \frac{d+1}{d} r(\mathcal{E}).
\label{eq:re_vs_r}
\end{eqnarray}

\noindent\medskip {\bf (Lower bound.)} 
We then consider the lower bound. Consider the two-outcome measurement $\{\ketbra{\Phi}{\Phi}, \mathds{1}-\ketbra{\Phi}{\Phi}\}$ on $\mathcal{H}\otimes\mathcal{H}$. Applied to the ideal Choi state $\hat{\rho}_{\mathcal{I}}=\ketbra{\Phi}{\Phi}$, the first outcome occurs with probability $1$, while applied to $\hat{\rho}_{\mathcal{E}}$ it occurs with probability $F_{\rm e}(\mathcal{E})$. Hence, the total-variation distance between these two induced classical distributions is $1 - F_{\rm e}(\mathcal{E})=r_{\rm e}(\mathcal{E})$. Since the trace distance upper bounds the distinguishing advantage over all POVMs, we have
\begin{eqnarray}
r_{\rm e}(\mathcal{E}) \le \frac{1}{2}\norm{\hat{\rho}_{\mathcal{E}} - \hat{\rho}_{\mathcal{I}}}_{1}.
\label{eq:re_le_trace_dist}
\end{eqnarray}
Finally, because the diamond norm maximizes $\left\|(\mathcal{E}-\mathcal{I})\otimes\mathcal{I}(\hat{\rho})\right\|_1$ over all input states $\hat{\rho}$, in particular it is at least the value at $\hat{\rho}=\ketbra{\Phi}{\Phi}$, so
\begin{eqnarray}
d_{\diamond}(\mathcal{E},\mathcal{I}) = \frac{1}{2}\norm{\mathcal{E}-\mathcal{I}}_{\diamond} \ge \frac{1}{2}\norm{\hat{\rho}_{\mathcal{E}} - \hat{\rho}_{\mathcal{I}}}_{1}.
\label{eq:diamond_ge_choi_trace}
\end{eqnarray}
By combining Eqs.~(\ref{eq:re_le_trace_dist})--(\ref{eq:diamond_ge_choi_trace}), we can yield $d_{\diamond}(\mathcal{E},\mathcal{I})\ge r_{\rm e}(\mathcal{E})$. Substituting Eq.~(\ref{eq:re_vs_r}) gives the stated lower bound.

\noindent\medskip  {\bf (Upper bound.)}
A standard inequality relates the diamond norm to the trace distance of Choi states~\footnote{One way to see this is via the SDP characterization of the diamond norm and the fact that the maximally entangled input captures the Choi operator, together with a dimension factor converting between input-state optimizations and the Choi representation.}:
\begin{eqnarray}
\norm{\mathcal{E}-\mathcal{I}}_{\diamond} \le d \norm{\hat{\rho}_{\mathcal{E}} - \hat{\rho}_{\mathcal{I}}}_{1}.
\label{eq:diamond_le_d_choi_trace}
\end{eqnarray}
Then, apply the Fuchs-van de Graaf inequality to states $\hat{\rho}_{\mathcal{E}}$ and $\hat{\rho}_{\mathcal{I}}$~\cite{fuchs2002cryptographic}:
\begin{eqnarray}
\frac{1}{2}\norm{\hat{\rho}_{\mathcal{E}} - \hat{\rho}_{\mathcal{I}}}_{1} \le \sqrt{1 - F(\hat{\rho}_{\mathcal{E}}, \hat{\rho}_{\mathcal{I}})}.
\label{eq:fvG_upper}
\end{eqnarray}
Since $\hat{\rho}_{\mathcal{I}}=\ketbra{\Phi}{\Phi}$ is pure, we have $F(\hat{\rho}_{\mathcal{E}}, \hat{\rho}_{\mathcal{I}})=\bra{\Phi}\hat{\rho}_{\mathcal{E}}\ket{\Phi}=F_{\rm e}(\mathcal{E})$.
Therefore,
\begin{eqnarray}
\norm{\hat{\rho}_{\mathcal{E}} - \hat{\rho}_{\mathcal{I}}}_{1} \le 2\sqrt{1-F_{\rm e}(\mathcal{E})} = 2\sqrt{r_{\rm e}(\mathcal{E})}.
\label{eq:choi_trace_le_sqrt_re}
\end{eqnarray}
By combining Eqs.~(\ref{eq:diamond_le_d_choi_trace})--(\ref{eq:choi_trace_le_sqrt_re}), we can yield
\begin{eqnarray}
d_{\diamond}(\mathcal{E},\mathcal{I}) = \frac{1}{2}\norm{\mathcal{E}-\mathcal{I}}_{\diamond} \le \frac{1}{2} d \cdot 2\sqrt{r_{\rm e}(\mathcal{E})} = d\sqrt{r_{\rm e}(\mathcal{E})}.
\end{eqnarray}
Finally, substitute Eq.~(\ref{eq:re_vs_r}) to obtain $d_\diamond(\mathcal{E},\mathcal{I})\le \sqrt{d(d+1)\,r(\mathcal{E})}$.
\end{proof}

The problem is that Eq.~(\ref{eq:rb_best_known_bounds}) is often summarized as ``$d_\diamond$ can be as large as $O(\sqrt{r})$''. It is important to understand that this is not merely a loose artifact: the square-root penalty is a structural feature of coherent (systematic/unitary) miscalibrations. Intuitively, $r$ is an averaged overlap-squared quantity and thus becomes second order in a small coherent angle, whereas $d_{\diamond}$ is a worst-case distinguishability and becomes first order.

To make this concrete, consider a purely unitary error channel $\mathcal{U}(\hat{\rho}) = \hat{V}\hat{\rho}\hat{V}^{\dagger}$. For such channels, one can show the general inequalities of the form
\begin{eqnarray}
\sqrt{\frac{d+1}{d} r_U} \le d_{\diamond}(\mathcal{U}, \mathcal{I}) \le \sqrt{d (d+1) r_U},
\label{eq:rb_unitary_sqrt_scaling}
\end{eqnarray}
where $r_U$ is the corresponding average infidelity. Thus, every unitary error exhibits worst-case behavior that scales like $\Theta(\sqrt{r_U})$. Here, a particularly transparent illustration is the single-qubit coherent over-rotation $\hat{V}(\delta)=e^{-i \delta \hat{\sigma}_z}$ (calibration error). For the small error, i.e., $\delta \ll 1$, one finds explicitly
\begin{eqnarray}
r(\delta) = \frac{2}{3}\sin^2\delta \simeq \frac{2}{3}\delta^2,
\quad
d_{\diamond}(\mathcal{U}, \mathcal{I}) = \abs{\sin\delta} \simeq \abs{\delta}.
\label{eq:rb_qubit_rotation_example}
\end{eqnarray}
Thus, we can have the canonical coherent relation
\begin{eqnarray}
d_{\diamond}(\mathcal{U}_{V(\delta)},\mathcal{I}) \simeq \sqrt{\frac{3}{2}\,r(\delta)}.
\label{eq:rb_diamond_vs_r_qubit}
\end{eqnarray}
This is the source of the coherent ``catastrophe'': for example, if a benchmarking experiment reports $r\simeq 10^{-4}$, Eq.~(\ref{eq:rb_diamond_vs_r_qubit}) predicts a worst-case error at the percent level~\cite{Kueng2016}.

We can observe a sharp crossover at between the dephasing and calibration error. In particular, we can understand why the coherence matters. Note that the square-root scaling is not universal. For incoherent processes such as depolarizing or pure dephasing, the worst-case and average-case errors often scale linearly. Hence the decisive question is not only whether $r$ is small, but whether the noise is coherence-dominated. A minimal model vividly demonstrating the transition is a single-qubit channel with simultaneous (i) dephasing with probability $p$ and (ii) a coherent over-rotation $\delta$ about $Z$. In the low-error regime, one finds the asymptotic forms~\cite{sanders2015bounding}
\begin{eqnarray}
r_{\mathrm{CD}} \simeq \frac{2}{3}\left(p+\delta^2\right),
\quad
d_{\diamond}^{\mathrm{CD}} \simeq \sqrt{p^2+\delta^2}.
\label{eq:rb_CD_asymptotics}
\end{eqnarray}
These two expressions already encode the qualitative physics:
\begin{itemize}
\item Dephasing-dominated regime ($p\gg \abs{\delta}$). In this case, Eq.~(\ref{eq:rb_CD_asymptotics}) reduces to
\begin{eqnarray}
d_{\diamond}^{\mathrm{CD}} \simeq p \simeq \frac{3}{2}\,r_{\mathrm{CD}},
\label{eq:rb_CD_linear_regime}
\end{eqnarray}
so the worst- and average-case errors are comparable up to a modest constant factor.

\item Calibration-dominated regime ($\abs{\delta}\gg p$). In this case, Eq.~(\ref{eq:rb_CD_asymptotics}) reduces to
\begin{eqnarray}
d_{\diamond}^{\mathrm{CD}} \simeq \abs{\delta} \simeq \sqrt{\frac{3}{2}\,r_{\mathrm{CD}}},
\label{eq:rb_CD_sqrt_regime}
\end{eqnarray}
so the worst-case error becomes parametrically larger than the average infidelity.
\end{itemize}

Eq.~(\ref{eq:rb_best_known_bounds}) is therefore not a mere technical inequality: it encodes a fundamental obstruction to using $r$ alone as a fault-tolerance-relevant figure of merit. In favorable (incoherent) regimes, $r$ can indeed be an accurate proxy for the worst-case behavior. But in coherent-dominated regimes, $r$ systematically underestimates operational risk, and one must supplement $r$ with additional experimentally accessible information that diagnoses \emph{how coherent} the noise is. This motivates incorporating structure parameters such as the unitarity $u$ in the prior study in Ref.~\cite{Kueng2016}, and also motivates
our development of alternative sharp worst-case certification strategy in the present work.

\subsection{Unitarity-assisted refinement : power and a key limitation at $u\simeq 1$}\label{subsec:unitarity_limitation}

To bridge the gap between average-case and worst-case metrics beyond Eq.~(\ref{eq:rb_best_known_bounds}), a prior work Ref.~\cite{Kueng2016} introduced an additional experimentally accessible observable: the \emph{unitarity} $u(\mathcal{E})$, which quantifies how close the noise is to being unitary and can be estimated efficiently using benchmarking-type experiments.

For the unital noise processes without leakage, the pair $(r, u)$ yields a two-sided characterization of the worst-case error via an inequality of the form
\begin{eqnarray}
c_d\sqrt{u(\mathcal{E})+\frac{2d}{d-1}r(\mathcal{E})-1} \le d_{\diamond}(\mathcal{E},\mathcal{I}) \le d^2 c_d\sqrt{u(\mathcal{E})+\frac{2d}{d-1}r(\mathcal{E})-1},
\label{eq:unitarity_bound_general}
\end{eqnarray}
where
\begin{eqnarray}
c_d := \frac{1}{2}\sqrt{1-\frac{1}{d^2}}.
\label{eq:def_cd}
\end{eqnarray}
This refinement is extremely informative when the noise is \emph{not} close to unitary: the condition
\begin{eqnarray}
u(\mathcal{E}) = 1 - \frac{2d}{d-1}r(\mathcal{E}) + O(r^2)
\label{eq:unitarity_linear_condition}
\end{eqnarray}
is necessary and sufficient (within that framework) for the favorable linear scaling
$d_{\diamond}=O(r)$.
In this sense, $u(\mathcal{E})$ acts as a quantitative witness of whether the device has entered the ``good'' incoherent-dominated regime where the average and worst-case metrics are closely aligned.

Nevertheless, there is a crucial limitation directly relevant to our goal: in the coherent-dominated regime where $u(\mathcal{E}) \to 1$, the unitarity ceases to provide additional resolving power. Indeed, for a purely unitary error channel one has $u(\mathcal{E}) = 1$, and Eq.~(\ref{eq:unitarity_bound_general}) becomes
\begin{eqnarray}
d_{\diamond}(\mathcal{E},\mathcal{I}) \le d^2 c_d \sqrt{\frac{2d}{d-1}\,r(\mathcal{E})} = \frac{d}{\sqrt{2}}\sqrt{d(d+1)\,r(\mathcal{E})}.
\label{eq:unitarity_bound_at_u1}
\end{eqnarray}
Comparing Eq.~(\ref{eq:unitarity_bound_at_u1}) with the fidelity-only upper bound in Eq.~(\ref{eq:rb_best_known_bounds}), we obtain the strict multiplicative gap
\begin{eqnarray}
\frac{\left[d^2 c_d \sqrt{\frac{2d}{d-1}\,r}\right]}{\sqrt{d(d+1)\,r}} = \frac{d}{\sqrt{2}}.
\label{eq:ratio_unitarity_vs_r_only}
\end{eqnarray}
Thus, when coherent errors dominate so strongly that $u=1$, the unitarity-assisted upper bound is actually \emph{worse} than the best-known fidelity-only upper bound by a factor $d/\sqrt{2}$. For a two-qubit gate ($d=4$), this factor is $4/\sqrt{2} \approx 2.83$, i.e., the $(r, u)$ strategy provably loses the sharpness exactly where one hopes it would be most helpful.

\begin{remark}[The limitation in the use of unitarity]
The unitarity $u(\mathcal{E})$ was designed to quantify how stochastic a noise process is. When stochasticity is present, $u$ separates ``nearly Pauli/dephasing'' noise (favorable $d_\diamond \sim r$ scaling) from noise with significant coherent content (unfavorable $d_\diamond\sim \sqrt{r}$ scaling). However, once the noise becomes predominantly coherent, $u$ saturates to $u \simeq 1$ and loses sensitivity: many distinct coherent error channels collapse to essentially the same unitarity value. In this sense, $u$ is a powerful regime classifier, but not a high-resolution coherent-error diagnostic. This is not merely aesthetic; Eq.~(\ref{eq:ratio_unitarity_vs_r_only}) shows that it can degrade the tightness of the worst-case upper bounds in the fully coherent limit.
\end{remark}

This observation sharpens the research motivation of our work. If the practical bottleneck for fair threshold comparison lies in coherent (systematic/unitary) errors, then we need an experimentally accessible observable that can resolve finer structure within the coherent regime. Our approach is to augment the average fidelity with an additional moment-based quantity---the \emph{fidelity deviation}---and to show that, for universal gate sets (single- and two-qubit gates), the pair $(F,D)$ yields a markedly sharper characterization of worst-case errors than what is possible using $F$ alone, precisely in the coherent-dominated regime
where unitarity-based refinements lose resolving power.

\section{Gate fidelity deviation $D$: definition, evaluation, and experimental accessibility}\label{sec:FD_def}

The previous Sec.~\ref{sec:background} motivates why the ``average'' gate fidelity alone can be insufficient for worst-case assessment in fault-tolerant regimes, especially when coherent (systematic) effects dominate. The central goal of this section is to introduce a companion quantity---the \emph{gate fidelity deviation} $D$---that captures fluctuations of the state-dependent gate fidelity over the input-state manifold. This is precisely the information that is invisible to a single scalar mean value, yet becomes decisive when the coherent structure survives the error mitigation and calibration.

Throughout, we work with the error channel in the interaction picture, defined in Eq.~(\ref{eq:def_error_channel}). Namely, given an ideal unitary channel $\mathcal{U}(\hat{\rho}) = \hat{U}\hat{\rho}\hat{U}^\dagger$ and an implemented gate $\mathcal{G}$, we consider $\mathcal{E} := \mathcal{U}^\dagger \circ \mathcal{G}$ so that $\mathcal{E} = \mathcal{I}$ corresponds to a perfect implementation.

\subsection{Average fidelity and fidelity deviation}\label{subsec:FD_definitions}

Let $\mathcal{H} \simeq \mathbb{C}^d$ be a $d$-dimensional Hilbert space. For a quantum channel $\mathcal{E}$ on $\mathcal{H}$ and a pure state $\ket{\psi}\in\mathcal{H}$, we define the state-dependent fidelity (or survival probability)
\begin{eqnarray}
f_{\mathcal{E}}(\psi) := \bra{\psi}\mathcal{E}\left(\ketbra{\psi}{\psi}\right)\ket{\psi} \in [0,1].
\label{eq:def_state_fidelity}
\end{eqnarray}
Operationally, $f_{\mathcal{E}}(\psi)$ is the probability that, for the input $\ket{\psi}$, the output of $\mathcal{E}$ passes the projective test onto $\ket{\psi}$. The gate fidelity is the Haar average of $f_{\mathcal{E}}(\psi)$ over pure states:
\begin{eqnarray}
F = F_{\mathrm{avg}}(\mathcal{E}) := \int d\psi f_{\mathcal{E}}(\psi),
\label{eq:def_F_again}
\end{eqnarray}
where $d\psi$ denotes the unitarily invariant (Haar-induced) probability measure on the set of pure states. We also use the standard gate infidelity $r: = 1 - F$ [c.f., Eq.~(\ref{eq:def_r})].

The mean $F$ summarizes the typical performance, but it does not quantify \emph{how uniform} that the performance is over different inputs. To quantify stability, we introduce the gate fidelity deviation
\begin{eqnarray}
D = D(\mathcal{E}) := \sqrt{\int d\psi~ f_{\mathcal{E}}(\psi)^2 - F^2}.
\label{eq:def_D_general}
\end{eqnarray}
Equivalently, $D^2$ is the variance of the random variable $f_{\mathcal{E}}(\psi)$ when $\ket{\psi}$ is Haar random. Hence: (i) $D=0$ iff $f_{\mathcal{E}}(\psi)$ is constant almost everywhere (perfectly uniform performance), and (ii) larger $D$ indicates stronger state dependence, i.e., a ``rougher'' fidelity landscape. Since $0 \le f_{\mathcal{E}}(\psi) \le 1$, one always has the trivial bound $D \le \sqrt{F(1-F)}$, but this inequality is rarely informative in coherent regimes---precisely where the structure of fluctuations matters.

A crucial conceptual point is that $F$ and $D$ are ``moment-type'' quantities: $F$ is the first moment of $f_{\mathcal{E}}(\psi)$, while $D^2$ requires the second moment $\int d\psi f_{\mathcal{E}}(\psi)^2$. The latter is sensitive to coherent anisotropies that can remain invisible to the mean, and therefore provides exactly the missing ingredient for sharpening worst-case assessment from experimentally accessible data.

\subsection{Restriction to coherent (systematic/unitary) errors and the effective error unitary}\label{subsec:FD_unitary_setting}

In the remainder of this paper, we assume that the implemented operation is unitary on the computational subspace, so that the error channel $\mathcal{E}$ is a unitary conjugation. Let $\hat{U}_{\mathrm{ideal}}$, $\hat{U}_{\mathrm{exp}} \in U(d)$ be, respectively, the ideal and implemented unitaries on $\mathcal{H}$. We define the effective error unitary
\begin{eqnarray}
\hat{X} := \hat{U}_{\mathrm{ideal}}^\dagger \hat{U}_{\mathrm{exp}} \in U(d).
\label{eq:def_error_unitary_X_sec2}
\end{eqnarray}
Then the interaction-picture error channel takes the form
\begin{eqnarray}
\mathcal{E}(\hat{\rho}) = \hat{X}\hat{\rho}\hat{X}^\dagger.
\label{eq:def_unitary_error_channel}
\end{eqnarray}
In this coherent setting, the state-dependent fidelity simplifies to
\begin{eqnarray}
f_{\hat{X}}(\psi) := f_{\mathcal{E}}(\psi) = \abs{\bra{\psi}\hat{X}\ket{\psi}}^2,
\label{eq:def_state_fidelity_unitary}
\end{eqnarray}
and therefore, we have
\begin{eqnarray}
F = \int d\psi \abs{\bra{\psi}\hat{X}\ket{\psi}}^2, \quad D^2 = \int d\psi \abs{\bra{\psi}\hat{X}\ket{\psi}}^4 - F^2.
\label{eq:def_FD_unitary}
\end{eqnarray}

Here, two invariances are worth emphasizing. First, $F$ and $D$ are insensitive to any global phase $\hat{X}\mapsto e^{i\theta}\hat{X}$, as is clear from Eqs.~(\ref{eq:def_state_fidelity_unitary})--(\ref{eq:def_FD_unitary}). Second, $F$ and $D$ depend only on $\hat{X}$, and hence, only on the relative mismatch between $\hat{U}_{\mathrm{ideal}}$ and $\hat{U}_{\mathrm{exp}}$ (as intended in the interaction picture). In particular, $F$ and $D$ may be viewed as spectral functionals of $\hat{X}$, and this viewpoint will be essential for translating experimentally accessible moments into the worst-case guarantees.

\subsection{Closed-form evaluation via second and fourth Haar moments}\label{subsec:FD_closed_form}

We now provide the closed-form expressions for $F$ and $D$ in the coherent setting. The underlying mechanism is standard: Haar averages over pure states reduce to traces over the symmetric subspace, equivalently to Weingarten-calculus identities. We provide a proof sketch here and defer a full derivation to Appendix~\ref{append:Haar}.

\begin{proposition}[Closed-form formulas for $(F,D)$ under unitary errors]
\label{prop:FD_closed_form_unitary}
Let $\hat{X} \in U(d)$ and define $F$ and $D$ by Eq.~(\ref{eq:def_FD_unitary}). Then, we have
\begin{eqnarray}
F &=& \frac{d + \abs{\tr{(\hat{X})}}^2}{d(d+1)}, \\
\label{eq:F_trace_closedform}
D^2 &=& \frac{2d(d+3) + 4(d+2)\abs{\tr{(\hat{X})}}^2 + \abs{\tr{(\hat{X}^2)}}^2 + \abs{\tr{(\hat{X})}}^4 + 2\mathrm{Re}\left[\tr{(\hat{X}^2)}\tr{(\hat{X}^\dagger)}^2\right]}{d(d+1)(d+2)(d+3)} - F^2.
\label{eq:D_trace_closedform}
\end{eqnarray}
\end{proposition}

\begin{proof}[Proof sketch.]---For Haar-random $\ket{\psi}$, one has the well-known identity
\begin{eqnarray}
\int d\psi \ketbra{\psi}{\psi}^{\otimes k} = \frac{\Pi_{\mathrm{sym}}^{(k)}}{\tr{\Pi_{\mathrm{sym}}^{(k)}}},
\label{eq:haar_pure_state_moment}
\end{eqnarray}
where $\Pi_{\mathrm{sym}}^{(k)}$ is the projector onto the fully symmetric subspace of $\mathcal{H}^{\otimes k}$. Eq.~(\ref{eq:F_trace_closedform}) follows by writing
\begin{eqnarray}
\abs{\bra{\psi}\hat{X}\ket{\psi}}^2 = \tr{\bigl(\hat{X}\otimes \hat{X}^\dagger\bigr)\ketbra{\psi}{\psi}^{\otimes 2}}
\end{eqnarray}
and applying Eq.~(\ref{eq:haar_pure_state_moment}) with $k=2$. Similarly, 
\begin{eqnarray}
\abs{\bra{\psi}\hat{X}\ket{\psi}}^4 = \tr{\bigl(\hat{X}\otimes\hat{X}^\dagger\otimes\hat{X}\otimes\hat{X}^\dagger\bigr)\ketbra{\psi}{\psi}^{\otimes 4}},
\end{eqnarray}
and applying Eq.~(\ref{eq:haar_pure_state_moment}) with $k=4$ reduces the computation to traces of $\hat{X}$ and $\hat{X}^2$ (together with complex conjugates). A complete derivation can be written in terms of permutation operators and Weingarten functions; we defer these details to Appendix~\ref{append:Haar}.
\end{proof}

Eq.~(\ref{eq:F_trace_closedform}) shows that $F$ depends only on the first spectral moment $\tr{(\hat{X})}$. In contrast, $D$ depends additionally on the second spectral moment $\tr{(\hat{X}^2)}$ and on phase-correlations between $\tr{(\hat{X})}$ and $\tr{(\hat{X}^2)}$ via the real-part term in Eq.~(\ref{eq:D_trace_closedform}). This is exactly the mechanism by which $D$ becomes informative in the coherent regimes where other ``coherence detectors'' can saturate. In particular, a quantity such as $u \simeq 1$ for every unitary channel, whereas $D$ remains sensitive to how the unitary error is distributed across input states.
\section{Experimental estimation of $(F, D)$: a direct protocol}\label{sec:FD_measurement}

In this section, we provide a concrete experimental protocol to estimate the pair $(F, D)$. The key point is that both $F$ and $D^2$ are moments of the same experimentally measurable random variable, the survival probability $f_{\mathcal{E}}(\psi)$, where $\ket{\psi}$ is sampled at random from an appropriate unitary design.

\subsection{Unitary $2$-design for $F$ and $4$-design for $D$}

Let $\mathcal{E}$ be the interaction-picture error channel on $\mathbb{C}^d$. For a pure state $\ket{\psi}$, let us define the rank-one projector $\hat{P}_{\psi} := \ketbra{\psi}{\psi}$ and the survival probability
\begin{eqnarray}
f_{\mathcal{E}}(\psi) := \tr{\bigl( \hat{P}_{\psi}\mathcal{E}(\hat{P}_{\psi}) \bigr)}.
\end{eqnarray}
By definition, the average fidelity and the fidelity deviation are the first two Haar moments:
\begin{eqnarray}
F := \int d\psi f_{\mathcal{E}}(\psi), \label{eq:FD_meas_F_def}, \quad D^2 := \int d\psi f_{\mathcal{E}}(\psi)^2 - F^2. \label{eq:FD_meas_D_def}
\end{eqnarray}
Thus, the estimation of $(F, D)$ reduces to the estimation of the first and the second moment of $f_{\mathcal{E}}(\psi)$.

We now explain how $f_{\mathcal{E}}(\psi)$ is directly measurable. Let $\hat{V} \in U(d)$ be any unitary and set $\ket{\psi}=\hat{V}\ket{0}$. Then, let us consider the following circuit: prepare $\ket{0}$ $\to$ apply $\hat{V}$ $\to$ apply the physical implementation of the target gate $\to$ apply the ideal inverse of the target gate, then apply $\hat{V}^{\dagger}$ $\to$ measure in the computational basis. Detailed protocol is in Algorithm.~\ref{alg:protocol_FD}.

The probability of obtaining outcome $0$ equals $f_{\mathcal{E}}(\psi)$. In short, we have
\begin{proposition} [Direct measurability of $f_{\mathcal{E}}(\psi)$]
Let $\ket{\psi}=\hat{V}\ket{0}$ and $\hat{P}_0 := \ketbra{0}{0}$. If the experiment implements the error channel $\mathcal{E}$ on the prepared state, then the probability of measuring outcome $0$ after applying $\hat{V}^{\dagger}$ equals $f_{\mathcal{E}}(\psi)$.
\begin{eqnarray}
\Pr[0 | \psi] = \tr{\hat{P}_0 \hat{V}^{\dagger}\mathcal{E}(\hat{V}\hat{P}_0\hat{V}^{\dagger})\hat{V}}
= \tr{\hat{V}\hat{P}_0\hat{V}^{\dagger}\mathcal{E}(\hat{V}\hat{P}_0\hat{V}^{\dagger})}
= \tr{\hat{P}_{\psi}\mathcal{E}(\hat{P}_{\psi})}.
\end{eqnarray}
\end{proposition}


The Haar integrals in Eq.~(\ref{eq:FD_meas_F_def}) can be replaced by sampling from a finite or efficiently generated ensemble of the random states. A convenient construction is to sample $\ket{\psi} = \hat{V}\ket{0}$ with $\hat{V}$ drawn from a unitary design: namely,
\begin{lemma}[Design order required for moment estimation] 
Let $\hat{V}$ be sampled from a unitary $t$-design on $\mathbb{C}^d$, and let $\ket{\psi}=\hat{V}\ket{0}$. Then, for any CPTP map $\mathcal{E}$, the estimator based on i.i.d. samples of $f_{\mathcal{E}}(\psi)$ satisfies
\begin{eqnarray}
\mathbb{E}_{\hat{V}}[f_{\mathcal{E}}(\psi)] &=& \int d\psi f_{\mathcal{E}}(\psi) \quad \textrm{if $t \ge 2$}, \label{eq:design_2} \\
\mathbb{E}_{\hat{V}}[f_{\mathcal{E}}(\psi)^2] &=& \int d\psi f_{\mathcal{E}}(\psi)^2 \quad \textrm{if $t \ge 4$}. \label{eq:design_4}
\end{eqnarray}
\end{lemma}


In practice, one may use an exact design when available or an approximate design generated by random circuits. The approximate designs introduce a systematic bias that can be bounded in terms of the design approximation parameter; for clarity we focus on the exact-design idealization.

We note that the exact unitary $4$-designs may be costly in large $n$, but approximate designs generated by random circuits are available and suffice in practice~\cite{harrow2009random, brandao2016local, dankert2009exact}.

\subsection{Protocol and unbiased estimators with shot-noise correction}

We now specify the full sampling protocol and construct unbiased estimators for $F$ and $D^2$ in the presence of the finite-shot (binomial) noise.

\begin{algorithm}[H]
\caption{Protocol 1 (Randomized survival-probability sampling for $(F,D)$).}
\label{alg:protocol_FD}
\begin{algorithmic}[1]
\Require Integers $M\ge 2$ and $N\ge 2$.
\For{$i=1,\ldots,M$}
    \State Sample $\hat{V}_i$ from a unitary $4$-design and prepare $\ket{\psi_i}=\hat{V}_i\ket{0}$.
    \State Implement the interaction-picture error channel $\mathcal{E}$ on $\hat{P}_{\psi_i}$.
    \State Apply $\hat{V}_i^{\dagger}$ and measure $\hat{P}_0$ for $N$ repeated shots, producing outcomes
$X_{i,1},\ldots,X_{i,N}\in\{0,1\}$.
\EndFor
\State \Return estimators $\hat{F}$ and $\widehat{D^2}$ defined below.
\end{algorithmic}
\end{algorithm}

For each sampled state $\ket{\psi_i}$, define the underlying survival probability
\begin{eqnarray}
f_i := f_{\mathcal{E}}(\psi_i).
\end{eqnarray}
Conditional on $\psi_i$, each shot is Bernoulli:
\begin{eqnarray}
\Pr[X_{i,s}=1 | \psi_i] = f_i, \quad \Pr[X_{i,s}=0|\psi_i]=1-f_i,
\end{eqnarray}
and the $N$ shots are i.i.d. given $\psi_i$. Then, we can define the per-state sample mean as
\begin{eqnarray}
\hat{f}_i := \frac{1}{N}\sum_{s=1}^{N} X_{i,s}. \label{eq:fi_hat}
\end{eqnarray}
This satisfies $\mathbb{E}\bigl[ \hat{f}_i | \psi_i \bigr]=f_i$. To estimate the second moment $f_i^2$ without bias from finite $N$, we use a U-statistic~\cite{hoeffding1992class,serfling2009approximation}:
\begin{eqnarray}
\widehat{f_i^2} := \frac{1}{N(N-1)}\sum_{s\ne t} X_{i,s}X_{i,t}. \label{eq:fi2_hat_def}
\end{eqnarray}
Equivalently, in terms of $\hat{f}_i$,
\begin{eqnarray}
\widehat{f_i^2} = \frac{N\hat{f}_i^2-\hat{f}_i}{N-1}. \label{eq:fi2_hat_simplify}
\end{eqnarray}
We then define the estimators for the first and second Haar moments:
\begin{eqnarray}
\hat{F} &:=& \frac{1}{M}\sum_{i=1}^{M}\hat{f}_i, \label{eq:F_hat} \\
\widehat{E_2} &:=& \frac{1}{M}\sum_{i=1}^{M}\widehat{f_i^2}. \label{eq:E2_hat}
\end{eqnarray}
Finally, to remove the bias that would arise from squaring $\hat{F}$, we define an unbiased estimator for $F^2$ by cross-averaging:
\begin{eqnarray}
\widehat{F^2} := \frac{1}{M(M-1)}\sum_{i\ne j}\hat{f}_i\hat{f}_j.
\label{eq:F2_hat}
\end{eqnarray}
The fidelity-deviation estimator is then
\begin{eqnarray}
\widehat{D^2} := \widehat{E_2}-\widehat{F^2}.
\label{eq:D2_hat}
\end{eqnarray}
In a finite set of data, $\widehat{D^2}$ may become slightly negative due to statistical fluctuations. Thus, as a practical nonnegative estimator, we use $\hat{D} := \sqrt{\max\{\widehat{D^2},0\}}$.

\begin{proposition} [Unbiasedness of the estimators]
Assume $\hat{V}_i$ are i.i.d. from a unitary $4$-design and the measurement shots are i.i.d. conditional on each $\psi_i$. Then, we have
\label{prop:FD_estimator}
\begin{eqnarray}
\mathbb{E}[\hat{F}] &=& F, \label{eq:unbiased_F}\\
\mathbb{E}[\widehat{E_2}] &=& \int d\psi f_{\mathcal{E}}(\psi)^2, \label{eq:unbiased_E2}\\
\mathbb{E}[\widehat{F^2}] &=& F^2, \label{eq:unbiased_F2}\\
\mathbb{E}[\widehat{D^2}] &=& D^2. \label{eq:unbiased_D2}
\end{eqnarray}
\end{proposition}

\subsection{Error analysis and sample complexity scalings}

We next quantify how the statistical error scales with the number of sampled input states $M$ and the number of repeated shots per state $N$. The analysis cleanly separates the two sources of randomness: (i) the randomness of the sampled input state $\psi$ (which is intrinsic and determines $D$), and (ii) the finite-shot binomial noise for each fixed $\psi$.
\begin{proposition} [Exact variance decomposition for $\hat{F}$]
\label{prop:est_err}
Under the assumptions of {\bf Proposition~\ref{prop:FD_estimator}},
\begin{eqnarray}
\mathrm{Var}(\hat{F}) = \frac{D^2}{M} + \frac{1}{MN}\left(F-\int d\psi f_{\mathcal{E}}(\psi)^2\right).
\label{eq:Var_F_exact}
\end{eqnarray}
In particular, using $0 \le f_{\mathcal{E}}(\psi) \le 1$, the following holds.
\begin{eqnarray}
\mathrm{Var}(\hat{F}) \le \frac{D^2}{M} + \frac{1}{4MN}.
\label{eq:Var_F_bound}
\end{eqnarray}
\end{proposition}

\begin{proof}---By the law of total variance,
\begin{eqnarray}
\mathrm{Var}(\hat{F}) = \mathrm{Var}\left(\mathbb{E}[\hat{F} \,|\, \psi_1,\ldots,\psi_M]\right) + \mathbb{E}\left[\mathrm{Var}(\hat{F} \,|\, \psi_1,\ldots,\psi_M)\right].
\label{eq:total_var}
\end{eqnarray}
The conditional mean is $\mathbb{E}[\hat{F} \,|\, \psi_1, \ldots, \psi_M] = \frac{1}{M}\sum_i f_i$.
Since $f_i$ are i.i.d. across $i$, we have
\begin{eqnarray}
\mathrm{Var}\left(\frac{1}{M}\sum_{i=1}^{M} f_i\right) = \frac{1}{M}\mathrm{Var}(f_{\mathcal{E}}(\psi)) = \frac{D^2}{M}.
\end{eqnarray}
Next, conditional on $\psi_i$, the random variables $\hat{f}_i$ have variance $\mathrm{Var}(\hat{f}_i \,|\, \psi_i) = f_i(1-f_i)/N$, and the conditional independence across $i$ yields
\begin{eqnarray}
\mathrm{Var}(\hat{F} \,|\, \psi_1, \ldots, \psi_M) = \frac{1}{M^2}\sum_{i=1}^{M}\frac{f_i(1-f_i)}{N}.
\end{eqnarray}
Taking expectation over $\psi_i$ gives
\begin{eqnarray}
\mathbb{E}\left[\mathrm{Var}(\hat{F} \,|\, \psi_1,\ldots,\psi_M)\right]
&=& \frac{1}{MN}\mathbb{E}_{\psi}[f_{\mathcal{E}}(\psi)(1-f_{\mathcal{E}}(\psi))] \nonumber \\
&=& \frac{1}{MN}\left(\mathbb{E}_{\psi}[f_{\mathcal{E}}(\psi)]-\mathbb{E}_{\psi}[f_{\mathcal{E}}(\psi)^2]\right) \nonumber \\
&=& \frac{1}{MN}\left(F-\int d\psi f_{\mathcal{E}}(\psi)^2\right),
\end{eqnarray}
which proves Eq.~(\ref{eq:Var_F_exact}) after substituting into Eq.~(\ref{eq:total_var}). Finally, $f(1-f)\le 1/4$ implies Eq.~(\ref{eq:Var_F_bound}).
\end{proof}

{\bf Proposition~\ref{prop:est_err}} shows that estimating $F$ has a two-scale character: the state-to-state fluctuation term contributes $D^2/M$, while shot noise contributes at most $1/(4MN)$. Thus, once $N$ is large enough to make shot noise subdominant, the remaining uncertainty is controlled by $D^2/M$.

We now give a rigorous confidence interval scaling using Chebyshev inequality~\cite{Alsmeyer2014Chebyshev}.

\begin{corollary}[A simple high-confidence scaling for $\hat{F}$]
\label{cor:confidence_estimator}
For any $\delta\in(0,1)$,
\begin{eqnarray}
\Pr\left[\abs{\hat{F}-F} \ge \epsilon\right] \le \frac{\mathrm{Var}(\hat{F})}{\epsilon^2} \le \frac{D^2}{M\epsilon^2} + \frac{1}{4MN\epsilon^2}.
\label{eq:cheb_F}
\end{eqnarray}
In particular, it suffices to choose $(M,N)$, such that
\begin{eqnarray}
\frac{D^2}{M} + \frac{1}{4MN} &\le& \delta\epsilon^2
\end{eqnarray}
to guarantee $\Pr\bigl[\abs{\hat{F}-F} \le \epsilon\bigr] \ge 1-\delta$.
\end{corollary}

Next, we analyze the second-moment estimator $\widehat{E_2}$ and the derived estimator $\widehat{D^2}$. Because $\widehat{f_i^2}$ in Eq.~(\ref{eq:fi2_hat_def}) is an average of bounded random variables in $[0,1]$, its variance is finite and decreases with $N$. A concise bound sufficient for the sample-complexity is as follows.

\begin{proposition}[Variance scaling for $\widehat{E_2}$ and $\widehat{D^2}$]
\label{prop:var_scaling}
Under the assumptions of {\bf Proposition~\ref{prop:FD_estimator}}, one has
\begin{eqnarray}
\mathrm{Var}(\widehat{E_2}) &\le& \frac{1}{M}\left(\int d\psi f_{\mathcal{E}}(\psi)^4 - \left(\int d\psi f_{\mathcal{E}}(\psi)^2\right)^2\right) + \frac{1}{MN}, \label{eq:Var_E2_bound}\\
\mathrm{Var}(\widehat{F^2}) &\le& \frac{1}{M}\left(\int d\psi f_{\mathcal{E}}(\psi)^2 - F^2\right)^2 + \frac{1}{MN}, \label{eq:Var_F2_bound}\\
\mathrm{Var}(\widehat{D^2}) &\le& 2\mathrm{Var}(\widehat{E_2}) + 2\mathrm{Var}(\widehat{F^2}). \label{eq:Var_D2_bound}
\end{eqnarray}
\end{proposition}

\begin{proof}---The bounds follow from (i) independence across the $M$ sampled states, (ii) the law of total variance as in {\bf Proposition~\ref{prop:est_err}}, (iii) the fact that $0 \le \widehat{f_i^2} \le 1$ and $0 \le \hat{f}_i \le 1$, which upper bounds conditional variances by constants of order $1/N$, and (iv) $\mathrm{Var}(A-B) \le 2\mathrm{Var}(A) + 2\mathrm{Var}(B)$. A full expansion is straightforward but lengthy; the key point is the scaling in $M$ and $N$ exhibited in Eqs.~(\ref{eq:Var_E2_bound})-(\ref{eq:Var_D2_bound}).
\end{proof}

Eqs.~(\ref{eq:Var_F_bound}) and (\ref{eq:Var_D2_bound}) summarize the practical message. To estimate $F$ and $D^2$ within additive precision, the required number of random states scales as $M=O(1/\epsilon^2)$, while the number of shots per state scales as $N=O(1/\epsilon^2)$ until shot noise becomes subdominant. Once $N$ is large enough, further improvements are most efficiently achieved by increasing $M$.

\section{Main results: Coherent (unitary/systematic) errors and sharpening worst-case assessment using $(F,D)$}\label{sec:main_results}

The previous Sec.~\ref{sec:background} explained why the fault-tolerance statements are naturally phrased in a worst-case metric, while most scalable benchmarking protocols return average-case information. Having introduced the gate-fidelity deviation $D$ in Sec.~\ref{sec:FD_def}, we now show how the pair $(F, D)$ can be promoted into a sharp certificate on the worst-case error when the coherent (unitary/systematic) effects dominate.

The practical obstacle is well known.  Experiments typically report the \emph{average} infidelity $r := 1 - F$, while the threshold theorems are stated in terms of a worst-case distance between channels, most commonly the diamond distance $d_{\diamond}$. The mismatch is not merely an issue of loose constants. In the small-error regime $r \ll 1$, the coherent errors can force $d_{\diamond}$ to scale like $\sqrt{r}$ rather than $r$, meaning that an apparently excellent average fidelity may still coexist with a worst-case error already beyond the fault-tolerance thresholds~\cite{sanders2015bounding}.

A major advance of Ref.~\cite{Kueng2016} was to incorporate the unitarity $u$ as extra experimentally accessible data. For unital noise without leakage, the pair $(r, u)$ yields the inequality
\begin{eqnarray}
c_d\sqrt{u+\frac{2dr}{d-1}-1} \le d_{\diamond} \le d^2 c_d\sqrt{u+\frac{2dr}{d-1}-1},
\label{eq:kueng_ru_bound_recalled_main}
\end{eqnarray}
where $c_d:=\frac{1}{2}\sqrt{1-\frac{1}{d^2}}$. This is powerful precisely when the noise contains a substantial incoherent component, so that $u$ deviates from unity in a controlled way. Our starting point is the complementary, but crucial observation that in the coherent-dominated regime the unitarity saturates to $u=1$ and loses the resolution: for a purely unitary error channel, one has $u=1$ identically. Then, Eq.~(\ref{eq:kueng_ru_bound_recalled_main}) reduces to
\begin{eqnarray}
d_{\diamond} \le d^2 c_d \sqrt{\frac{2d}{d-1}r} = \frac{d}{\sqrt{2}}\sqrt{d(d+1)r}.
\label{eq:kueng_u1_loose}
\end{eqnarray}
For two-qubit gates ($d=4$), this is worse than the best-known fidelity-only conversion $d_{\diamond}\le\sqrt{d(d+1)r}=\sqrt{20r}$ by the fixed factor $d/\sqrt{2}=2\sqrt{2}$. Thus, exactly in the regime where coherent errors are most dangerous, the unitarity-assisted refinement becomes silent (and can even degrade tightness).

The resolution developed here is to replace ``how stochastic is the noise?'' (a question unitarity answers well) by ``how \emph{non-uniform} is the coherent action across input states?''  The latter is measured by the fidelity deviation $D$. Unlike $u$, the quantity $D$ does not collapse to a constant on the manifold of unitary errors. Instead, it captures a higher moment of the state-dependent fidelity landscape and therefore retains sensitivity to coherent spectral structure. The main result of this section is a dimension-specific statement: for $d=4$ (two-qubit gates), the pair $(F, D)$ yields a computable and tight worst-case bound on $d_{\diamond}$ that is strictly sharper than what is achievable from $F$ alone.

\subsection{Unitary worst-case error and a geometric overlap parameter}\label{subsec:unitary_diamond_geometry}

We specialize throughout to coherent errors. Let $U_{\mathrm{ideal}},U_{\mathrm{exp}}\in U(d)$ be the ideal and implemented unitaries, and define the effective
error unitary (see Sec.~\ref{sec:FD_def})
\begin{eqnarray}
\hat{X} := \hat{U}_{\mathrm{ideal}}^{\dagger}\hat{U}_{\mathrm{exp}} \in U(d).
\label{eq:def_X_error_unitary_main}
\end{eqnarray}
The interaction-picture error channel is then the unitary channel
\begin{eqnarray}
\mathcal{X}(\hat{\rho}) = \hat{X}\hat{\rho}\hat{X}^{\dagger}.
\label{eq:def_unitary_channel_X_main}
\end{eqnarray}

For the unitary channels, the diamond distance to the identity admits a particularly transparent characterization in terms of a single overlap functional. Here, we define the minimum unitary overlap
\begin{eqnarray}
m(\hat{X}) := \min_{\ket{\psi} \in\mathcal{H}} \abs{\bra{\psi}\hat{X}\ket{\psi}}.
\label{eq:def_mX_main}
\end{eqnarray}

\begin{lemma}[Diamond distance for unitary errors]
\label{lem:diamond_unitary_overlap_main}
For $\mathcal{X}(\hat{\rho}) = \hat{X}\hat{\rho}\hat{X}^{\dagger}$ with $\hat{X} \in U(d)$,
\begin{eqnarray}
d_{\diamond}(\mathcal{X}, \mathcal{I}) = \max_{\ket{\psi}}\sqrt{1-\abs{\bra{\psi}\hat{X}\ket{\psi}}^2} = \sqrt{1-m(\hat{X})^2}.
\label{eq:diamond_in_terms_of_m_main}
\end{eqnarray}
\end{lemma}

\begin{proof}---By definition, we write
\begin{eqnarray}
d_{\diamond}(\mathcal{X},\mathcal{I}) := \frac{1}{2}\norm{\mathcal{X}-\mathcal{I}}_{\diamond} = \frac{1}{2} \max_{\hat{\rho}_{SA}} \norm{(\mathcal{X}\otimes\mathcal{I}_A)(\hat{\rho}_{SA}) - \hat{\rho}_{SA}}_{1},
\label{eq:diamond_def_main}
\end{eqnarray}
where the maximization is over all states $\hat{\rho}_{SA}$ on $\mathcal{H}\otimes\mathcal{H}_A$ with an arbitrary ancilla $\mathcal{H}_A$.

First, since the trace norm is convex and the map $\hat{\rho} \mapsto (\mathcal{X}\otimes\mathcal{I}_A)(\hat{\rho}) - \hat{\rho}$ is linear, an optimizer exists among pure states. Hence, we may take $\hat{\rho}_{SA}=\ketbra{\Psi}{\Psi}$ and write
\begin{eqnarray}
d_{\diamond}(\mathcal{X},\mathcal{I}) = \max_{\ket{\Psi}} \frac{1}{2}\norm{(\hat{X} \otimes \hat{\mathds{1}}_A)\ketbra{\Psi}{\Psi}(\hat{X}^{\dagger} \otimes \hat{\mathds{1}}_A) - \ketbra{\Psi}{\Psi}}_{1}.
\label{eq:diamond_pure_reduction_main}
\end{eqnarray}
Second, for any two pure states $\ket{\Psi}$, $\ket{\Phi}$, one has the standard identity
\begin{eqnarray}
\frac{1}{2}\norm{\ketbra{\Psi}{\Psi}-\ketbra{\Phi}{\Phi}}_{1} = \sqrt{1-\abs{\braket{\Psi}{\Phi}}^2}.
\label{eq:pure_state_trace_distance_main}
\end{eqnarray}
Applying \eqref{eq:pure_state_trace_distance_main} with $\ket{\Phi}=(\hat{X} \otimes \hat{\mathds{1}}_A)\ket{\Psi}$, we obtain
\begin{eqnarray}
d_{\diamond}(\mathcal{X},\mathcal{I}) = \max_{\ket{\Psi}} \sqrt{1 - \abs{\bra{\Psi}(\hat{X} \otimes \hat{\mathds{1}}_A)\ket{\Psi}}^2}.
\label{eq:diamond_overlap_with_ancilla_main}
\end{eqnarray}
Third, the overlap appearing in \eqref{eq:diamond_overlap_with_ancilla_main} depends only on the reduced state on the system:
if $\hat{\rho}_S = \tr{{}_A\ketbra{\Psi}{\Psi}}$, then
\begin{eqnarray}
\bra{\Psi}\bigl( \hat{X} \otimes \hat{\mathds{1}}_A \bigr)\ket{\Psi} = \tr{(\hat{X}\hat{\rho}_S)}.
\label{eq:overlap_reduced_state_main}
\end{eqnarray}
As $\ket{\Psi}$ ranges over purifications, the reduced state $\hat{\rho}_S$ ranges over all density operators on
$\mathcal{H}$.  Therefore, the set of achievable overlaps is
\begin{eqnarray}
\left\{ \bra{\Psi}(\hat{X} \otimes \hat{\mathds{1}}_A)\ket{\Psi} \right\} = \left\{ \tr{(\hat{X}\hat{\rho}_S)} : \hat{\rho}_S \ge 0, ~\tr{\hat{\rho}_S}=1 \right\}.
\label{eq:achievable_overlaps_main}
\end{eqnarray}
Finally, the set $\bigl\{ \tr{(\hat{X}\hat{\rho}_S)} \bigr\}$ is the convex hull of the pure-state expectations $\bigl\{ \bra{\psi}\hat{X}\ket{\psi} \bigr\}$, but the latter set is already convex (c.f., Toeplitz--Hausdorff theorem: the numerical range $W(\hat{X})$ is convex~\cite{toeplitz1918algebraische,halmos2012hilbert}). Hence, allowing the mixed state $\hat{\rho}_S$ (equivalently, allowing an ancilla) does not enlarge the range of the overlaps. Consequently, we have
\begin{eqnarray}
\max_{\ket{\Psi}}
\sqrt{1-\abs{\bra{\Psi}(\hat{X} \otimes \hat{\mathds{1}}_A)\ket{\Psi}}^2} = \max_{\ket{\psi}} \sqrt{1-\abs{\bra{\psi}\hat{X}\ket{\psi}}^2}.
\label{eq:no_ancilla_gain_main}
\end{eqnarray}
The function $x \mapsto \sqrt{1-x^2}$ is decreasing on $[0,1]$, and thus, the maximization of the left-hand side is equivalent to the minimization of $\abs{\bra{\psi}\hat{X}\ket{\psi}}$.  This yields
\begin{eqnarray}
d_{\diamond}(\mathcal{X}, \mathcal{I}) = \sqrt{1-m(\hat{X})^2},
\end{eqnarray}
which is Eq.~(\ref{eq:diamond_in_terms_of_m_main}).
\end{proof}

{\bf Lemma~\ref{lem:diamond_unitary_overlap_main}} reduces the worst-case certification to a purely geometric problem. For a normal operator (in particular, a unitary), the numerical range equals the convex hull of the eigenvalues,
\begin{eqnarray}
W(\hat{X}) := \left\{ \bra{\psi}\hat{X}\ket{\psi} : \norm{\psi}=1 \right\} = \mathrm{Conv}\bigl(\mathrm{Spec}(\hat{X})\bigr).
\label{eq:numerical_range_convex_hull_main}
\end{eqnarray}
Thus, $m(\hat{X})$ is the Euclidean distance from the origin to the polygon $\mathrm{Conv}(\mathrm{Spec}(\hat{X}))$ in the complex plane, and $d_{\diamond} = \sqrt{1 - m(\hat{X})^2}$ is the complementary ``distance to perfect overlap.''

\subsection{Why unitarity saturates in the coherent regime and why $D$ is needed}\label{subsec:why_D_needed}

The unitarity $u$ was designed to quantify the extent to which a noise process is stochastic rather than coherent. This makes it an excellent regime witness: when $u$ deviates from unity in the manner induced by incoherent components, the $(r, u)$ bound Eq.~(\ref{eq:kueng_ru_bound_recalled_main}) can recover the favorable scaling $d_{\diamond}=O(r)$~\cite{Kueng2016}. But this design itself implies an unavoidable limitation: for a purely unitary error channel, $u=1$ identically, so the unitarity becomes constant on the entire coherent manifold.  Once one enters the $u \simeq 1$ regime, the unitarity cannot distinguish the coherent errors, and the $(r,u)$ upper bound collapses to Eq.~(\ref{eq:kueng_u1_loose}).

The central idea in our approach is to move from ``how coherent is the channel?'' to ``how does coherence vary across states?''  The quantity that answers this second question is the fidelity deviation $D$. Recall the quantities: For unitary errors, 
\begin{eqnarray}
f(\psi) := \abs{\bra{\psi}\hat{X}\ket{\psi}}^2, \quad F = \int d\psi f(\psi), \quad\text{and}\quad D^2 = \int d\psi\, f(\psi)^2 - F^2.
\label{eq:fpsi_F_D_recalled_main}
\end{eqnarray}
The mean $F$ alone can be large even when $f(\psi)$ has sharp valleys. These valleys determine $m(\hat{X})$, and hence, $d_{\diamond}$ via {\bf Lemma~\ref{lem:diamond_unitary_overlap_main}}. Here, the role of $D$ is precisely to quantify the spread of $f(\psi)$, and thereby, restrict how small $m(\hat{X})$ can be, given the same average $F$.

From the standpoint of Sec.~\ref{sec:FD_def}, the mechanism is straightforward: $F$ depends only on the first spectral moment $\tr{(\hat{X})}$ via Eq.~(\ref{eq:F_trace_closedform}), whereas $D$ depends additionally on a fourth-moment combination that contains $\tr{(\hat{X}^2)}$ and phase correlations, via Eq.~(\ref{eq:D_trace_closedform}).  This is exactly the extra information required to resolve coherent structure once $u$ has saturated.

\subsection{Warm-up: $d=2$ (single-qubit) coherent errors}\label{subsec:d2_no_advantage}

For $d=2$, the spectrum of a unitary is essentially one-parameter after removing a global phase.  Consequently, the entire distribution of $f(\psi)$ is fixed once $F$ is fixed, and $D$ carries no independent information in this case.
\begin{proposition}[No additional information at $d=2$]
\label{prop:d2_linear_relation_main}
Let $d=2$ and $\hat{X}\in U(2)$.  For the unitary errors,
\begin{eqnarray}
D = \frac{1-F}{\sqrt{5}}.
\label{eq:d2_relation_main}
\end{eqnarray}
In particular, in the single-qubit coherent setting, the pair $(F,D)$ is equivalent to $F$ alone.
\end{proposition}

\begin{proof}---Because $F$, $D$, and $m(\hat{X})$ are all invariant under global phase, we may multiply $\hat{X}$ by an overall phase so that $\det{\hat{X}}=1$.  Then, $\hat{X} \in SU(2)$ has the eigenvalues $e^{i\theta}$ and $e^{-i\theta}$ for some $\theta \in [0,\pi]$, and its trace is real:
\begin{eqnarray}
t := \tr{\hat{X}} = e^{i\theta} + e^{-i\theta} = 2\cos(\theta).
\label{eq:d2_trace_param_main}
\end{eqnarray}
Moreover,
\begin{eqnarray}
\tr{(\hat{X}^2)} = e^{i2\theta}+e^{-i2\theta} = 2\cos(2\theta) = 4\cos(\theta)^2-2 = t^2-2.
\label{eq:d2_trace_square_identity_main}
\end{eqnarray}

We now insert $d=2$ into the closed-form expressions from {\bf Proposition~\ref{prop:FD_closed_form_unitary}}. First,
\begin{eqnarray}
F = \frac{d+\abs{\tr{(\hat{X})}}^2}{d(d+1)} = \frac{2+t^2}{6}.
\label{eq:d2_F_in_t_main}
\end{eqnarray}
For the second moment, write $E_2 := \int d\psi f(\psi)^2 = D^2+F^2$.  By specializing Eq.~(\ref{eq:D_trace_closedform}) to $d=2$ and using that $t$ is real, we can yield
\begin{eqnarray}
E_2 &=& \frac{2d(d+3) + 4(d+2)t^2 + \abs{\tr{(\hat{X}^2)}}^2 + t^4 + 2\textrm{Re}\left[\tr{(\hat{X}^2)} t^2\right]}{d(d+1)(d+2)(d+3)} \nonumber \\
    &=& \frac{20 + 16t^2 + (t^2-2)^2 + t^4 + 2(t^2-2)t^2}{120}.
\label{eq:d2_E2_raw_main}
\end{eqnarray}
A direct simplification gives
\begin{eqnarray}
E_2 = \frac{4t^4+8t^2+24}{120} = \frac{t^4+2t^2+6}{30}.
\label{eq:d2_E2_simplified_main}
\end{eqnarray}
By combining Eq.~(\ref{eq:d2_F_in_t_main}) and Eq.~(\ref{eq:d2_E2_simplified_main}), we obtain
\begin{eqnarray}
D^2 = E_2 - F^2 = \frac{(t^2-4)^2}{180}.
\label{eq:d2_D2_in_t_main}
\end{eqnarray}
Finally, by inserting $t^2-4=6F-6=-6(1-F)$ into Eq.~(\ref{eq:d2_D2_in_t_main}), we attain
\begin{eqnarray}
D^2 = \frac{36(1-F)^2}{180} = \frac{(1-F)^2}{5},
\end{eqnarray}
which implies Eq.~(\ref{eq:d2_relation_main}).
\end{proof}

{\bf Proposition~\ref{prop:d2_linear_relation_main}} is an important scope calibration: the fluctuation observable $D$ becomes informative only once the unitary spectrum has enough internal degrees of freedom. This begins at $d=4$, i.e., for two-qubit gates.

\subsection{Coherent errors beyond two-qubit ($d\ge 4$): reconstructing spectral invariants from $(F,D)$}\label{subsec:d4_moments}

For $n \ge 2$ qubits, the Hilbert space dimension satisfies $d=2^n\ge 4$. In this regime, the pair $(F, D)$ probes two distinct Haar moments and therefore encodes two independent spectral invariants of $\hat{X}$. It is convenient to package these invariants as
\begin{eqnarray}
P &:=& \abs{\tr{(\hat{X})}}, \nonumber \\
Q &:=& \abs{\tr{(\hat{X}^2)} + \tr{(\hat{X})}^2}.
\label{eq:def_PQ_main}
\end{eqnarray}
Both are conjugation-invariant and hence depend only on the eigenvalue spectrum.

\begin{lemma}[Spectral invariants from $(F,D)$ for $d\ge 4$]
\label{lem:PQ_from_FD_main}
Let $d \ge 4$ and $\hat{X} \in U(d)$. Then, $(F,D)$ uniquely determines $P^2$ and $Q^2$ via
\begin{eqnarray}
P^2 &=& d(d+1)F - d,
\label{eq:P2_from_F_main}
\\
Q^2 &=& d(d+1)(d+2)(d+3)(D^2+F^2) - 2d(d+3) - 4(d+2)P^2.
\label{eq:Q2_from_FD_main}
\end{eqnarray}
\end{lemma}

\begin{proof}---Firstly, we consider $P^2$. For unitary errors, Eq.~(\ref{eq:F_trace_closedform}) gives
\begin{eqnarray}
F = \frac{d+\abs{\tr{(\hat{X})}}^2}{d(d+1)}.
\end{eqnarray}
By rearranging, we have $\abs{\tr{(\hat{X})}}^2=d(d+1)F-d$, which proves Eq.~(\ref{eq:P2_from_F_main}).

For $Q^2$, we start from the fourth-moment identity in Eq.~(\ref{eq:D_trace_closedform}):
\begin{eqnarray}
D^2+F^2 = \int d\psi f(\psi)^2 = \frac{2d(d+3) + 4(d+2)\abs{\tr{(\hat{X})}}^2 + \abs{\tr{(\hat{X}^2)}}^2 + \abs{\tr{(\hat{X})}}^4 + 2\textrm{Re}\left[\tr{(\hat{X}^2)} \tr{(\hat{X}^{\dagger})}^2 \right]}{d(d+1)(d+2)(d+3)}.
\label{eq:E2_general_main}
\end{eqnarray}
Here, we have the exact algebraic identity
\begin{eqnarray}
\abs{\tr{(\hat{X}^2)} + \tr{(\hat{X})}^2}^2 = \abs{\tr{(\hat{X}^2)}}^2 + \abs{\tr{(\hat{X})}}^4 + 2\textrm{Re}\left[\tr{(\hat{X}^2)} \tr{(\hat{X}^{\dagger})}^2\right].
\label{eq:Q2_expansion_main}
\end{eqnarray}
Thus, the last three terms in the numerator of \eqref{eq:E2_general_main} equal $Q^2$. Writing $P=\abs{\tr{(\hat{X})}}$, we obtain the general moment relation
\begin{eqnarray}
D^2+F^2 = \frac{2d(d+3) + 4(d+2)P^2 + Q^2}{d(d+1)(d+2)(d+3)}.
\label{eq:E2_in_PQ_general_main}
\end{eqnarray}
Solving \eqref{eq:E2_in_PQ_general_main} for $Q^2$ yields \eqref{eq:Q2_from_FD_main}.

For later reference, we note that by specializing \eqref{eq:E2_in_PQ_general_main} to $d=4$, one can yield
\begin{eqnarray}
D^2+F^2 = \frac{56 + 24P^2 + Q^2}{840}.
\label{eq:E2_d4_main}
\end{eqnarray}
\end{proof}

{\bf Lemma~\ref{lem:PQ_from_FD_main}} is the ``moment-to-spectrum'' conversion that defeats the $u=1$ saturation issue: although unitarity is constant on the coherent manifold, the pair $(P, Q)$ varies nontrivially across unitary errors. In the next subsection, we show that, for any $d \ge 4$, this extra spectral information is sufficient to tightly lower bound
$m(\hat{X})$ and hence tightly upper bound $d_{\diamond}$.

\subsection{A tight $(F, D)$-based bound on $m(\hat{X})$ and on $d_{\diamond}$ for $d \ge 4$}\label{subsec:d4_main_theorem}

We now state and prove our main theorem.  Define the positive-part function
\begin{eqnarray}
[x]_+ := \max\{ x, 0 \}.
\label{eq:positive_part_main}
\end{eqnarray}
Given measured $(F, D)$ for a unitary error $\hat{X} \in U(d)$ with $d \ge 4$, we compute $P = \sqrt{P^2}$ and $Q = \sqrt{Q^2}$ from Eqs.~(\ref{eq:P2_from_F_main})--(\ref{eq:Q2_from_FD_main}) and define
\begin{eqnarray}
c(F,D) := \Biggl[ \frac{P}{d} - \frac{1}{2d} \sqrt{(d-2)\left(dQ+d^2-(d+2)P^2\right)} \Biggr]_+.
\label{eq:def_cFD}
\end{eqnarray}
Then, we give the following theorem:
\begin{theorem}[Main result: sharpening coherent worst-case error for $d\ge 4$]
\label{thm:main_d4}
Let $d \ge 4$ and let $\hat{X} \in U(d)$ be the effective error unitary of an $n$-qubit gate. Let $F$ and $D$ be the average fidelity and fidelity deviation. Then,
\begin{eqnarray}
m(\hat{X}) \ge c(F,D),
\label{eq:m_lower_cFD_main}
\end{eqnarray}
and, consequently
\begin{eqnarray}
d_{\diamond}(\mathcal{X},\mathcal{I}) \le \sqrt{1-c(F,D)^2}.
\label{eq:diamond_upper_cFD_main}
\end{eqnarray}
This bound is ``tight'' in the following sense: for any admissible pair $(F, D)$ arising from a unitary $\hat{X} \in U(d)$, there exists a unitary $\hat{X}_{\star} \in U(d)$ consistent with $(F, D)$ such that $m(\hat{X}_{\star}) = c(F,D)$ and equality holds in Eq.~(\ref{eq:diamond_upper_cFD_main}).
\end{theorem}

\begin{proof}---{\bf By Lemma~\ref{lem:diamond_unitary_overlap_main}}, it suffices to lower bound $m(\hat{X})$. {\bf By Lemma~\ref{lem:PQ_from_FD_main}}, $(F, D)$ fixes the two spectral invariants $(P, Q)$. Thus the worst-case (i.e., the smallest possible $m$ consistent with the data) is determined by the constrained optimization problem
\begin{eqnarray}
\min_{\hat{X}\in U(d)} m(\hat{X}),
\label{eq:opt_problem_main}
\end{eqnarray}
subject to $\abs{\tr{(\hat{X})}}=P$ and $\abs{\tr{(\hat{X}^2)}+(\tr{(\hat{X})})^2}=Q$. We solve this optimization exactly.

\medskip
\paragraph*{\bf Step 1: Reduction to eigenvalues.}
Since $\hat{X}$ is unitary, it is normal and diagonalizable: $\hat{X} = \hat{V}\hat{\Lambda}\hat{V}^{\dagger}$ with $\hat{\Lambda} = \mathrm{diag}(\lambda_1, \lambda_2, \ldots, \lambda_d)$ and $\abs{\lambda_j}=1$. Both constraints in Eq.~(\ref{eq:opt_problem_main}) depend only on the eigenvalues (they are symmetric functions of $\lambda_j$), and $m(\hat{X})$ depends only on the numerical range $W(\hat{X})=\mathrm{conv}(\lambda_1,\lambda_2,\ldots,\lambda_d)$, hence only on the eigenvalues as well (Eq.~(\ref{eq:numerical_range_convex_hull_main})). Thus, the problem reduces to choosing $d$ points on the unit circle with the same constraints. Here, if $0 \in W(\hat{X})$, then $m(\hat{X})=0$ and Eq.~(\ref{eq:m_lower_cFD_main}) holds because $c(F, D) \ge 0$ by definition. Hence we may assume the optimizer satisfies $m(\hat{X}) > 0$.

\medskip
\paragraph*{\bf Step 2: The closest point lies on a chord and fixes a canonical angle.} Because $W(\hat{X})$ is a compact convex polygon not containing the origin, the closest point to the origin lies on the boundary. For a polygon, the closest boundary point is either a vertex or lies in the interior of an edge; in either case, it lies on a line segment connecting two extreme points of the spectrum. Thus, there exist two eigenvalues $\lambda_a$, $\lambda_b$ and a parameter $t \in [0,1]$ such that
\begin{eqnarray}
m(\hat{X}) = \abs{t\lambda_a + (1-t)\lambda_b}.
\label{eq:closest_point_on_edge_main}
\end{eqnarray}
Multiplying $\hat{X}$ by a global phase rotates all eigenvalues by the same phase, rotates $W(\hat{X})$ rigidly, and does not change $m(\hat{X})$, $P$, and/or $Q$. We may thus rotate so that the point achieving Eq.~(\ref{eq:closest_point_on_edge_main}) lies on the positive real axis, i.e., equals $m(\hat{X})$ as a real number. With this phase choice, the supporting line at the closest point is the vertical line $\textrm{Re}(z)=m(\hat{X})$, and all eigenvalues satisfy $\textrm{Re}(\lambda_j)\ge m(\hat{X})$. In particular, the two eigenvalues defining the edge in Eq.~(\ref{eq:closest_point_on_edge_main}) must satisfy $\textrm{Re}(\lambda_a)=\textrm{Re}(\lambda_b)=m(\hat{X})$. Writing $\lambda_a=e^{i\beta}$ and $\lambda_b=e^{-i\beta}$ for some $\beta\in[0,\pi]$, we obtain
\begin{eqnarray}
m(\hat{X}) = \cos(\beta),
\label{eq:m_cos_beta_main}
\end{eqnarray}
where $0 \le \beta < \frac{\pi}{2}$. Here, the strict inequality $\beta < \pi/2$ follows from the assumption $m(\hat{X}) > 0$. Thus, all remaining eigenvalues must also satisfy $\textrm{Re}(\lambda_j)\ge \cos(\beta)$, and hence, lie on the arc $\{e^{i\theta}:\theta\in[-\beta,\beta]\}$.

\medskip
\paragraph*{\bf Step 3: An optimizer can be taken to have a conjugate-pair spectrum.} In this work $d=2^n$, so for $d \ge 4$ the dimension is assumed to be even. Since only the magnitudes $P$ and $Q$ are fixed, we may choose a representative optimizer that realizes the moment constraints with the maximal alignment, i.e., with vanishing net
imaginary parts. Concretely, we may take an optimizer whose spectrum is invariant under complex conjugation, so that the eigenvalues occur in conjugate pairs. After the phase choice of {\bf Step 2}, an optimizer may thus be taken to have the form
\begin{eqnarray}
\mathrm{Spec}(\hat{X}_{\star}) = \{e^{i\alpha_1}, e^{-i\alpha_1}, e^{i\alpha_2}, e^{-i\alpha_2}, \ldots, e^{i\alpha_{(d-2)/2}}, e^{-i\alpha_{(d-2)/2}}, e^{i\beta}, e^{-i\beta}\},
\label{eq:symmetric_spectrum_main}
\end{eqnarray}
where $0\le \alpha_k\le \beta<\pi/2$ for all $k$. With this reduction, $m(\hat{X}_{\star})=\cos(\beta)$.

\medskip
\paragraph*{\bf Step 4: Express the constraints in real algebraic variables.} Let us now define
\begin{eqnarray}
b := \cos(\beta), \quad x_k &:=& \cos(\alpha_k), k=1, 2, \ldots, \frac{d-2}{2}.
\end{eqnarray}
Then, $1 \ge x_k \ge b > 0$. For the spectrum in Eq.~(\ref{eq:symmetric_spectrum_main}), the trace is real and equals
\begin{eqnarray}
\tr{(\hat{X}_{\star})} = 2\cos(\beta)+2\sum_{k=1}^{(d-2)/2}\cos(\alpha_k) = 2b+2\sum_{k=1}^{(d-2)/2}x_k.
\label{eq:trace_bx_main}
\end{eqnarray}
Since $P=\abs{\tr{(\hat{X}_{\star})}}$ and the right-hand side is nonnegative in the regime of interest, we may write
\begin{eqnarray}
P = 2b+2\sum_{k=1}^{(d-2)/2}x_k.
\label{eq:sum_xk_main}
\end{eqnarray}
Next, the squared trace is $(\tr{(\hat{X}_{\star})})^2=P^2$, and
\begin{eqnarray}
\tr{\hat{X}_{\star}^2}
&=&
2\cos(2\beta)+2\sum_{k=1}^{(d-2)/2}\cos(2\alpha_k)
\nonumber\\
&=&
2(2b^2-1)+2\sum_{k=1}^{(d-2)/2}(2x_k^2-1)
\nonumber\\
&=&
4b^2+4\sum_{k=1}^{(d-2)/2}x_k^2-d.
\label{eq:trace2_bx_main}
\end{eqnarray}
Therefore,
\begin{eqnarray}
\tr{\hat{X}_{\star}^2}+(\tr{\hat{X}_{\star}})^2
=
P^2+4b^2+4\sum_{k=1}^{(d-2)/2}x_k^2-d.
\label{eq:moment2_bx_main}
\end{eqnarray}
In the fault-tolerance regime of interest (near the identity), the quantity in \eqref{eq:moment2_bx_main} is nonnegative,
and for admissible $(F,D)$ arising from a unitary it is consistent to take
\begin{eqnarray}
Q = P^2+4b^2+4\sum_{k=1}^{(d-2)/2}x_k^2-d.
\label{eq:sum_xk2_main}
\end{eqnarray}

\medskip
\paragraph*{\bf Step 5: Cauchy-Schwarz reduction and an explicit lower bound on $b$.} By the Cauchy-Schwarz inequality applied to the vector $(x_1,x_2,\ldots,x_{(d-2)/2})$,
\begin{eqnarray}
\sum_{k=1}^{(d-2)/2}x_k^2 \ge \frac{2}{d-2}\left(\sum_{k=1}^{(d-2)/2}x_k\right)^2.
\label{eq:CS_xk_main}
\end{eqnarray}
By substituting Eq.~(\ref{eq:sum_xk_main}) and Eq.~(\ref{eq:sum_xk2_main}) into Eq.~(\ref{eq:CS_xk_main}), we obtain
\begin{eqnarray}
\frac{Q+d-P^2}{4}-b^2 \ge \frac{2}{d-2}\left(\frac{P}{2}-b\right)^2.
\label{eq:ineq_before_quadratic_main}
\end{eqnarray}
A direct rearrangement of Eq.~(\ref{eq:ineq_before_quadratic_main}) yields the quadratic inequality
\begin{eqnarray}
4d b^2 - 8P b + \bigl( dP^2 - d^2 + 2d-(d-2)Q \bigr) \le 0.
\label{eq:quadratic_ineq_b_main}
\end{eqnarray}
The discriminant of the corresponding quadratic polynomial equals
\begin{eqnarray}
\Delta = 16(d-2)\left(dQ+d^2-(d+2)P^2\right).
\label{eq:disc_main}
\end{eqnarray}
Hence the smaller root is
\begin{eqnarray}
b_{-} = \frac{8P-\sqrt{\Delta}}{8d} = \frac{P}{d} - \frac{1}{2d}\sqrt{(d-2)\left(dQ+d^2-(d+2)P^2\right)}.
\label{eq:bminus_main}
\end{eqnarray}
Since \eqref{eq:quadratic_ineq_b_main} holds only between the two roots, any feasible optimizer must satisfy $b\ge b_{-}$. Finally, since $m(\hat{X})=\cos(\beta)=b$ when $m(\hat{X})>0$, we conclude that
\begin{eqnarray}
m(\hat{X}) \ge [b_{-}]_+.
\label{eq:m_lower_bminus_main}
\end{eqnarray}
By combining Eq.~(\ref{eq:m_lower_bminus_main}) with Eq.~(\ref{eq:def_cFD}), we can prove Eq.~(\ref{eq:m_lower_cFD_main}). Eq.~(\ref{eq:diamond_upper_cFD_main}) follows immediately from {\bf Lemma~\ref{lem:diamond_unitary_overlap_main}}.

\medskip
\paragraph*{\bf Tightness.} The Cauchy-Schwarz inequality \eqref{eq:CS_xk_main} is saturated if and only if $x_1=x_2=\cdots=x_{(d-2)/2}$. Thus, the equality in Eq.~(\ref{eq:m_lower_cFD_main}) is achieved by a two-angle spectrum with a single bulk angle $\alpha$: the set
\begin{eqnarray}
b = [b_{-}]_+, \quad a = \frac{P-2b}{d-2}.
\end{eqnarray}
Define $\alpha=\arccos(a)$ and $\beta=\arccos(b)$, and let
\begin{eqnarray}
\hat{X}_{\star} = \mathrm{diag}(\underbrace{e^{i\alpha},e^{-i\alpha},\ldots,e^{i\alpha},e^{-i\alpha}}_{d-2},e^{i\beta},e^{-i\beta}).
\label{eq:Xstar_explicit_main}
\end{eqnarray}
Then, $\hat{X}_{\star} \in U(d)$ is consistent with the same $(P, Q)$, and hence the same $(F, D)$. Thus, its numerical range is the convex hull of points on the arc $[-\beta, \beta]$ with the extremal points $e^{i\beta}$ and $e^{-i\beta}$. Therefore, $m(\hat{X}_{\star})=[b]_+=c(F,D)$ and the bound is attained.
\end{proof}

{\bf Theorem~\ref{thm:main_d4}} should be read as a dimension-dependent ``moment upgrade'': for coherent unitary errors in any $d \ge 4$, the second-moment fluctuation information encoded in $D$ converts average benchmarking data into a sharp worst-case certificate. The bound is tight because it matches the exact solution of the constrained spectral extremal problem in Eq.~(\ref{eq:opt_problem_main}) under the moment constraints fixed by $(F,D)$.

\subsection{Physical interpretation: what $c(F,D)$ means and why $D$ is the right observable}\label{subsec:physical_interpretation_cFD}

The quantity $c(F,D)$ is not an algebraic complement; it has a clean operational and geometric meaning.

\begin{itemize}
\item {\bf $c(F,D)$ as a certified minimum overlap.} By construction, $m(\hat{X}) = \min_{\psi}\abs{\bra{\psi}\hat{X}\ket{\psi}}$ is the guaranteed overlap amplitude between an input and its image under the error unitary. Thus, $m(\hat{X})$ answers the worst-case question: is there an input state whose phase evolution under the implemented gate becomes almost orthogonal to the target evolution? {\bf Theorem~\ref{thm:main_d4}} states that from the identified $(F, D)$, we can certify that no input state suffers overlap below $c(F, D)$. Equivalently, $\sqrt{1 - c(F,D)^2}$ is the largest trace-distance distortion that is compatible with the measured moments.

\item {\bf $c(F,D)$ as the ``radius'' of the numerical range.} Since $\bra{\psi}\hat{X}\ket{\psi}$ ranges over the numerical range $W(\hat{X})$, $m(\hat{X})$ is the distance from the origin to $W(\hat{X})$. Hence, $c(F,D)$ is a certified inner radius of $W(\hat{X})$ obtained by the Haar-moment data. In the coherent regime, this is precisely the missing geometric control that $F$ alone cannot provide: $F$ fixes only a coarse second-moment projection (essentially $\abs{\tr{(\hat{X})}}$), while $D$ injects the fourth-moment information needed to constrain the shape of $W(\hat{X})$.

\item {\bf Why $D$ is indispensable when $u=1$.} The unitarity is an excellent measure at detecting departures from coherence: it measures how much the channel deviates from being unitary. But in the coherent-dominated regime, every unitary error has $u=1$, so the unitarity cannot resolve which coherent error one has. This is not merely a lack of improvement: in that regime, the unitarity-assisted bound reduces to Eq.~(\ref{eq:kueng_u1_loose}), and even can be looser than the fidelity-only bound. By contrast, $D$ measures fluctuations of $f(\psi)=\abs{\bra{\psi}\hat{X}\ket{\psi}}^2$ across $\psi$, and therefore remains sensitive even when $u$ is maximally saturated. In short, the unitarity tells us that we are in the coherent regime; whereas, the fidelity deviation tells us how dangerous that the coherence actually is.

\item {\bf ``Same experiment, more truth.''} A final practical point is that $D$ is not an exotic new primitive requiring full tomography. It is a moment of the same state-dependent fidelity distribution that underlies $F$. Thus, once one is already investing experimental effort to estimate $F$ accurately, the additional information carried by $D$ can be harvested from the same sampling process. From a fault-tolerance standpoint, this is precisely the upgrade one wants: a strictly sharper worst-case certification, purchased not by a leap to tomography, but by extracting a higher moment from data that experiments already generate.
\end{itemize}

Taken all together, {\bf Theorem~\ref{thm:main_d4}} and its interpretation provide a concrete resolution of the coherent-error dilemma: when coherence dominates, $F$ alone is a dangerously lossy proxy and $u$ saturates to a constant, but the pair $(F, D)$ retains the spectral information required to \emph{sharpen} worst-case error assessment for $n \ge 2$ qubit gates, which are the critical entangling primitives of any universal gate set. 

\section{Examples: two-qubit CZ gate, Toffoli gate, and quantum Fourier transform}

We illustrate {\bf Theorem~\ref{thm:main_d4}} on an experimentally relevant coherent error model: a controlled-phase (CZ-like) gate, a decomposed Toffoli gate, and $n$-qubit QFT circuit.

\subsection{Example $1$: a CZ-like coherent phase error}\label{subsec:cz_example}

Firstly, we consider a controlled-phase (CZ-like) gate with a systematic phase miscalibration on the $\ket{11}$ component.  The model cleanly exhibits the $\sqrt{r}$ coherent amplification in $d_{\diamond}$, the saturation $u=1$, and the sharpening obtained by incorporating $D$. In addition, this example provides an end-to-end demonstration of how the direct $(F, D)$ estimation protocol (Sec.~\ref{sec:FD_measurement}) produces finite-sample data that can be inserted into the bounds. In Fig.~\ref{fig:cz_comparison}, we overlay protocol-based point estimates of the $F$-only and $(F, D)$-assisted bounds on top of the corresponding theoretical curves.

\begin{figure}[t]
  \centering
  \includegraphics[width=0.90\linewidth]{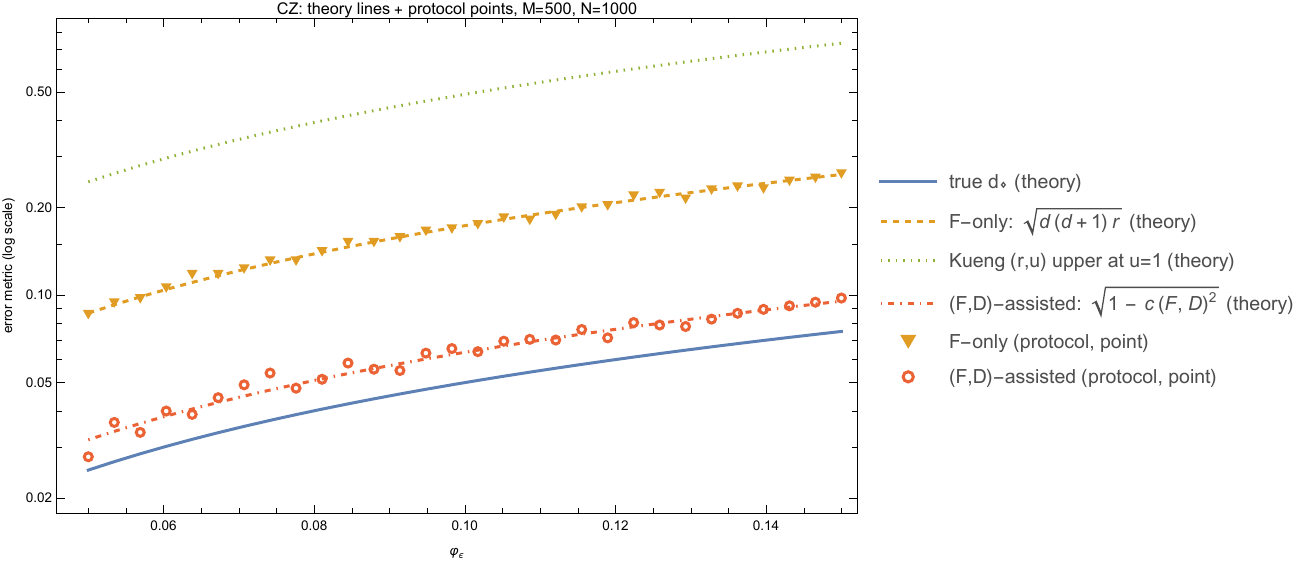}
  \caption{\textbf{CZ-like coherent phase error: exact diamond distance vs.\ upper bounds, with protocol-based estimates.} We consider the unitary error $\hat{X}=\mathrm{diag}(1,1,1,e^{i\phi_{\epsilon}})$ arising from a coherent phase miscalibration of a controlled-phase gate. The solid curve shows the exact worst-case error $d_{\diamond}(\mathcal{X},\mathcal{I})=\abs{\sin(\phi_{\epsilon}/2)}$ from Eq.~(\ref{eq:cz_diamond_exact_main}). The dashed curve shows the fidelity-only bound $\sqrt{20r}$ at $d=4$. The dotted curve shows Kueng's $(r,u)$-based upper bound in Eq.~(\ref{eq:kueng_ru_bound_recalled_main}) evaluated at $u=1$, which is provably looser than $\sqrt{20r}$ by the factor $2\sqrt{2}$ in this coherent regime. The dash-dotted curve shows our moment-assisted bound $\sqrt{1-c(F,D)^2}$ from {\bf Theorem~\ref{thm:main_d4}}, which substantially sharpens worst-case certification by exploiting the fluctuation information encoded in $D$. In addition, the markers overlay finite-sample, protocol-based estimates obtained by simulating the direct $(F, D)$ estimation procedure of Sec.~\ref{sec:FD_measurement}: for each $\phi_{\epsilon}$ we estimate $(\hat{F}, \hat{D})$ from $M=500$ random input states and $N=1000$ projective-test shots per state, and then evaluate the $F$-only and $(F,D)$-assisted bounds at $(\hat{F},\hat{D})$ as in Eq.~(\ref{eq:cz_protocol_bounds_main}). The vertical axis is shown on a logarithmic scale to highlight the separation between average-case and worst-case estimates in the small-error regime relevant to fault tolerance.}
  \label{fig:cz_comparison}
\end{figure}

\medskip
\paragraph*{\bf Model.}---Let us consider
\begin{eqnarray}
U_{\mathrm{ideal}} = \mathrm{diag}(1,1,1,e^{i\phi}), \quad U_{\mathrm{exp}} = \mathrm{diag}(1,1,1,e^{i(\phi+\phi_{\epsilon})}),
\label{eq:cz_model_unitaries_main}
\end{eqnarray}
where $\phi_{\epsilon}$ is a coherent over-rotation of the controlled phase.  The effective error unitary is
\begin{eqnarray}
\hat{X} = \hat{U}_{\mathrm{ideal}}^{\dagger} \hat{U}_{\mathrm{exp}} = \mathrm{diag}(1,1,1,e^{i\phi_{\epsilon}}).
\label{eq:cz_error_unitary_main}
\end{eqnarray}

\medskip
\paragraph*{\bf Closed forms for $(F,D)$ from the trace-moment formulas.}---We compute $(F,D)$ by inserting the traces of Eq.~(\ref{eq:cz_error_unitary_main}) into {\bf Proposition~\ref{prop:FD_closed_form_unitary}}. First,
\begin{eqnarray}
\tr{(\hat{X})} = 3+e^{i\phi_{\epsilon}}, \quad \abs{\tr{(\hat{X})}}^2 = 16-12\sin^2\left(\frac{\phi_{\epsilon}}{2}\right).
\label{eq:cz_trace_and_norm_main}
\end{eqnarray}
Using $F=\frac{4+\abs{\tr{(\hat{X})}^2}}{20}$ at $d=4$, we obtain
\begin{eqnarray}
F(\phi_{\epsilon}) = 1-\frac{3}{5}\sin^2\left(\frac{\phi_{\epsilon}}{2}\right),
\quad
r(\phi_{\epsilon}) = 1 - F(\phi_{\epsilon}) = \frac{3}{5}\sin^2\left(\frac{\phi_{\epsilon}}{2}\right).
\label{eq:cz_F_and_r_closed_main}
\end{eqnarray}
For $D$, it is simplest to compute the second moment $E_2:=D^2+F^2$ using Eq.~(\ref{eq:E2_d4_main}):
\begin{eqnarray}
E_2 = \frac{56+24P^2+Q^2}{840},
\quad
P = \abs{\tr{(\hat{X})}},
\quad
Q = \abs{\tr{(\hat{X}^2)} + \bigl(\tr{(\hat{X})}\bigr)^2}.
\label{eq:E2_in_PQ_main}
\end{eqnarray}
For Eq.~(\ref{eq:cz_error_unitary_main}), we have
\begin{eqnarray}
\tr{(\hat{X}^2)} = 3+e^{i2\phi_{\epsilon}}.
\label{eq:cz_trace2_main}
\end{eqnarray}
A direct evaluation of the full fourth-moment numerator in {\bf Proposition~\ref{prop:FD_closed_form_unitary}} at $d=4$ simplifies remarkably when expressed in $x:=\sin^2(\phi_{\epsilon}/2)$.
Using $\cos(\phi_{\epsilon})=1-2x$ and $\cos(2\phi_{\epsilon})=1-8x+8x^2$, one finds
\begin{eqnarray}
\int d\psi f(\psi)^2 = 1-\frac{6}{5}x+\frac{16}{35}x^2,
\quad
F^2 = \left(1-\frac{3}{5}x\right)^2 = 1-\frac{6}{5}x+\frac{9}{25}x^2.
\label{eq:cz_second_moment_simplified_main}
\end{eqnarray}
Therefore, we attain
\begin{eqnarray}
D(\phi_{\epsilon}) = \frac{1}{5}\sqrt{\frac{17}{7}}\sin^2\left(\frac{\phi_{\epsilon}}{2}\right).
\label{eq:cz_D_closed_main}
\end{eqnarray}
In particular, as $\phi_{\epsilon}\to 0$, the average infidelity scales as $r(\phi_{\epsilon})\sim \frac{3}{20}\phi_{\epsilon}^2$, while $D(\phi_{\epsilon})\sim \frac{1}{20}\sqrt{\frac{17}{7}}\phi_{\epsilon}^2$, so both are quadratic in the coherent miscalibration angle (but encode different spectral information).

\medskip
\paragraph*{\bf Exact worst-case error.}---Since the error is unitary, Lemma~\ref{lem:diamond_unitary_overlap_main} applies. Here, the numerical range of $\hat{X}=\mathrm{diag}(1,1,1,e^{i\phi_{\epsilon}})$ is the convex hull of $\{1,e^{i\phi_{\epsilon}}\}$ (the extra $1$'s do not change the hull), namely the chord between $1$ and $e^{i\phi_{\epsilon}}$. The distance from the origin to this chord is $\cos(\phi_{\epsilon}/2)$, and hence, we have
\begin{eqnarray}
m(\hat{X}) = \cos\left(\frac{\phi_{\epsilon}}{2}\right),
\quad
d_{\diamond}(\mathcal{X},\mathcal{I}) = \sqrt{1-m(\hat{X})^2} = \left| \sin\left(\frac{\phi_{\epsilon}}{2}\right) \right|.
\label{eq:cz_diamond_exact_main}
\end{eqnarray}
By combining Eq.~(\ref{eq:cz_F_and_r_closed_main}) and Eq.~(\ref{eq:cz_diamond_exact_main}), we find the canonical coherent scaling
\begin{eqnarray}
d_{\diamond}(\mathcal{X},\mathcal{I}) = \sqrt{\frac{5}{3}r(\phi_{\epsilon})},
\label{eq:cz_diamond_vs_r_main}
\end{eqnarray}
which makes explicit how a very small average infidelity can still correspond to a worst-case error of order $\sqrt{r}$.

\medskip
\paragraph*{\bf Three driven upper bounds.}
We compare three standard ways to upper bound $d_{\diamond}$ from experimentally accessible data:

\emph{Fidelity-only conversion.} For $d=4$, the standard conversion yields
\begin{eqnarray}
d_{\diamond}(\mathcal{X},\mathcal{I}) \le \sqrt{d(d+1)r} = \sqrt{20r}.
\label{eq:cz_bound_Fonly_main}
\end{eqnarray}

\emph{Kueng's $(r,u)$ bound and the $u=1$ pitfall.} For unital noise with no leakage, Eq.~(\ref{eq:kueng_ru_bound_recalled_main}) gives an upper bound in terms of $(r,u)$. In the present coherent model $u=1$ exactly, so it collapses to
\begin{eqnarray}
d_{\diamond}(\mathcal{X},\mathcal{I}) \le \frac{d}{\sqrt{2}}\sqrt{d(d+1)r} = 2\sqrt{2}\sqrt{20r},
\label{eq:cz_bound_kueng_u1_main}
\end{eqnarray}
which is uniformly looser than \eqref{eq:cz_bound_Fonly_main} by the factor $2\sqrt{2}$.

\emph{Our $(F,D)$-assisted bound.}
Finally, {\bf Theorem~\ref{thm:main_d4}} yields
\begin{eqnarray}
d_{\diamond}(\mathcal{X},\mathcal{I}) \le \sqrt{1-c(F,D)^2},
\label{eq:cz_bound_cFD_main}
\end{eqnarray}
where $c(F, D)$ is computed from the observed $(F, D)$ via Eq.~(\ref{eq:def_cFD}) and {\bf Lemma~\ref{lem:PQ_from_FD_main}}. In contrast to the unitarity $u$, the deviation $D$ remains informative on the coherent manifold, and in this model it substantially tightens the certification of the worst-case error (see Fig.~\ref{fig:cz_comparison}).

\medskip
\paragraph*{\bf Protocol-based estimation of $(F,D)$ and data-driven bounds.}---We now connect the bounds above to the direct estimation protocol of Sec.~\ref{sec:FD_measurement}. For each random trial $i=1, \ldots, M=500$, sample a random pure state $\ket{\psi_i}$ (e.g., by preparing $\ket{\psi_i}=\hat{V}_i\ket{0}$ with $\hat{V}_i$ drawn from a unitary $4$-design), apply the interaction-picture error channel $\mathcal{X}(\hat{\rho})=\hat{X}\hat{\rho}\hat{X}^{\dagger}$, and perform a projective test onto $\ket{\psi_i}$.  The corresponding survival probability is
\begin{eqnarray}
f_i = f_{\mathcal{X}}(\psi_i) = \abs{\bra{\psi_i}\hat{X}\ket{\psi_i}}^2.
\label{eq:cz_protocol_fi_main}
\end{eqnarray}
Using $N=1000$ repeated shots for each $\ket{\psi_i}$, let $K_i$ be the number of ``pass'' outcomes. Conditioned on $f_i$, we have the binomial model
\begin{eqnarray}
K_i \sim \mathrm{Bin}(N,f_i),
\quad
\hat{f}_i = \frac{K_i}{N}.
\label{eq:cz_protocol_binomial_main}
\end{eqnarray}
Then $\hat{f}_i$ is an unbiased estimator of $f_i$:
\begin{eqnarray}
\mathbb{E}[\hat{f}_i|f_i] = \frac{1}{N}\mathbb{E}[K_i|f_i] = \frac{1}{N}Nf_i = f_i.
\label{eq:cz_protocol_f_unbiased_main}
\end{eqnarray}
To estimate the second moment without bias from shot noise, we use the factorial-moment identity $\mathbb{E}[K_i(K_i-1)|f_i]=N(N-1)f_i^2$, which follows by writing $K_i=\sum_{j=1}^N Z_{ij}$ with i.i.d. Bernoulli$(f_i)$ variables $Z_{ij}$ and expanding $K_i(K_i-1)=\sum_{j\neq k} Z_{ij}Z_{ik}$.  This gives an unbiased estimator of $f_i^2$:
\begin{eqnarray}
\widehat{f_i^2} = \frac{K_i(K_i-1)}{N(N-1)} = \frac{N\hat{f}_i^2-\hat{f}_i}{N-1},
\quad
\mathbb{E}[\widehat{f_i^2} \;|\; f_i] = f_i^2.
\label{eq:cz_protocol_f2_unbiased_main}
\end{eqnarray}
The Haar moments $F=\int d\psi f(\psi)$ and $E_2=\int d\psi f(\psi)^2$ are then estimated by
\begin{eqnarray}
\hat{F} = \frac{1}{M}\sum_{i=1}^M \hat{f}_i,
\quad
\widehat{E_2} = \frac{1}{M}\sum_{i=1}^M \widehat{f_i^2}.
\label{eq:cz_protocol_F_E2_main}
\end{eqnarray}
To remove the finite-$M$ bias in estimating $F^2$, we use the cross-average
\begin{eqnarray}
\widehat{F^2} = \frac{1}{M(M-1)}\sum_{i\neq j}\hat{f}_i\hat{f}_j = \frac{\left(\sum_{i=1}^M \hat{f}_i\right)^2-\sum_{i=1}^M \hat{f}_i^2}{M(M-1)}.
\label{eq:cz_protocol_F2_main}
\end{eqnarray}
Conditioned on the sampled states, the outcomes for different $i$ are independent, hence $\mathbb{E}[\hat{f}_i\hat{f}_j \;|\; \psi_i, \psi_j]=f_i f_j$ for $i\neq j$, and averaging over the random draws yields $\mathbb{E}[\widehat{F^2}]=F^2$.  Therefore, the shot-noise-corrected estimator
\begin{eqnarray}
\widehat{D^2} = \widehat{E_2} - \widehat{F^2}
\label{eq:cz_protocol_D2_main}
\end{eqnarray}
is unbiased for $D^2=E_2-F^2$ in the idealized model where $\ket{\psi_i}$ are exact $4$-design samples and each projective test implements $\ketbra{\psi_i}{\psi_i}$ exactly.  In finite data, $\widehat{D^2}$ can fluctuate slightly negative, in which case we report
\begin{eqnarray}
\hat{D} = \sqrt{\max\{\widehat{D^2},0\}}.
\label{eq:cz_protocol_Dhat_main}
\end{eqnarray}
Finally, from the same dataset we obtain data-driven certificates by evaluating the bounds at the estimated moments:
\begin{eqnarray} 
\hat{r} = 1 - \hat{F},
\quad
d_{\diamond}(\mathcal{X},\mathcal{I})
\le
\sqrt{20\hat{r}},
\qquad
d_{\diamond}(\mathcal{X},\mathcal{I}) \le \sqrt{1-c(\hat{F},\hat{D})^2}.
\label{eq:cz_protocol_bounds_main}
\end{eqnarray}
These are the two sets of protocol-based points overlaid in Fig.~\ref{fig:cz_comparison}.

\medskip
\paragraph*{\bf Numerical comparison.}---Fig.~\ref{fig:cz_comparison} compares the exact value in Eq.~(\ref{eq:cz_diamond_exact_main}) against the three bounds in Eq.~(\ref{eq:cz_bound_Fonly_main}), Eq.~(\ref{eq:cz_bound_kueng_u1_main}), and Eq.~(\ref{eq:cz_bound_cFD_main}) over a range of $\phi_{\epsilon}$. The separation between the average-case and worst-case error is most pronounced in the small-error regime relevant to fault tolerance. In that regime, incorporating $D$ allows the $(F,D)$-assisted bound to track the exact coherent worst-case scaling far more closely than bounds derived from $F$ alone, and dramatically improves over the unitarity-assisted bound which is forced to saturate at $u=1$. Moreover, the protocol-based point estimates computed from finite samples of $\hat{F}$ and $\hat{D}$ concentrate around the corresponding theoretical curves, directly illustrating how the proposed moment-assisted certification can be realized from experimentally accessible survival-probability data.

{\bf Theorem~\ref{thm:main_d4}} and the worked example above establish the main conceptual point: in the coherent regime, the worst-case assessment hinges on the spectral geometry (via $m(\hat{X})$), and the fidelity deviation $D$ supplies exactly the additional experimentally accessible moment information needed to tightly control that geometry for two-qubit gates.

\subsection{Example $2$: a decomposed Toffoli gate with coherent over-rotation}\label{subsec:toffoli_example}

We now present a three-qubit worked example showing that the same coherent worst-case certification phenomenon persists beyond two-qubit gates. In particular, we consider a standard Clifford+$T$ decomposition of the Toffoli gate~\cite{Selinger2013,amy2013meet} and assume a single coherent over-rotation parameter affects \emph{every} primitive in the decomposition (all CNOTs and all single-qubit gates). The resulting error channel is unitary, hence unitarity saturates at $u=1$, and Kueng's $(r,u)$-based worst-case bound becomes dimension-amplified and extremely loose. In contrast, the $(F, D)$-assisted bound based on $c(F,D)$ remains informative and tracks the true diamond distance far more closely.

\begin{figure}[t]
  \centering
  \includegraphics[width=0.90\linewidth]{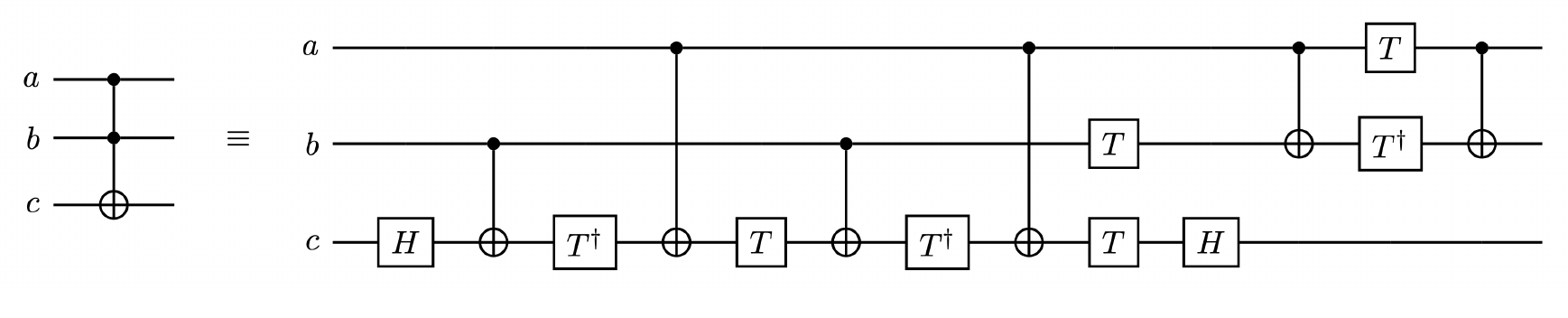}
  \caption{\textbf{A standard decomposition of Toffoli gate.} $6$ CNOTs, $2$ Hadamard, and $7$ $\hat{T}$ (or $\hat{T}^\dagger$)~\cite{amy2013meet}.}
  \label{fig:Toffoli_decomp}
\end{figure}

\medskip
\paragraph*{\bf Ideal gate and a standard decomposition.}---Let qubits $1,2$ be controls and qubit $3$ be the target.  Let $\hat{U}_{\mathrm{Tof}}$ denote the ideal Toffoli (CCNOT) unitary. We use the following well-known decomposition into $6$ CNOTs and single-qubit Clifford+$T$ gates~\cite{fowler2004constructing}:
\begin{eqnarray}
\hat{U}_{\mathrm{Tof}} &=& \hat{\mathrm{CNOT}}_{1\rightarrow 2} \hat{T}_{2}^{\dagger} \hat{T}_{1} \hat{H}_{3} \hat{\mathrm{CNOT}}_{1\rightarrow 2} \hat{T}_{3} \hat{T}_{2} \nonumber\\
&& \hat{\mathrm{CNOT}}_{1\rightarrow 3} \hat{T}_{3}^{\dagger} \hat{\mathrm{CNOT}}_{2\rightarrow 3} \hat{T}_{3} \hat{\mathrm{CNOT}}_{1\rightarrow 3} \hat{T}_{3}^{\dagger} \hat{\mathrm{CNOT}}_{2\rightarrow 3} \hat{H}_{3}.
\label{eq:toffoli_decomposition_ideal}
\end{eqnarray}
Here $\hat{H}$ is the Hadamard gate and $\hat{T}=\mathrm{diag}(1,e^{i\pi/4})$. See Fig.~\ref{fig:Toffoli_decomp}.

\medskip
\paragraph*{\bf Uniform coherent over-rotation model on all primitives.}---We assume that each primitive gate in Eq.~(\ref{eq:toffoli_decomposition_ideal}) is implemented with the same coherent over-rotation parameter $\epsilon$.

For the single-qubit primitives, we model a systematic angle miscalibration along the gate's native axis:
\begin{eqnarray}
\hat{T}^{(\epsilon)} &:=& \exp\left(-i\epsilon\frac{\hat{\sigma}_{z}}{2}\right)\hat{T},
\label{eq:T_overrot_model}
\\
(\hat{T}^{\dagger})^{(\epsilon)} &:=& \exp\left(-i\epsilon\frac{\hat{\sigma}_{z}}{2}\right)\hat{T}^{\dagger},
\label{eq:Tdg_overrot_model}
\\
\hat{H}^{(\epsilon)} &:=& \exp\left(-i\epsilon\frac{\hat{H}}{2}\right)\hat{H}.
\label{eq:H_overrot_model}
\end{eqnarray}
The choice of Eq.~(\ref{eq:H_overrot_model}) is natural because $\hat{H}$ is Hermitian and unitary, so it defines a rotation axis in $SU(2)$. This makes the coherent perturbation fully unitary while keeping $\hat{H}^{(\epsilon)} \to \hat{H}$ as $\epsilon\to 0$.

For each CNOT, we use an over-rotation model that is also unitary and experimentally motivated as a coherent calibration error on the controlled-$X$ generator:
\begin{eqnarray}
\hat{\Pi}_{1} &:=& \ket{1}\bra{1} = \frac{1}{2}\left(\hat{I}-\hat{\sigma}_{z}\right),
\label{eq:proj1_def_toffoli}
\\
\hat{\mathrm{CNOT}}_{c \rightarrow t}^{(\epsilon)} &:=& \exp\left(-i\epsilon\hat{\Pi}_{1}^{(c)}\otimes \hat{\sigma}_{x}^{(t)}\right)\hat{\mathrm{CNOT}}_{c\rightarrow t}.
\label{eq:CNOT_overrot_model}
\end{eqnarray}
By construction, all primitives remain unitary for all $\epsilon$.

Let $\hat{U}_{\mathrm{exp}}(\epsilon)$ be the implemented three-qubit unitary obtained by replacing every primitive in Eq.~(\ref{eq:toffoli_decomposition_ideal}) with its over-rotated version, i.e., Eqs.~(\ref{eq:T_overrot_model})--(\ref{eq:CNOT_overrot_model})). Define the effective error unitary
\begin{eqnarray}
\hat{X}(\epsilon) = \hat{U}_{\mathrm{Tof}}^{\dagger}\hat{U}_{\mathrm{exp}}(\epsilon).
\label{eq:toffoli_error_unitary}
\end{eqnarray}
Then, $\hat{X}(\epsilon) \in U(8)$ and the corresponding error channel $\mathcal{X}(\hat{\rho})=\hat{X}(\epsilon)\hat{\rho}\hat{X}(\epsilon)^{\dagger}$ is fully coherent. In particular, its unitarity satisfies $u(\mathcal{X})=1$ exactly for all $\epsilon$.

\medskip
\paragraph*{\bf Computing $(F,D)$ for $d=8$ from trace moments.} Since the error is unitary, {\bf Proposition~\ref{prop:FD_closed_form_unitary}} applies at $d=8$. Firstly, the average fidelity is given by
\begin{eqnarray}
F(\epsilon) = \frac{8+\abs{\tr{(\hat{X}(\epsilon))}}^{2}}{72},
\quad
r(\epsilon) = 1-F(\epsilon).
\label{eq:toffoli_F_r_d8}
\end{eqnarray}
The second moment $E_{2}:=D^{2}+F^{2}$ is
\begin{eqnarray}
E_{2}(\epsilon) = \frac{176 + 40\abs{\tr{(\hat{X}(\epsilon))}}^{2} + \abs{\tr{(\hat{X}(\epsilon)^{2})}}^{2} + \abs{\tr{(\hat{X}(\epsilon))}}^{4} + 2\mathrm{Re}\left(\tr{(\hat{X}(\epsilon)^{2})}\left(\tr{(\hat{X}(\epsilon))}^{\ast}\right)^{2}\right)}{7920},
\label{eq:toffoli_E2_d8}
\end{eqnarray}
and hence
\begin{eqnarray}
D(\epsilon)^{2} = E_{2}(\epsilon) - F(\epsilon)^{2}.
\label{eq:toffoli_D_d8_from_E2}
\end{eqnarray}
In this decomposed model, $\tr{(\hat{X}(\epsilon))}$ and $\tr{(\hat{X}(\epsilon)^{2})}$ can be evaluated from the explicit matrix product in Eq.~(\ref{eq:toffoli_error_unitary}).

\medskip
\paragraph*{\bf Exact worst-case error for the three-qubit coherent channel.}---Since $\hat{X}(\epsilon)$ is unitary, {\bf Lemma~\ref{lem:diamond_unitary_overlap_main}} gives
\begin{eqnarray}
d_{\diamond}(\mathcal{X},\mathcal{I}) = \sqrt{1-m(\hat{X}(\epsilon))^{2}}.
\label{eq:toffoli_diamond_exact_from_m}
\end{eqnarray}
Moreover, because $\hat{X}(\epsilon)$ is normal, its numerical range equals the convex hull of its eigenvalues. Writing $\mathrm{spec}(\hat{X}(\epsilon))=\{\lambda_{j}(\epsilon)\}_{j=1}^{8}$ with $\abs{\lambda_{j}(\epsilon)}=1$, we have
\begin{eqnarray}
m(\hat{X}(\epsilon)) = \min_{z\in \mathrm{conv}\{\lambda_{j}(\epsilon)\}}\abs{z}.
\label{eq:toffoli_m_from_convex_hull}
\end{eqnarray}
Eq.~(\ref{eq:toffoli_m_from_convex_hull}) yields the true diamond distance by a purely geometric computation in the complex plane (convex hull followed by distance-to-origin), which we implement numerically.

\medskip
\paragraph*{\bf Three driven upper bounds.}---We compare three ways to upper bound $d_{\diamond}(\mathcal{X},\mathcal{I})$.

\emph{Fidelity-only conversion.}
The standard conversion gives
\begin{eqnarray}
d_{\diamond}(\mathcal{X},\mathcal{I}) \le \sqrt{d(d+1)r} = \sqrt{72r}.
\label{eq:toffoli_bound_Fonly}
\end{eqnarray}

\emph{Kueng's $(r, u)$ bound and the $u=1$ pitfall.}
Because the error is unitary, $u=1$ and the $(r,u)$-based bound collapses to the dimension-amplified form
\begin{eqnarray}
d_{\diamond}(\mathcal{X},\mathcal{I}) \le \frac{d}{\sqrt{2}}\sqrt{d(d+1)r} = \frac{8}{\sqrt{2}}\sqrt{72r},
\label{eq:toffoli_bound_kueng_u1}
\end{eqnarray}
which is uniformly looser than \eqref{eq:toffoli_bound_Fonly} by the factor $8/\sqrt{2}$ in this coherent regime.

\begin{figure}[t]
  \centering
  \includegraphics[width=0.90\linewidth]{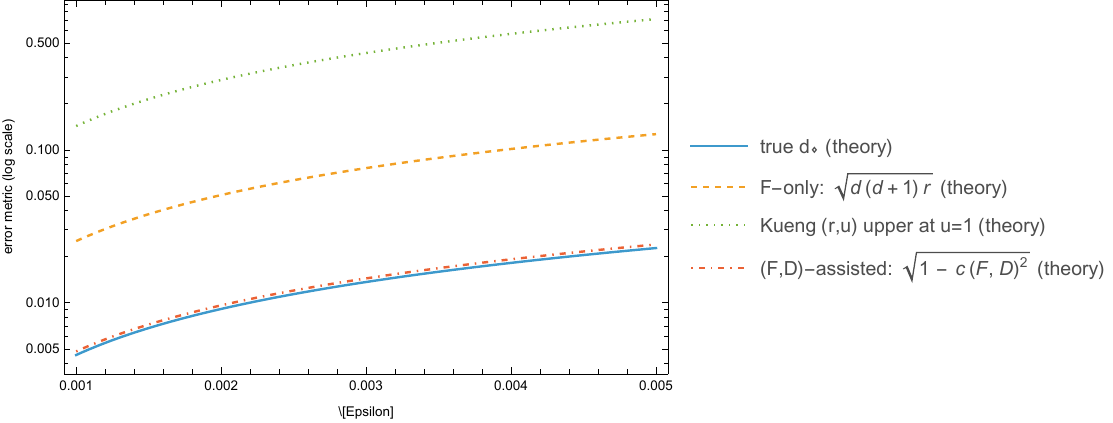}
  \caption{\textbf{Decomposed Toffoli with coherent primitive over-rotations: true diamond distance vs. upper bounds.} We consider a three-qubit Toffoli gate implemented by the standard $6$-CNOT Clifford+$T$ decomposition as in Eq.~(\ref{eq:toffoli_decomposition_ideal}) and/or Fig.~\ref{fig:Toffoli_decomp}. Every primitive (all CNOTs and all single-qubit gates) is assumed to carry the same coherent over-rotation parameter $\epsilon$ according to Eqs.~(\ref{eq:T_overrot_model})--(\ref{eq:CNOT_overrot_model}), yielding a unitary error channel with $u=1$. The solid curve shows the true diamond distance computed from the convex-hull characterization in Eq.~(\ref{eq:toffoli_m_from_convex_hull}). The dashed curve shows the fidelity-only bound $\sqrt{72r}$. The dotted curve shows Kueng's $(r, u)$-based upper bound evaluated at $u=1$, which is amplified by the factor $8/\sqrt{2}$ and becomes extremely loose in this coherent regime. The dash-dotted curve shows our $(F,D)$-assisted bound $\sqrt{1-c_{8}(F,D)^{2}}$, which uses the fluctuation information encoded in $D$ and yields substantially sharper worst-case certification. The vertical axis is logarithmic to emphasize separation in the low-error regime.}
  \label{fig:toffoli_comparison}
\end{figure}

\emph{$(F, D)$-assisted bound at $d=8$.} 
Using the $d\ge 4$ moment-assisted framework, we compute $P,Q$ from the observed $(F,D)$ and obtain a certified overlap $c(F,D)$, which yields $d_{\diamond}\le \sqrt{1-c(F,D)^{2}}$. For convenience, we record the $d=8$ specialization explicitly. First, $P$ and $Q$ are determined by $(F,D)$ via
\begin{eqnarray}
P^{2} &=& 72F-8,
\label{eq:toffoli_P2_from_F}
\\
Q^{2} &=& 7920(D^{2}+F^{2}) - 176 - 40P^{2}.
\label{eq:toffoli_Q2_from_FD}
\end{eqnarray}
Then, we can write
\begin{eqnarray}
c_{8}(F,D) = \Biggl[ \frac{P}{8} - \frac{1}{16} \sqrt{6\left(8Q+64-10P^{2}\right)} \Biggr]_{+}.
\label{eq:def_cFD_d8}
\end{eqnarray}
The $(F,D)$-assisted bound is
\begin{eqnarray}
d_{\diamond}(\mathcal{X},\mathcal{I}) \le \sqrt{1-c_{8}(F,D)^{2}}.
\label{eq:toffoli_bound_cFD}
\end{eqnarray}
In contrast to the unitarity $u$, the deviation $D$ remains nontrivial on the coherent manifold and thus provides additional spectral information beyond $F$ even when $u=1$.

\medskip
\paragraph*{\bf Numerical comparison.}---Fig.~\ref{fig:toffoli_comparison} compares the true value obtained from Eqs.~(\ref{eq:toffoli_diamond_exact_from_m})--(\ref{eq:toffoli_m_from_convex_hull}) against the three bounds in Eq.~(\ref{eq:toffoli_bound_Fonly}), Eq.~(\ref{eq:toffoli_bound_kueng_u1}), and Eq.~(\ref{eq:toffoli_bound_cFD}). As in the CZ example, the small-error regime exhibits the coherent $\sqrt{r}$ gap between average- and worst-case behavior. However, for a three-qubit gate the $u=1$ saturation pitfall is even more severe because the dimension factor in Eq.~(\ref{eq:toffoli_bound_kueng_u1}) grows with $d$. In this regime, incorporating $D$ allows the $(F,D)$-assisted bound to track the true worst-case scaling far more closely than $F$-only conversion and vastly improves over the $(r,u)$ bound which collapses at $u=1$.

\subsection{Example $3$: a coherent over-rotation model for the $n$-qubit QFT circuit}\label{subsec:qft_example}

We extend the worked-example to a genuinely multi-qubit unitary of practical interest: the $n$-qubit quantum Fourier transform (QFT). The consideration is twofold. First, the QFT circuit has a gate count that grows quadratically in $n$, so coherent miscalibrations can accumulate across many locations. Second, since the overall error is unitary, the unitarity satisfies $u=1$ and the $(r,u)$-based refinement loses the resolving power; this is precisely the regime where $D$ retains the nontrivial spectral information.

\begin{figure}[t]
  \centering
  \includegraphics[width=0.90\linewidth]{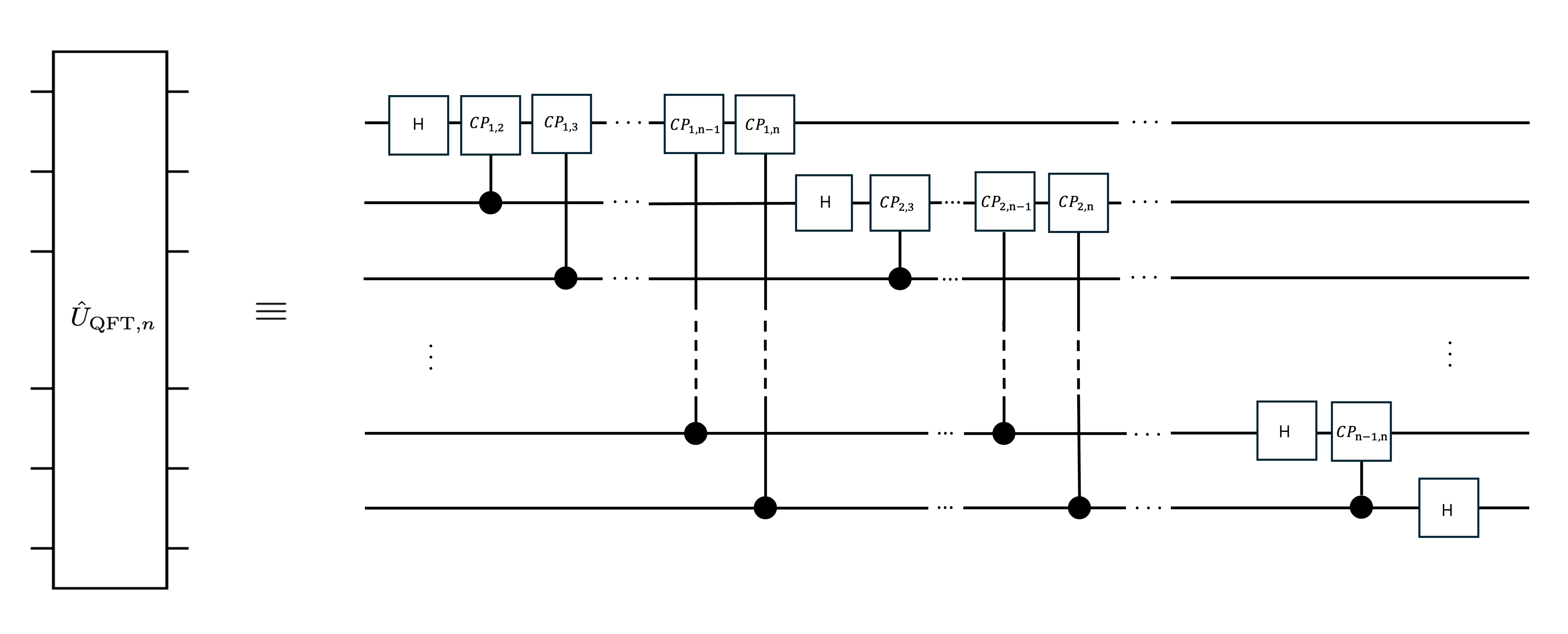}
  \caption{\textbf{A standard quantum Fourier transform circuit.}}
  \label{fig:QFT_circuit}
\end{figure}

\medskip
\paragraph*{\bf Ideal QFT circuit and a uniform coherent over-rotation model.}---Let $n\ge2$ and $d=2^n$.
We consider the standard QFT decomposition into Hadamard and controlled-phase gates (we omit the final bit-reversal swaps, which only permute computational basis labels and do not change the coherent-scaling mechanism we highlight).
We write the ideal QFT unitary as a product of elementary gates in chronological order (the first gate acts first):
\begin{eqnarray}
\hat{U}_{\mathrm{QFT},n} = \hat{G}_L \cdots \hat{G}_2 \hat{G}_1, \quad d=2^n,
\label{eq:qft_ideal_gate_list}
\end{eqnarray}
where the gate sequence is
\begin{eqnarray}
(\hat{G}_1,\ldots,\hat{G}_L) = \left( \hat{H}_1, \hat{CP}_{2,1}(\theta_{2,1}), \ldots, \hat{CP}_{n,1}(\theta_{n,1}), \hat{H}_2, \hat{CP}_{3,2}(\theta_{3,2}), \ldots, \hat{H}_n \right),
\label{eq:qft_gate_sequence}
\end{eqnarray}
with
\begin{eqnarray}
\theta_{k,j}=\frac{\pi}{2^{k-j}}, \quad 1\le j < k \le n.
\label{eq:qft_angles}
\end{eqnarray}
Here, $\hat{H}_j$ denotes the Hadamard on qubit $j$, and $\hat{CP}_{k,j}(\theta)$ is the controlled-phase gate acting on the two-qubit subspace of qubits $(k, j)$:
\begin{eqnarray}
\hat{CP}_{k,j}(\theta) = \hat{\mathds{1}}+(e^{i\theta}-1)\ketbra{11}{11}_{k,j}, \quad \ketbra{11}{11}_{k,j} = \ketbra{1}{1}_k\otimes\ketbra{1}{1}_j.
\label{eq:qft_cp_def}
\end{eqnarray}

Then, we assume a uniform coherent miscalibration parameter $\epsilon \in \mathbb{R}$ affecting every primitive in the same fractional way. Concretely, we model an over-rotation of the Hadamard about its own Hermitian axis and an over-rotation of each controlled phase about the projector $\ketbra{11}{11}$:
\begin{eqnarray}
\hat{H}_j^{(\epsilon)} &=& \exp\left(-i\frac{\epsilon}{2}\hat{H}_j\right)\hat{H}_j, \nonumber \\
\hat{CP}_{k,j}^{(\epsilon)}(\theta) &=& \exp\left(i\epsilon\theta\ket{11}\bra{11}_{k,j}\right)\hat{CP}_{k,j}(\theta) = \hat{CP}_{k,j}((1+\epsilon)\theta).
\label{eq:qft_overrot_model}
\end{eqnarray}
The implemented QFT is then
\begin{eqnarray}
\hat{U}_{\mathrm{QFT},n}^{(\epsilon)} = \hat{G}_L^{(\epsilon)}\cdots\hat{G}_2^{(\epsilon)}\hat{G}_1^{(\epsilon)},
\label{eq:qft_Uexp}
\end{eqnarray}
where each $\hat{G}_\ell^{(\epsilon)}$ is obtained from $\hat{G}_\ell$ by applying Eq.~(\ref{eq:qft_overrot_model}). The effective error unitary is
\begin{eqnarray}
\hat{X}_n(\epsilon) = \hat{U}_{\mathrm{QFT},n}^{\dagger}\hat{U}_{\mathrm{QFT},n}^{(\epsilon)} \in U(d),
\label{eq:qft_X_def}
\end{eqnarray}
and the coherent interaction-picture error channel is
\begin{eqnarray}
\mathcal{E}_n^{(\epsilon)}(\hat{\rho}) = \hat{X}_n(\epsilon)\hat{\rho}\hat{X}_n(\epsilon)^{\dagger}.
\label{eq:qft_error_channel}
\end{eqnarray}
Since $\mathcal{E}_n^{(\epsilon)}$ is unitary for every $\epsilon$, its unitarity equals
\begin{eqnarray}
u(\mathcal{E}_n^{(\epsilon)})=1,
\label{eq:qft_unitarity_one}
\end{eqnarray}
and therefore the $(r,u)$-assisted conversion necessarily collapses to the $u=1$ saturation form.

\medskip
\paragraph*{\bf Computing $(F,D)$ from trace moments.}---Let us define the spectral invariants
\begin{eqnarray}
P(\epsilon) := \abs{\tr{\hat{X}_n(\epsilon)}},
\quad
Q(\epsilon) := \abs{\tr{\hat{X}_n(\epsilon)^2}+\tr{\hat{X}_n(\epsilon)}^2}.
\label{eq:qft_PQ_def}
\end{eqnarray}
The average fidelity depends only on $P(\epsilon)$, and thus, we have
\begin{eqnarray}
F(\epsilon) = \frac{d+P(\epsilon)^2}{d(d+1)}.
\label{eq:qft_F_from_P}
\end{eqnarray}
For the fidelity deviation, it is convenient to compute the second moment
\begin{eqnarray}
E_2(\epsilon):=D(\epsilon)^2+F(\epsilon)^2 = \int d\psi \abs{\bra{\psi}\hat{X}_n(\epsilon)\ket{\psi}}^4,
\label{eq:qft_E2_def}
\end{eqnarray}
which is evaluated in terms of $\tr{(\hat{X}_n(\epsilon))}$ and $\tr{(\hat{X}_n(\epsilon)^2)}$:
\begin{eqnarray}
E_2(\epsilon) = \frac{ 2d(d+3)+4(d+2)\abs{\tr{(\hat{X}_n(\epsilon))}}^2 + \abs{\tr{(\hat{X}_n(\epsilon)^2)}}^2 + \abs{\tr{(\hat{X}_n(\epsilon))}}^4 + 2\mathrm{Re}\left(\tr{(\hat{X}_n(\epsilon)^2)}\left(\tr{(\hat{X}_n(\epsilon))}^{\ast}\right)^2\right)}{d(d+1)(d+2)(d+3)}.
\label{eq:qft_E2_formula}
\end{eqnarray}
Finally, we can obtain $D$ by
\begin{eqnarray}
D(\epsilon) = \sqrt{E_2(\epsilon)-F(\epsilon)^2}.
\label{eq:qft_D_from_E2}
\end{eqnarray}

\medskip
\paragraph*{\bf True diamond distance from the numerical range.}---Because $\hat{X}_n(\epsilon)$ is unitary, we can write
\begin{eqnarray}
d_{\diamond}(\mathcal{E}_n^{(\epsilon)},\mathcal{I}) = \sqrt{1-m(\hat{X}_n(\epsilon))^2},
\label{eq:qft_true_diamond}
\end{eqnarray}
where
\begin{eqnarray}
m(\hat{X}_n(\epsilon)) = \min\left\{ \abs{z} : z \in W(\hat{X}_n(\epsilon)) \right\},
\quad
W(\hat{X}_n(\epsilon))=\mathrm{conv}\left(\mathrm{Spec}(\hat{X}_n(\epsilon))\right).
\label{eq:qft_m_from_hull}
\end{eqnarray}
Thus, for the QFT example the true $d_{\diamond}$ can be computed exactly by extracting the eigenvalues of $\hat{X}_n(\epsilon)$ and taking the Euclidean distance from the origin to their convex hull in the complex plane.

\medskip
\paragraph*{\bf Numerical comparison (up to $n=10$).}---Let $r(\epsilon)=1-F(\epsilon)$.
The best-known fidelity-only bound gives
\begin{eqnarray}
d_{\diamond}(\mathcal{E}_n^{(\epsilon)},\mathcal{I}) \le \sqrt{d(d+1)r(\epsilon)}.
\label{eq:qft_fonly_bound}
\end{eqnarray}
In contrast, the unitarity-assisted upper bound of Kueng et al. becomes strictly worse when $u=1$:
\begin{eqnarray}
d_{\diamond}(\mathcal{E}_n^{(\epsilon)},\mathcal{I}) \le \frac{d}{\sqrt{2}}\sqrt{d(d+1)r(\epsilon)},
\quad
u(\mathcal{E}_n^{(\epsilon)})=1.
\label{eq:qft_ru_u1_bound}
\end{eqnarray}
For $n=10$ one has $d=1024$, so Eq.~(\ref{eq:qft_ru_u1_bound}) is worse than Eq.~(\ref{eq:qft_fonly_bound}) by a factor $1024/\sqrt{2}$, i.e., the $(r,u)$ strategy provably loses sharpness by orders of magnitude in the fully coherent regime relevant to systematic circuit miscalibrations.

\begin{figure}[t]
  \centering
  \includegraphics[width=0.90\linewidth]{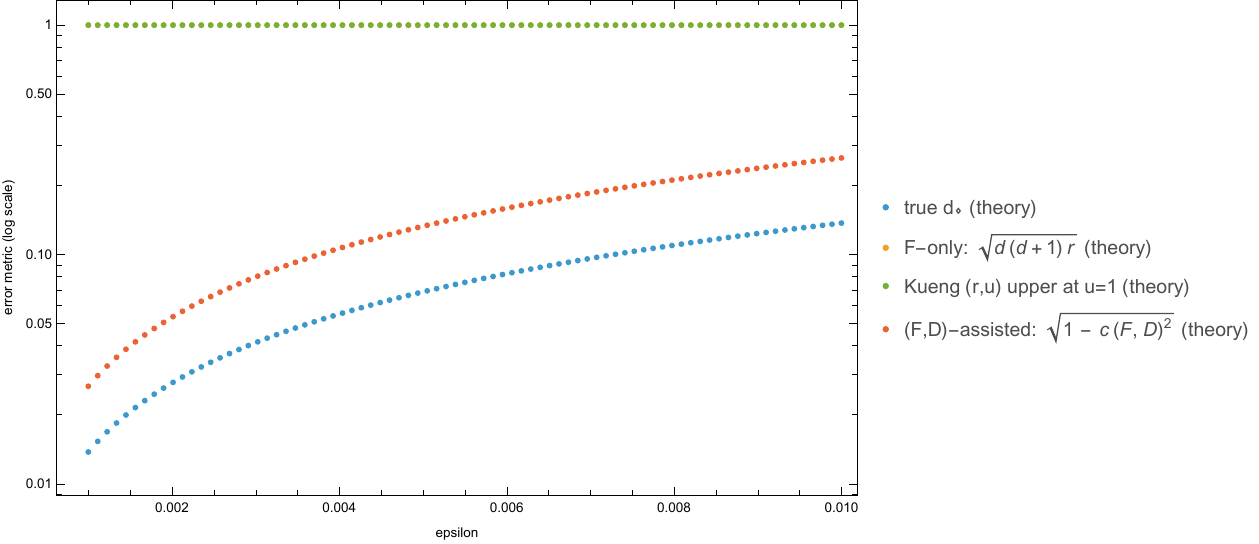}
  \caption{{\bf QFT coherent over-rotation example at $n=10$ qubits ($d=1024$)}. We plot the true diamond distance $d_{\diamond}(\mathcal{E}_n^{(\epsilon)},\mathcal{I})$ (computed from the convex hull of $\mathrm{Spec}(\hat{X}_n(\epsilon))$) and three upper bounds: the fidelity-only bound $\sqrt{d(d+1)r}$, the saturated $(r,u)$ bound at $u=1$, and the proposed $(F,D)$-assisted bound $\sqrt{1-c(F(\epsilon),D(\epsilon))^2}$. The QFT circuit is decomposed as in Eqs.~(\ref{eq:qft_ideal_gate_list})-(\ref{eq:qft_angles}) and each primitive is subject to the uniform coherent over-rotation model in Eq. (\ref{eq:qft_overrot_model}).}
  \label{fig:QFT_example}
\end{figure}

From the same unitary $\hat{X}_n(\epsilon)$, the pair $(F(\epsilon),D(\epsilon))$ fixes the two moment-invariants $(P(\epsilon),Q(\epsilon))$ through Eqs.~(\ref{eq:qft_PQ_def})--(\ref{eq:qft_D_from_E2}). We then compute the moment-assisted quantity $c(F,D)$ for dimension $d=2^n$:
\begin{eqnarray}
c(F(\epsilon),D(\epsilon)) = \left[ \frac{P(\epsilon)}{d} - \frac{1}{2d}\sqrt{(d-2)\left(dQ(\epsilon)+d^2-(d+2)P(\epsilon)^2\right)} \right]_+,
\label{eq:qft_cFD_def}
\end{eqnarray}
and obtain the corresponding coherent worst-case certification
\begin{eqnarray}
d_{\diamond}(\mathcal{E}_n^{(\epsilon)},\mathcal{I}) \le \sqrt{1-c(F(\epsilon),D(\epsilon))^2}.
\label{eq:qft_FD_bound}
\end{eqnarray}
In Fig.~\ref{fig:QFT_example}, we compare four curves as functions of $\epsilon$: (i) the true diamond distance in Eq.~(\ref{eq:qft_true_diamond}), (ii) the fidelity-only bound in Eq.~(\ref{eq:qft_fonly_bound}), (iii) the saturated $(r, u)$ bound Eq.~(\ref{eq:qft_ru_u1_bound}), (iv) our $(F, D)$-assisted bound Eq.~(\ref{eq:qft_FD_bound}). Fig.~\ref{fig:QFT_example} plots these quantities on a log-scale for the case $n=10$ (i.e., $d=1024$), computed exactly from the eigenvalues of $\hat{X}_n(\epsilon)$.

\section{A regime-adaptive certification strategy using $(r,u)$ and $(F,D)$}\label{sec:hybrid_FDu}

The diamond distance $d_{\diamond}$ is the fault-tolerance-relevant worst-case metric, but it is not directly measurable. On the other hand, the experimentally accessible data---such as $r=1-F$ obtained from
randomized benchmarking---can be related to $d_{\diamond}$ only through bounds whose tightness depends strongly on the structure of the noise. The most delicate regime is the coherent-dominated one:
even when $r \ll 1$, systematic (unitary) errors can force $d_{\diamond}$ to scale as $\Theta(\sqrt{r})$.

A celebrated refinement proposed by Kueng et al.~\cite{Kueng2016} supplements $r$ with the \emph{unitarity} $u$, an experimentally accessible witness of how close a channel is to being unitary. This refinement is especially powerful in decoherence-dominated regimes, where it can certify the favorable scaling $d_{\diamond}=O(r)$. However, as emphasized in Sec.~\ref{subsec:unitarity_limitation}, unitarity necessarily \emph{saturates} to $u=1$ for purely coherent errors and therefore loses resolving power exactly in the regime where the coherent errors can be the most damaging. In this work, by incorporating the fidelity deviation $D$ (Sec.~\ref{sec:FD_def}), we obtain a sharp worst-case certification inside the coherent regime. 

In this section, we formalize a simple but rigorous hybrid strategy: use $(r, u)$ to exploit decoherence-dominated linearity when available, and use $(F,D)$ to sharpen worst-case assessment when unitarity saturates.

\subsection{Unitarity as a regime witness and $(r, u)$-based worst-case bounds}\label{subsec:hybrid_ru}

We briefly recall the role of the unitarity $u(\mathcal{E})$. For a quantum channel $\mathcal{E}$ on $\mathbb{C}^d$, we adopt the standard definition introduced in Ref.~\cite{wallman2015estimating} (also used in Ref.~\cite{Kueng2016}):
\begin{eqnarray}
u(\mathcal{E}) := \frac{d}{d-1}\int d\psi \; \tr{\left[\mathcal{E}\!\left(\ketbra{\psi}{\psi}-\frac{\hat{\mathds{1}}}{d}\right)\right]^2}, \quad 0 \le u(\mathcal{E}) \le 1,
\label{eq:def_unitarity}
\end{eqnarray}
where $d\psi$ is the Haar measure on pure states. In particular, $u(\mathcal{E})=1$ for unitary channels, while $u(\mathcal{E})<1$ indicates the presence of irreversible (incoherent) components.

For a unital noise without leakage, Kueng et al. derived the following characterization of the worst-case error in terms of the experimentally accessible pair $(r, u)$:
\begin{eqnarray}
c_d\sqrt{u+\frac{2dr}{d-1}-1} \le d_{\diamond}(\mathcal{E},\mathcal{I}) \le d^2 c_d\sqrt{u+\frac{2dr}{d-1}-1},
\label{eq:kueng_ru_bound_hybrid}
\end{eqnarray}
where $c_d=\frac{1}{2}\sqrt{1-\frac{1}{d^2}}$. This inequality implies that the scaling of $d_{\diamond}$ with $r$ is controlled by the deviation of $u$ from unity. Indeed, defining the non-negative quantity 
\begin{eqnarray}
\eta(r, u) := u+\frac{2dr}{d-1}-1,
\end{eqnarray}
Eq.~(\ref{eq:kueng_ru_bound_hybrid}) reads
\begin{eqnarray}
d_{\diamond}(\mathcal{E},\mathcal{I}) \le d^2 c_d \sqrt{\eta(r,u)}.
\label{eq:kueng_ru_upper_eta}
\end{eqnarray}
Moreover, within this framework, Kueng et al. identify the condition
\begin{eqnarray}
u = 1-\frac{2dr}{d-1}+O(r^2) \quad \Longleftrightarrow \quad \eta(r,u)=O(r^2),
\label{eq:kueng_linear_condition_hybrid}
\end{eqnarray}
as necessary and sufficient for recovering the favorable linear scaling $d_{\diamond}=O(r)$. Thus, $u$ can be viewed as a \emph{regime witness}: it certifies whether the device has entered a decoherence-dominated regime where the average-case estimates are representative of the worst-case behavior.

However, when coherence dominates, one typically has $u \simeq 1$, and for a purely unitary error channel $u=1$ identically. Inserting $u=1$ into Eq.~(\ref{eq:kueng_ru_upper_eta}), we have
\begin{eqnarray}
d_{\diamond}(\mathcal{E}, \mathcal{I}) \le d^2 c_d \sqrt{\frac{2d}{d-1}r} = \frac{d}{\sqrt{2}}\sqrt{d(d+1)r}.
\label{eq:kueng_u1_hybrid}
\end{eqnarray}
Comparing with the best-known fidelity-only upper bound $d_{\diamond}\le \sqrt{d(d+1)r}$ (Sec.~\ref{subsec:rb_diamond_bounds}), we see that in the fully coherent limit the unitarity-assisted upper bound becomes looser by a fixed factor $d/\sqrt{2}$. This is not a contradiction: $u$ is doing what it is designed to do---detect incoherence---but it cannot resolve different coherent errors once it saturates.

\subsection{$(F,D)$ as coherent-resolution data and the $c(F,D)$ certification}\label{subsec:hybrid_FD}

Now let us recall our complementary certification obtained in Sec.~\ref{sec:main_results}. The fidelity deviation $D$ defined in Sec.~\ref{sec:FD_def} is the standard deviation of the state-dependent fidelity $f_{\mathcal{E}}(\psi)=\bra{\psi}\mathcal{E}(\ketbra{\psi}{\psi})\ket{\psi}$:
\begin{eqnarray}
F=\int d\psi~ f_{\mathcal{E}}(\psi),
\quad
D^2=\int d\psi f_{\mathcal{E}}(\psi)^2 - F^2.
\label{eq:FD_moments_recall}
\end{eqnarray}
Unlike unitarity, $D$ does not collapse to a constant on the set of unitary errors; instead, it probes a higher Haar moment of the same experimentally measurable fidelity landscape.

In the coherent (unitary/systematic) setting where $\mathcal{E}(\hat{\rho})=\hat{X}\hat{\rho} \hat{X}^{\dagger}$ with $\hat{X} \in U(d)$, we introduced a quantity $c(F,D)$ with the following operational meaning: it is a certified lower bound on the minimum overlap amplitude $m(\hat{X})=\min_{\psi}\abs{\bra{\psi}\hat{X}\ket{\psi}}$, and therefore yields an upper bound on the diamond distance. In particular, for $d \ge 4$, {\bf Theorem~\ref{thm:main_d4}} proves that
\begin{eqnarray}
d_{\diamond}(\mathcal{E},\mathcal{I}) \le \sqrt{1-c(F,D)^2},
\label{eq:FD_upper_recall}
\end{eqnarray}
and that this bound is tight for admissible $(F, D)$. We emphasize that this certification addresses precisely the regime where $u$ saturates: when $u=1$ provides no resolution within the coherent manifold, the additional fourth-moment information in $D$ constrains the coherent spectral spread, and thereby, sharpens worst-case guarantees.

\subsection{A hybrid regime-adaptive worst-case bound}\label{subsec:hybrid_theorem}

We now formalize a simple hybrid strategy that combines the strengths of both approaches.

\begin{definition}[Two computable upper bounds]
\label{def:two_bounds}
Let $r=1-F$ be the gate infidelity and $u$ the unitarity. Define the unitarity-assisted bound (valid for unital noise without leakage) as
\begin{eqnarray}
B_{ru}(r,u) := d^2 c_d\sqrt{u+\frac{2dr}{d-1}-1},
\quad
c_d=\frac{1}{2}\sqrt{1-\frac{1}{d^2}}.
\label{eq:def_Bru}
\end{eqnarray}
Define the moment-assisted coherent bound (valid under the coherent model assumptions of Sec.~\ref{sec:main_results}) as
\begin{eqnarray}
B_{FD}(F,D) := \sqrt{1-c(F,D)^2},
\label{eq:def_BFD}
\end{eqnarray}
where $c(F,D)$ is the quantity defined in Eq.~(\ref{eq:def_cFD}) for $d\ge 4$.
\end{definition}

\begin{theorem}[Hybrid certification]
\label{thm:hybrid_certification}
Assume that $\mathcal{E}$ is unital and has no leakage so that Eq.~(\ref{eq:kueng_ru_bound_hybrid}) applies. Assume further that the coherent model hypotheses needed for Eq.~(\ref{eq:FD_upper_recall}) apply to $\mathcal{E}$ (e.g., $\mathcal{E}(\hat{\rho})=\hat{X}\hat{\rho}\hat{X}^{\dagger}$ for a unitary $\hat{X}$ in the $d\ge 4$ coherent setting). Then the diamond distance obeys the hybrid upper bound
\begin{eqnarray}
d_{\diamond}(\mathcal{E},\mathcal{I}) \le B_{\mathrm{hyb}}(F,D,r,u) := \min \left\{ B_{ru}(r,u), B_{FD}(F,D) \right\}.
\label{eq:hybrid_min_bound}
\end{eqnarray}
\end{theorem}

\begin{proof}---Under the first assumption, $d_{\diamond}(\mathcal{E},\mathcal{I})\le B_{ru}(r,u)$ by Eq.~(\ref{eq:kueng_ru_upper_eta}). Under the second assumption, $d_{\diamond}(\mathcal{E},\mathcal{I})\le B_{FD}(F,D)$ by Eq.~(\ref{eq:FD_upper_recall}). Taking the minimum of two valid upper bounds yields Eq.~(\ref{eq:hybrid_min_bound}).
\end{proof}

{\bf Theorem~\ref{thm:hybrid_certification}} turns the conceptual complementarity of $(r,u)$ and $(F,D)$ into a rigorous recipe.

\begin{corollary}[Decoherence-dominated regime: unitarity enables linear scaling]
\label{cor:hybrid_decoherence}
If $\eta(r,u)=u+\frac{2dr}{d-1}-1 = O(r^2)$, or equivalently, $u=1-\frac{2dr}{d-1}+O(r^2)$, then $B_{ru}(r,u)=O(r)$ and therefore $B_{\mathrm{hyb}}(F,D,r,u)=O(r)$.
\end{corollary}

\begin{proof}---Immediate from $B_{ru}(r,u)=d^2c_d\sqrt{\eta(r,u)}$ and $\eta(r,u)=O(r^2)$.
\end{proof}

\begin{corollary}[Coherent-dominated regime: $u\simeq 1$ saturates while $D$ remains informative]
\label{cor:hybrid_coherent}
In the fully coherent limit $u=1$, the unitarity-assisted bound reduces to Eq.~(\ref{eq:kueng_u1_hybrid}) and is looser than the fidelity-only bound by a factor $d/\sqrt{2}$. In this regime, $B_{\mathrm{hyb}}$ selects the sharper moment-assisted certification whenever $B_{FD}(F,D) < B_{ru}(r,1)$, which is typical for structured two-qubit coherent errors (Sec.~\ref{sec:main_results}).
\end{corollary}

The operational meaning of this hybrid strategy can be summarized as follows. The unitarity $u$ acts as a coherence thermometer: it signals whether the system has left the favorable decoherence-dominated regime where $d_{\diamond}=O(r)$ may be certified. If $u$ is far below unity, $B_{ru}$ is often the decisive certification. If $u$ is close to unity, $u$ becomes nearly uninformative and $B_{ru}$ may lose sharpness; in that regime, the additional higher-moment data contained in $D$ provides resolution within the coherent manifold, and $B_{FD}$ becomes the decisive certification.

\section{Discussions}\label{sec:discussion}


In this work, we have introduced the fidelity deviation $D$ as the standard deviation of the state-dependent survival probability and treated $(F, D)$ as a pair of moment data describing a fidelity landscape rather than a single mean value. We derived closed-form relations connecting $F$ and $D$ to spectral trace invariants of the effective error unitary $\hat{X}$ and showed that, for $d \ge 4$, $(F, D)$ determined two independent invariants that include both $\tr{(\hat{X})}$ and $\tr{(\hat{X}^2)}$. Using this extra spectral information, we solved the corresponding constrained extremal problem and obtained a tight lower bound on the minimum overlap amplitude $m(\hat{X})$, and hence a sharp worst-case certificate for the diamond distance $d_{\diamond}$. We also presented a direct sampling protocol that estimated $D$ from the same survival-probability data stream used for $F$, and we numerically demonstrated with representative coherent examples (CZ gate, Toffoli gate, and quantum Fourier transform) that the resulting $(F,D)$-assisted bound tracked the true worst-case behavior far more closely than fidelity-only conversions in the coherent limit.

From an engineering and benchmarking standpoint, the main practical message is that the $(F,D)$-assisted certificate should be treated as a necessary performance report for high-fidelity gate development, not as an optional post-processing trick. Average fidelity alone is a dangerously lossy summary in the coherent-dominated regime because it fixes only a coarse first-moment constraint, effectively through $\abs{\tr{\hat X}}$, while leaving enough freedom for the spectral polygon of $\hat X$ to approach the origin and collapse $m(\hat X)$. The deviation $D$ provides the missing second-moment constraint that controls spectral spread (through invariants that include $\tr{\hat X^2}$), and therefore controls how small the worst-case overlap can be at fixed $F$. Consequently, in an FTQC-oriented evaluation it is not sufficient to compare gates by reporting only $F$ (or $r$). We advocate that future gate-quality tables and device-comparison plots should report $F$ together with $D$, and should compare the derived certified overlap $c(F,D)$ (or equivalently the certified worst-case bound $\sqrt{1-c(F,D)^2}$). This recommendation is especially relevant for multi-qubit entangling gates (starting at $d \ge 4$), which dominate logical error budgets and where coherent structure has enough internal degrees of freedom that $D$ carries genuinely new information beyond $F$. Importantly, this does not demand tomography: $D$ is a second moment of the same survival-probability data stream already used to estimate $F$, so once an experiment commits to measuring fidelities carefully, the additional overhead is primarily computational and statistical rather than conceptual or infrastructural.

Finally, we can provide several forward-looking points. First, $(F,D)$ and $(r,u)$ are informative in complementary regimes, suggesting a regime-adaptive certification workflow: use $u$ as a coherence thermometer to detect when incoherent components make linear-in-$r$ bounds meaningful, and use $(F,D)$ to regain sharp worst-case resolution when coherence dominates and $u$ saturates. Second, the emergence of $D$ as a physically meaningful ``coherent non-uniformity'' indicator suggests a new calibration target: beyond maximizing $F$, control design can aim to reduce $D$ to suppress state-dependent valleys that drive worst-case risk. Third, while the present results focus on coherent unitary errors, the same moment-based philosophy motivates extensions to near-unitary channels, coherent-plus-incoherent mixtures, and leakage-aware settings, where one may combine fluctuation information with additional witnesses to retain both practicality and worst-case relevance. Overall, we view $(F,D)$ reporting and the derived $c(F,D)$ comparison as a concrete step toward standardizing gate metrics that are directly aligned with FTQC requirements rather than with average-case convenience.


\medskip\medskip
{\bf Acknowledgement.}---This work was supported by the Ministry of Science, ICT and Future Planning (MSIP) by the National Research Foundation of Korea (RS-2024-00432214, RS-2025-03532992, and RS-2025-18362970) and the Institute of Information and Communications Technology Planning and Evaluation grant funded by the Korean government (RS-2019-II190003, ``Research and Development of Core Technologies for Programming, Running, Implementing and Validating of Fault-Tolerant Quantum Computing System''), the Korean ARPA-H Project through the Korea Health Industry Development Institute (KHIDI), funded by the Ministry of Health \& Welfare, Republic of Korea (RS-2025-25456722). We acknowledge the Yonsei University Quantum Computing Project Group for providing support and access to the Quantum System One (Eagle Processor), which is operated at Yonsei University.


\appendix
\section{Representation-Theoretic Derivation of Proposition~\ref{prop:FD_closed_form_unitary}}
\label{append:Haar}
This appendix provides a self-contained account of the algebraic mechanisms that underlie the Haar-integral computations used in the main text. Rather than reproducing full proofs or detailed derivations of the relevant results, our aim is to isolate the representation-theoretic principles that explain why the Haar-averaged expressions appearing in Proposition~\ref{prop:FD_closed_form_unitary} take their particular closed forms. In particular, we clarify how quantities such as the averaged fidelity $F$ and fidelity deviation $D$, introduced in Sec.~\ref{sec:FD_def}, ultimately reduce to combinations of permutation-invariant tensors, as dictated by the commutant structure of tensor-power representations.

Several ingredients presented here are standard and overlap with earlier works, including Refs.~\cite{cho2025entangling,cho2025fundamental,cho2025witness}. Our emphasis, however, is not on novelty but on organization and clarity: we collect these tools in a unified framework tailored to the specific Haar averages required in this paper. For a broader and more systematic introduction to Haar integration and its applications in quantum information theory, we refer the reader to Ref.~\cite{Mele2024introductiontohaar}.

The guiding theme of this appendix is symmetry. The invariance of the Haar measure under left and right multiplication imposes strong algebraic constraints on all averaged quantities. In the context of the unitary group $U(d)$, these constraints imply that Haar integration effectively projects operators onto the subalgebra that is invariant under collective conjugation by tensor powers of unitaries. As a result, the structure of Haar averages is governed entirely by the commutant of the tensor-power representation.

Concretely, for each $k \ge 1$, we consider the $k$-fold twirling map
\begin{eqnarray}
\mathcal{T}_k(\hat{O}) := \int \hat{U}^{\otimes k} \hat{O} (\hat{U}^\dagger)^{\otimes k} \, d\hat{U},
\end{eqnarray}
which averages an operator over conjugation by $k$ copies of a Haar-random unitary. By construction, $\mathcal{T}_k$ is equivariant with respect to the tensor-power action of $U(d)$ and therefore maps any operator to an element of the corresponding commutant. The Schur--Weyl duality identifies this commutant with the algebra generated by permutation operators acting on $(\mathbb{C}^d)^{\otimes k}$, explaining the ubiquitous appearance of permutations in Haar-averaged expressions.

To make this structure computationally explicit, we introduce the $k$-th moment operators associated with the Haar measure and describe their expansion in the permutation basis. The coefficients in this expansion are determined by the Weingarten calculus~\cite{Collins:2006jgn}, which provides closed-form expressions for integrals of products of matrix elements of $\hat{U}$ and $\hat{U}^\dagger$, such as
\begin{eqnarray}
\int \hat{U}_{i_1 j_1} \cdots \hat{U}_{i_k j_k} \overline{\hat{U}}_{i'_1 j'_1} \cdots \overline{\hat{U}}_{i'_k j'_k} d\hat{U} = \sum_{\sigma,\tau \in S_k}  \delta_{i_1, i'_{\sigma(1)}} \cdots \delta_{i_k, i'_{\sigma(k)}} \delta_{j_1, j'_{\tau(1)}} \cdots \delta_{j_k, j'_{\tau(k)}} Wg(\sigma^{-1}\tau , d),
\end{eqnarray}
In this framework, the Schur--Weyl duality fixes the invariant tensor structure, while the Weingarten function supplies the weight associated with each permutation. We further summarize a collection of trace identities and cycle-decomposition rules that allow these general formulas to be evaluated efficiently for the specific operators relevant to Proposition~\ref{prop:FD_closed_form_unitary}.

\subsection{Haar measure}

The Haar measure provides the canonical notion of uniform randomness on a compact group. In the present work, it serves as the reference distribution for random unitaries drawn from the unitary group $U(d)$. Rather than emphasizing its measure-theoretic construction, we focus on the invariance properties of the Haar measure, as these are directly responsible for the structural simplifications observed in Haar-averaged quantities throughout this work.

For a compact group $G$, the Haar measure $\mu_H$ is the unique probability measure that is invariant under the group action. In the case of $U(d)$, we denote integration with respect to this measure simply by $dU$. Given an integrable function $f : U(d) \to \mathbb{C}$, its Haar average is written as
\begin{eqnarray}
\int f(U)\, dU .
\end{eqnarray}

The defining feature of the Haar measure is its invariance under left and right multiplication, as well as under conjugation and inversion. Concretely, for any fixed $V \in U(d)$ and any integrable function $f$, one has
\begin{eqnarray}
\int f(VU)\, dU = \int f(UV)\, dU = \int f(VUV^\dagger)\, dU = \int f(U)\, dU .
\end{eqnarray}
Invariance under inversion, $\int f(U^\dagger)\,dU = \int f(U)\,dU$, follows as a special case. These identities express the fact that averaging over the Haar measure completely removes any dependence on a preferred reference frame in $U(d)$.

From the perspective of operator averages, these invariance properties have a particularly important consequence. Consider an operator $\hat{O}$ acting on a $k$-fold tensor product space $(\mathbb{C}^d)^{\otimes k}$, and define the associated $k$-fold twirling map
\begin{eqnarray}
\mathcal{T}_k(\hat{O}) := \int \hat{U}^{\otimes k} \hat{O} (\hat{U}^\dagger)^{\otimes k} \, d\hat{U}.
\end{eqnarray}
By construction, $\mathcal{T}_k(\hat{O})$ is invariant under simultaneous conjugation by $\hat{V}^{\otimes k}$ for any $\hat{V} \in U(d)$. Equivalently, the Haar twirl acts as an orthogonal projection
onto the subspace of operators that commute with the collective action of $U(d)$ on $k$ copies.

This observation is the starting point for all subsequent calculations. It implies that Haar integration produces only operators belonging to a highly constrained invariant subalgebra, whose structure is fixed entirely
by symmetry. In the following subsections, we characterize this invariant algebra explicitly using representation theory, showing that it is spanned by permutation operators acting on the tensor factors. This structural result ultimately explains why Haar-averaged expressions can be written in closed form and evaluated using combinatorial data such as permutation cycles.
\subsection{$k$-th moment operators}

To formalize the effect of Haar-random unitaries on multi-copy systems, it is convenient to package Haar averages into a family of linear maps, known as moment operators. These maps capture, in a basis-independent way, how conjugation by random unitaries acts on operators defined on tensor-power Hilbert spaces.

For each integer $k \ge 1$, the $k$-th moment operator is defined as the Haar average of conjugation by the $k$-fold tensor representation of $U(d)$. Explicitly, it is the linear map
\begin{eqnarray}
\mathcal{T}_k : \mathcal{L}\bigl((\mathbb{C}^d)^{\otimes k}\bigr) \longrightarrow \mathcal{L}\bigl((\mathbb{C}^d)^{\otimes k}\bigr),
\end{eqnarray}
given by
\begin{eqnarray}
\mathcal{T}_k(\hat{O}) := \int \hat{U}^{\otimes k} \hat{O} \hat{U}^{\dagger \otimes k} \, d\hat{U},
\qquad
\hat{O} \in \mathcal{L}\bigl((\mathbb{C}^d)^{\otimes k}\bigr).
\end{eqnarray}

Several basic properties of $\mathcal{T}_k$ follow directly from Haar invariance. The map is linear, trace preserving, and self-adjoint with respect to the Hilbert--Schmidt inner product. More importantly for our purposes, $\mathcal{T}_k$ is idempotent,
\begin{eqnarray}
\mathcal{T}_k^2 = \mathcal{T}_k ,
\end{eqnarray}
and therefore acts as an orthogonal projector.

The image of this projection is precisely the set of operators that are invariant under collective conjugation by $U(d)$, i.e., operators $\hat{A}$ satisfying
\begin{eqnarray}
\hat{U}^{\otimes k} \hat{A} \hat{U}^{\dagger \otimes k} = \hat{A}
\qquad
(\forall \hat{U} \in U(d)).
\end{eqnarray}
Equivalently, the moment operator projects onto the commutant of the $k$-fold tensor representation of the unitary group. As a result, evaluating Haar averages is reduced to identifying and manipulating elements of this invariant subalgebra.

Beyond their role in explicit Haar integrations, moment operators also provide a natural benchmark for quantifying how closely a finite ensemble of unitaries approximates the Haar measure. In particular, many notions of pseudorandomness in quantum information are phrased in terms of reproducing Haar moments up to a fixed order. This motivates the definition of \emph{unitary $t$-design}~\cite{dankert2009exact}.

\begin{definition}[Unitary $t$-design]
A finite ensemble of unitaries $\mathcal{E} = \{\hat{U}_j\}$ forms an exact unitary $t$-design if, for all $k \le t$ and all operators
$\hat{O} \in \mathcal{L}\bigl((\mathbb{C}^d)^{\otimes k}\bigr)$,
\begin{eqnarray}
\frac{1}{|\mathcal{E}|} \sum_{\hat{U}_j \in \mathcal{E}}
\hat{U}_j^{\otimes k} \hat{O} \hat{U}_j^{\dagger \otimes k}
= \mathcal{T}_k(\hat{O}) .
\end{eqnarray}
\end{definition}

In this sense, the Haar moment operators provide the universal reference against which finite ensembles are compared: an ensemble reproduces the statistics of Haar-random unitaries up to order $t$ if and only if it matches the moment operators $\mathcal{T}_k$ for all $k \le t$. In the remainder of this appendix, we exploit the projection property of $\mathcal{T}_k$ together with representation-theoretic tools—most notably Schur--Weyl duality— to obtain explicit formulas for these operators and for the Haar-averaged quantities appearing in Proposition~\ref{prop:FD_closed_form_unitary}.

\subsection{Schur--Weyl duality and the structure of Haar averages}

The defining property of the $k$-th moment operator is that it projects onto operators that are invariant under the collective action of the unitary group. As a consequence, the evaluation of Haar averages reduces to a purely algebraic question: how can one characterize all operators on $(\mathbb{C}^d)^{\otimes k}$ that commute with $\hat{U}^{\otimes k}$ for every $\hat{U} \in U(d)$?

To formalize this, we consider the commutant of the $k$-fold tensor representation of the unitary group. This object captures the full set of symmetry constraints imposed by Haar invariance and determines the possible structure of all Haar-averaged expressions.

\begin{definition}[$k$-th order commutant]
For a set of operators $S \subset \mathcal{L}(\mathbb{C}^d)$, we define the $k$-th order commutant as
\begin{eqnarray}
\mathrm{Comm}(S,k)
:= \left\{
\hat{A} \in \mathcal{L}\bigl((\mathbb{C}^d)^{\otimes k}\bigr)
\,:\,
[\hat{A}, \hat{B}^{\otimes k}] = 0
\;\; \forall \hat{B} \in S
\right\}.
\end{eqnarray}
This is a subalgebra of $\mathcal{L}((\mathbb{C}^d)^{\otimes k})$, and it consists of the operators that are fixed by the Haar twirl.
\end{definition}

By construction, the image of the moment operator $\mathcal{T}_k$ coincides with $\mathrm{Comm}(U(d),k)$. Since $\mathcal{T}_k$ is an orthogonal projection, any operator $\hat{O}$ admits a decomposition of the form
\begin{eqnarray}
\mathcal{T}_k(\hat{O})
= \sum_{i=1}^{\dim \mathrm{Comm}(U(d),k)}
\langle \hat{P}_i , \hat{O} \rangle_{\mathrm{HS}} \, \hat{P}_i ,
\label{eq:Tk_commutant_expansion_re}
\end{eqnarray}
where $\{\hat{P}_i\}$ is an orthonormal basis of the commutant with respect to the Hilbert--Schmidt inner product. The practical challenge, therefore, is to identify a concrete and computationally convenient basis for this invariant subalgebra.

A central result of representation theory provides a complete answer. The Schur--Weyl duality states that the commutant of the tensor-power representation of $U(d)$ is generated by the natural action of the symmetric group on tensor factors. More precisely, the only operators that commute with $\hat{U}^{\otimes k}$ for all unitaries $\hat{U}$ are linear combinations of permutation operators.

\begin{theorem}[Schur--Weyl duality]
For the natural action of $U(d)$ on $(\mathbb{C}^d)^{\otimes k}$,
\begin{eqnarray}
\mathrm{Comm}(U(d),k)
= \mathrm{span}\bigl\{ \hat{V}_d(\pi) : \pi \in S_k \bigr\},
\end{eqnarray}
where $\hat{V}_d(\pi)$ denotes the operator that permutes tensor factors according to $\pi$.
\end{theorem}

This classification has an immediate and far-reaching consequence: Haar integration cannot generate arbitrary tensor structures. Instead, every Haar-averaged operator must be expressible as a linear combination of permutation operators. In particular, the general expansion of the moment operator takes the explicit form
\begin{eqnarray}
\mathcal{T}_k(\hat{O})
= \sum_{\pi \in S_k} t_\pi \, \hat{V}_d(\pi),
\qquad
t_\pi \in \mathbb{C}.
\label{eq:Tk_perm_expansion_re}
\end{eqnarray}

From this perspective, the role of Haar invariance becomes transparent: it eliminates all operator components except those compatible with permutation symmetry. What remains is a finite-dimensional algebra whose basis elements are labeled by elements of the symmetric group. The task of evaluating Haar averages is therefore reduced to determining the coefficients $\{t_\pi\}$ associated with each permutation.

In the following sections, we make this decomposition fully explicit. We first summarize the algebraic properties of permutation operators and their associated invariant projectors, and then introduce the Weingarten calculus, which provides a systematic method for computing the coefficients $t_\pi$ appearing in Eq.~(\ref{eq:Tk_perm_expansion_re}).

\subsection{Permutation operators and invariant projectors}\label{sec:perm_proj}

The Schur--Weyl duality identifies permutation operators as the fundamental building blocks of Haar-invariant operators on tensor-power Hilbert spaces. In this subsection, we collect the concrete properties of these operators that are repeatedly used in explicit moment calculations. Our emphasis is practical: permutation operators provide a convenient algebraic basis in which Haar averages can be expanded and evaluated.

For each permutation $\pi \in S_k$, we define the corresponding operator $\hat{V}_d(\pi)$ acting on $(\mathbb{C}^d)^{\otimes k}$ by permuting tensor factors. In the computational basis, this action can be written as
\begin{eqnarray}
\hat{V}_d(\pi)
= \sum_{j_1,\ldots,j_k=0}^{d-1}
\ket{j_{\pi^{-1}(1)}, \ldots, j_{\pi^{-1}(k)}}
\bra{j_1,\ldots,j_k},
\end{eqnarray}
where the inverse permutation appears due to the convention that operators act from the left on kets.

The mapping $\pi \mapsto \hat{V}_d(\pi)$ defines a unitary representation of the symmetric group on the tensor-product space. In particular, permutation operators satisfy the relations
\begin{eqnarray}
\hat{V}_d(\mathrm{id}) &=& \hat{\openone}, \nonumber\\
\hat{V}_d(\pi)\hat{V}_d(\nu) &=& \hat{V}_d(\pi \circ \nu), \nonumber\\
\hat{V}_d(\pi)^\dagger &=& \hat{V}_d(\pi^{-1}).
\end{eqnarray}
When acting on product states or product operators, $\hat{V}_d(\pi)$ simply relabels the tensor factors. For instance,
\begin{eqnarray}
\hat{V}_d(\pi)
\bigl(\ket{\psi_1}\otimes\cdots\otimes\ket{\psi_k}\bigr)
&=&
\ket{\psi_{\pi^{-1}(1)}}\otimes\cdots\otimes\ket{\psi_{\pi^{-1}(k)}},
\nonumber\\
\hat{V}_d(\pi)
\bigl(\hat{A}_1\otimes\cdots\otimes\hat{A}_k\bigr)
\hat{V}_d(\pi)^\dagger
&=&
\hat{A}_{\pi^{-1}(1)}\otimes\cdots\otimes\hat{A}_{\pi^{-1}(k)}.
\end{eqnarray}

A key feature that makes permutation operators particularly useful in Haar calculations is that their traces depend only on the cycle structure of the underlying permutation. Specifically, for $\sigma,\pi \in S_k$,
\begin{eqnarray}
\tr\!\left[\hat{V}_d(\sigma)\hat{V}_d(\pi)\right]
= \tr\!\left[\hat{V}_d(\sigma\circ\pi)\right]
= d^{\#\mathrm{cyc}(\sigma\circ\pi)},
\end{eqnarray}
where $\#\mathrm{cyc}(\rho)$ denotes the number of disjoint cycles in the permutation $\rho$. This cycle-counting rule is the origin of the characteristic polynomial dependence on $d$ that appears in Haar-averaged expressions.

Beyond serving as a basis for invariant operators, permutation operators also enable the construction of projectors onto symmetry sectors of the tensor-product space. The simplest examples are the fully symmetric and fully antisymmetric subspaces, which are associated with the one-dimensional irreducible representations of the symmetric group, namely the trivial and sign representations, under the natural action of $S_k$ on $(\mathbb{C}^d)^{\otimes k}$.

The projector onto the symmetric subspace is given by
\begin{eqnarray}
\hat{P}_{\mathrm{sym}}^{(d,k)}
= \frac{1}{k!} \sum_{\pi \in S_k} \hat{V}_d(\pi),
\end{eqnarray}
and projects onto the subspace of states invariant under all permutations. Similarly, the projector onto the antisymmetric subspace reads
\begin{eqnarray}
\hat{P}_{\mathrm{asym}}^{(d,k)}
= \frac{1}{k!} \sum_{\pi \in S_k} \mathrm{sgn}(\pi)\, \hat{V}_d(\pi),
\end{eqnarray}
where $\mathrm{sgn}(\pi)$ denotes the signature of the permutation.

These projectors are Hermitian, idempotent, and orthogonal to each other, illustrating how permutation symmetry organizes $(\mathbb{C}^d)^{\otimes k}$ into invariant subspaces. More generally, Schur--Weyl duality implies that further symmetry sectors can be constructed from linear combinations of permutation operators weighted by the characters of $S_k$. 

In the context of Haar integration, however, it is typically sufficient to work with the symmetric projector, which provides a natural and compact way to express averaged operators and to isolate the symmetry sector relevant for state moments.

\subsection{$k$-th moment operator and the Weingarten expansion}\label{sec:kth_moment_weingarten}

Having identified permutation operators as a complete basis for Haar-invariant operators, we now describe how the $k$-th moment operator is constructed explicitly in this basis. The key idea is that Haar integration induces a projection whose coefficients are fixed uniquely by the orthogonality relations among permutation operators. The Weingarten calculus provides an efficient and systematic way to compute these coefficients.

We recall that the $k$-th moment (or twirling) operator acts on $\hat{O} \in \mathcal{L}((\mathbb{C}^d)^{\otimes k})$ as
\begin{eqnarray}
\mathcal{T}_k(\hat{O})
= \int \hat{U}^{\otimes k} \hat{O} (\hat{U}^\dagger)^{\otimes k} \, d\hat{U}.
\end{eqnarray}
By construction, $\mathcal{T}_k$ is an orthogonal projector onto the commutant algebra $\mathrm{Comm}(U(d),k)$ with respect to the Hilbert--Schmidt inner product. Therefore, expanding $\mathcal{T}_k(\hat{O})$ in the permutation basis amounts to determining the projection of $\hat{O}$ onto this invariant subspace.

To make this projection explicit, we consider the Gram matrix of permutation operators,
\begin{eqnarray}
G_{\pi,\sigma}
:= \tr\!\left[ \hat{V}_d(\pi)^\dagger \hat{V}_d(\sigma) \right]
= d^{\#\mathrm{cyc}(\pi^{-1} \circ \sigma)},
\qquad
\pi,\sigma \in S_k,
\label{eq:gram_perm}
\end{eqnarray}
which encodes the non-orthogonality of the permutation basis. The coefficients of the expansion of $\mathcal{T}_k(\hat{O})$ are obtained by inverting this Gram matrix on the support of the commutant. This inversion implements the orthogonal projection onto a non-orthonormal basis.

More concretely, one may write
\begin{eqnarray}
\mathcal{T}_k(\hat{O})
= \sum_{\pi,\sigma \in S_k}
\tr\!\left[\hat{V}_d(\pi)^\dagger \hat{O}\right]
\,
(G^{-1})_{\pi,\sigma}
\,
\hat{V}_d(\sigma).
\label{eq:Tk_Weingarten_re}
\end{eqnarray}
The matrix $G^{-1}$ appearing in Eq.~(\ref{eq:Tk_Weingarten_re}) defines the Weingarten function. Specifically, the Weingarten coefficient associated with a permutation $\rho \in S_k$ is given by
\begin{eqnarray}
\mathrm{Wg}(\rho,d)
:= (G^{-1})_{\mathrm{id},\rho}.
\end{eqnarray}
By permutation invariance, all matrix elements depend only on $\pi^{-1} \circ \sigma$.

With this definition, the expansion simplifies to the familiar form
\begin{eqnarray}
\mathcal{T}_k(\hat{O})
= \sum_{\pi,\sigma \in S_k}
\mathrm{Wg}(\pi^{-1}\sigma,d)
\,
\tr\!\left[\hat{V}_d(\pi)^\dagger \hat{O}\right]
\,
\hat{V}_d(\sigma).
\end{eqnarray}

This formulation highlights the complementary roles of symmetry and computation in Haar integration. The permutation operators specify the allowed invariant tensor structures, while the Weingarten function encodes the dimension-dependent weights required to assemble the correct projection. Importantly, the dimension dependence associated with the Haar integration enters through the Weingarten coefficients, while any additional $d$-dependence may arise from the operator $\hat{O}$ itself through its trace invariants.

In summary, the Weingarten expansion provides an explicit and algorithmic representation of the $k$-th moment operator. Combined with the structural insights from Schur--Weyl duality, it furnishes a complete framework for evaluating Haar-averaged quantities and for reducing seemingly complicated integrals to finite sums over permutations.

\subsection{Weingarten calculus and trace identities}

In the previous subsection, we expressed the $k$-th moment operator as an explicit expansion in the permutation basis, with coefficients given by the unitary Weingarten function. To turn this structural representation into concrete Haar-average formulas, one needs a collection of summation identities for the Weingarten coefficients together with trace identities for permutation operators. The purpose of this subsection is to gather these ingredients in a form tailored to later calculations.

We recall that the unitary Weingarten function $\mathrm{Wg}(\pi,d)$ is defined as the inverse of the Gram matrix
\begin{eqnarray}
G^{(k)}_{\pi,\sigma} = d^{\#\mathrm{cyc}(\pi\circ\sigma^{-1})},
\end{eqnarray}
whose entries depend only on the cycle decomposition of permutations. As a consequence, many identities involving $\mathrm{Wg}(\pi,d)$ reduce to elementary combinatorial properties of the symmetric group.

One identity that will appear repeatedly is the total sum of the Weingarten coefficients,
\begin{eqnarray}
\sum_{\tau \in S_k} \mathrm{Wg}(\tau,d)
=
\frac{1}{\prod_{i=0}^{k-1}(d+i)}.
\end{eqnarray}
This identity follows from the classical relation
\begin{eqnarray}
\sum_{\sigma \in S_k} d^{\#\sigma}
=
\prod_{i=0}^{k-1}(d+i),
\end{eqnarray}
which expresses the generating function of the Stirling numbers of the first kind. From the perspective of the $k$-th moment operator, this relation reflects the fact that a uniform contraction of the inverse Gram matrix yields the reciprocal of the total cycle weight encoded in $G^{(k)}$.

A second useful property concerns the convolution of the Weingarten function with functions on the symmetric group. Since $\mathrm{Wg}(\cdot,d)$ is a class function, its convolution factorizes in a simple way: for any class function $f:S_k\to\mathbb{C}$, depending only on the cycle structure,
\begin{eqnarray}
\sum_{\pi,\sigma\in S_k}
\mathrm{Wg}(\pi^{-1}\sigma,d)\, f(\sigma)
=
\left(
\sum_{\pi\in S_k} \mathrm{Wg}(\pi,d)
\right)
\left(
\sum_{\tau\in S_k} f(\tau)
\right).
\end{eqnarray}
This factorization plays a central role in simplifying Haar averages, as it allows one to separate the dimension-dependent Weingarten weights from operator-dependent traces.

As a direct consequence, one obtains a compact expression for matrix elements of the $k$-th moment operator in the computational basis:
\begin{eqnarray}
\bra{0}^{\otimes k} \mathcal{T}_k(\hat{O}) \ket{0}^{\otimes k}
=
\frac{1}{d(d+1)\cdots(d+k-1)}
\sum_{\pi\in S_k}
\tr\!\left[\hat{V}_d^\dagger(\pi)\hat{O}\right].
\end{eqnarray}
This identity will be used repeatedly in the main text, where Haar-averaged quantities reduce to such scalar contractions.

As a further application of the Weingarten expansion, we recall that the Haar twirling of tensor powers of a pure state yields the normalized projector onto the symmetric subspace. For any $\ket{\phi}\in\mathbb{C}^d$ and any integer $k\ge1$,
\begin{eqnarray}
\mathbb{E}_{U\sim\mu_H}
\left[
\hat{U}^{\otimes k}(\ket{\phi}\bra{\phi})^{\otimes k}\hat{U}^{\dagger\otimes k}
\right]
=
\frac{\hat{P}_{\mathrm{sym}}^{(d,k)}}{\tr\!\left[\hat{P}_{\mathrm{sym}}^{(d,k)}\right]},
\end{eqnarray}
where $\tr[\hat{P}_{\mathrm{sym}}^{(d,k)}] = {{k+d-1}\choose k}$. This result follows from the permutation invariance of $(\ket{\phi}\bra{\phi})^{\otimes k}$ and the uniqueness of the Haar-invariant projector in the commutant.

Finally, to evaluate explicit Haar averages appearing throughout this work, one needs to compute traces of operator tensors contracted with permutation operators,
\begin{eqnarray}
\tr{\left[
\bigl(\hat{A}_1\otimes\cdots\otimes\hat{A}_k\bigr)\hat{V}_d(\pi)
\right]}.
\end{eqnarray}
The following proposition provides the corresponding cycle-decomposition formula, which serves as the final computational ingredient in the Weingarten framework.
\begin{proposition}\label{cycle-decomposition}\cite{cho2025entangling}
Let $\pi \in S_k$ with disjoint cycle decomposition
\begin{eqnarray}
\pi = c_1 c_2 \cdots c_r,
\end{eqnarray}
where each cycle $c=(l_1\,l_2\,\dots\,l_{k_c})$ has length $k_c$. Then, for operators $\hat{A}_1 ,\ldots, \hat{A}_k \in \mathcal{L}(\mathbb{C}^d)$, the trace has the following decomposition:
\begin{eqnarray}
\tr{\left[ \bigl( \hat{A}_1 \otimes\cdots\otimes \hat{A}_k \bigr) \hat{V}_d(\pi) \right]}
=
\prod_{c \in \pi}
\tr{\left[ \prod_{m=0}^{k_c-1} \hat{A}_{c^{-m} (l_c)} \right]}.
\end{eqnarray}
\end{proposition}

Thus, the trace factorizes over the cycles of $\pi$, with each cycle contributing a product of operators in cyclic order. Together with the Weingarten identities summarized above, this trace factorization completes the set of tools required for all Haar-average calculations appearing in the main text.

\subsection{Derivation of Proposition~\ref{prop:FD_closed_form_unitary}}

We now combine the representation-theoretic tools developed in the previous subsections to derive the closed-form expressions for Haar-averaged quantities appearing in the main text, in particular Proposition~\ref{prop:FD_closed_form_unitary}. Conceptually, this subsection serves as the point where the abstract structure of Haar integration---expressed in terms of permutation operators, Weingarten coefficients, and cycle decompositions---is translated into explicit formulas involving simple trace invariants.

The essential object to be evaluated is a Haar average of multilinear forms of the type
\begin{eqnarray}
\int d\phi \,
\bra{\phi}^{\otimes k} \hat{O} \ket{\phi}^{\otimes k},
\end{eqnarray}
where $d\phi$ denotes the unitarily invariant measure on pure states and $\hat{O}$ is an operator acting on $(\mathbb{C}^d)^{\otimes k}$ with a fixed tensor-product structure.

As shown in Sec.~\ref{sec:kth_moment_weingarten}, Haar averaging over pure states can be expressed in terms of the $k$-th moment operator. In particular, using the fact that the Haar twirling of $(\ket{\phi}\bra{\phi})^{\otimes k}$ yields the normalized projector onto the symmetric subspace, one obtains the general identity
\begin{eqnarray}
\int d \phi \,
\bra{\phi}^{\otimes k} \hat{O} \ket{\phi}^{\otimes k}
=
\frac{
\tr\!\left[ \hat{O} \, \hat{P}_{\mathrm{sym}}^{(d,k)} \right]
}{
\tr\!\left[ \hat{P}_{\mathrm{sym}}^{(d,k)} \right]
}.
\label{eq:pure_state_average_general}
\end{eqnarray}
Thus, the evaluation of Haar-averaged state moments reduces to computing traces of $\hat{O}$ contracted with permutation operators spanning the symmetric projector. This reduction is purely structural and relies only on unitary invariance.

In the situations relevant to this work, the operator $\hat{O}$ possesses a product structure built from operators $\hat{X}_\alpha \in \mathcal{L}(\mathbb{C}^d)$. The explicit computation then proceeds by expanding $\hat{P}_{\mathrm{sym}}^{(d,k)}$ in the permutation basis and applying the cycle-decomposition formula of Proposition~\ref{cycle-decomposition}. This step converts each tensor trace into products of ordinary operator traces, with the cycle structure of permutations dictating the precise form of the resulting expressions.

We summarize below the resulting formulas for the second and fourth pure-state moments, which serve as the basic building blocks for the quantities studied in the main text. These expressions are obtained by expanding the symmetric projector in the permutation basis and applying the cycle-decomposition identity of Proposition~\ref{cycle-decomposition} to reduce each contribution to products of ordinary operator traces. The explicit evaluation follows the general procedure outlined above; closely related calculations have appeared in Refs.~\cite{cho2025fundamental,cho2025witness}.

\begin{proposition}[\cite{cho2025fundamental,cho2025witness}]
For $k=2$, let $\hat{O}_{\alpha} = \hat{X}_\alpha \otimes \hat{X}_\alpha^\dagger$. Then
\begin{eqnarray}
\int d\phi \,
\bra{\phi}^{\otimes 2} \hat{O}_{\alpha} \ket{\phi}^{\otimes 2}
=
\frac{
\abs{\tr\!\left(\hat{X}_\alpha\right)}^2 + d
}{
d(d+1)
}.
\label{eq:scalar2}
\end{eqnarray}
For $k=4$, define
\begin{eqnarray}
\hat{O}_{\alpha\beta}
=
\hat{X}_\alpha \otimes \hat{X}_\alpha^\dagger \otimes \hat{X}_\beta \otimes \hat{X}_\beta^\dagger .
\label{eq:Oab}
\end{eqnarray}
Then
\begin{eqnarray}
\int d\phi \,
\bra{\phi}^{\otimes 4} \hat{O}_{\alpha\beta} \ket{\phi}^{\otimes 4}
&=&
\frac{1}{d(d+1)(d+2)(d+3)}
\Big\{
d(d+4)
+ (d+4)\bigl(
\abs{\tr(\hat{X}_\alpha)}^2 + \abs{\tr(\hat{X}_\beta)}^2
\bigr)
+ \abs{\tr(\hat{X}_\alpha)}^2 \abs{\tr(\hat{X}_\beta)}^2
\nonumber \\
&&\quad
+ \abs{\tr(\hat{X}_\alpha \hat{X}_\beta)}^2
+ \abs{\tr(\hat{X}_\alpha \hat{X}_\beta^\dagger)}^2
+ 2\,\mathrm{Re}\!\left[
\tr(\hat{X}_\alpha \hat{X}_\beta)
\tr(\hat{X}_\alpha^\dagger)
\tr(\hat{X}_\beta^\dagger)
\right]
\nonumber \\
&&\quad
+ 2\,\mathrm{Re}\!\left[
\tr(\hat{X}_\alpha \hat{X}_\beta^\dagger)
\tr(\hat{X}_\alpha^\dagger)
\tr(\hat{X}_\beta)
\right]
+ 2\,\mathrm{Re}\!\left[
\tr(\hat{X}_\alpha \hat{X}_\beta \hat{X}_\alpha^\dagger \hat{X}_\beta^\dagger)
\right]
\Big\}.
\label{eq:dbar}
\end{eqnarray}
Each term corresponds to a distinct conjugacy class in $S_4$.

\end{proposition}

These expressions make explicit how Haar-averaged state moments reduce to a finite set of trace invariants determined by the cycle structure of permutations in $S_k$. In particular, the appearance of mixed trace terms in the $k=4$ case reflects the richer permutation structure beyond simple pairings. As a direct application, by setting $\alpha=\beta$ and rearranging the resulting terms, one obtains closed-form formulas for the quantities $(F,D)$ in the case of unitary errors, as presented in Proposition~\ref{prop:FD_closed_form_unitary}.

\begin{corollary}
\label{prop:FD_closed_form_unitary2}
Let $\hat{X} \in U(d)$ and define $F$ and $D$ by Eq.~(\ref{eq:def_FD_unitary}). Then
\begin{eqnarray}
F &=&
\frac{d + \abs{\tr(\hat{X})}^2}{d(d+1)},
\\
D^2 &=&
\frac{
2d(d+3)
+ 4(d+2)\abs{\tr(\hat{X})}^2
+ \abs{\tr(\hat{X}^2)}^2
+ \abs{\tr(\hat{X})}^4
+ 2\,\mathrm{Re}\!\left[
\tr(\hat{X}^2)\tr(\hat{X}^\dagger)^2
\right]
}{
d(d+1)(d+2)(d+3)
}
- F^2 .
\label{eq:D_trace_closedform2}
\end{eqnarray}
where the subtraction of $F^2$ reflects the variance structure of the fourth-order moment.
\end{corollary}
Several structural features of the above formulas are worth emphasizing. 
The fidelity $F$ depends only on the first trace invariant $\tr(\hat{X})$, reflecting the fact that it is determined entirely by second-order moments. 
By contrast, the quantity $D$ probes fourth-order fluctuations and consequently involves higher trace invariants, including $\tr(\hat{X}^2)$ and nonlinear combinations of $\tr(\hat{X})$. 
This hierarchy mirrors the moment order of the underlying Haar averages and highlights how increasingly refined spectral information of $\hat{X}$ enters at higher moments. More generally, at a fixed moment order, Haar averaging retains only those spectral invariants compatible with the corresponding permutation symmetries.

\paragraph{Concluding remarks.}
The derivation above illustrates a general mechanism underlying all Haar-averaged quantities considered in this work. While intermediate expressions involve sums over permutations and Weingarten coefficients, the final results depend only on a small number of trace invariants of the underlying unitary operator. This drastic simplification is a direct consequence of unitary invariance and the cycle structure of the symmetric group, as captured by the Schur--Weyl duality and the Weingarten calculus.

From this perspective, Proposition~\ref{prop:FD_closed_form_unitary} is not an isolated computation but an explicit instance of a systematic framework: Haar integration acts as a projection onto a finite-dimensional space of invariant tensors, whose coefficients are fixed uniquely by symmetry. The tools collected in this appendix provide a unified and reusable method for carrying out such projections, both for the quantities studied here and for higher-order moments or more general operator structures.


\begin{thebibliography}{22}%
\makeatletter
\providecommand \@ifxundefined [1]{%
 \@ifx{#1\undefined}
}%
\providecommand \@ifnum [1]{%
 \ifnum #1\expandafter \@firstoftwo
 \else \expandafter \@secondoftwo
 \fi
}%
\providecommand \@ifx [1]{%
 \ifx #1\expandafter \@firstoftwo
 \else \expandafter \@secondoftwo
 \fi
}%
\providecommand \natexlab [1]{#1}%
\providecommand \enquote  [1]{``#1''}%
\providecommand \bibnamefont  [1]{#1}%
\providecommand \bibfnamefont [1]{#1}%
\providecommand \citenamefont [1]{#1}%
\providecommand \href@noop [0]{\@secondoftwo}%
\providecommand \href [0]{\begingroup \@sanitize@url \@href}%
\providecommand \@href[1]{\@@startlink{#1}\@@href}%
\providecommand \@@href[1]{\endgroup#1\@@endlink}%
\providecommand \@sanitize@url [0]{\catcode `\\12\catcode `\$12\catcode
  `\&12\catcode `\#12\catcode `\^12\catcode `\_12\catcode `\%12\relax}%
\providecommand \@@startlink[1]{}%
\providecommand \@@endlink[0]{}%
\providecommand \url  [0]{\begingroup\@sanitize@url \@url }%
\providecommand \@url [1]{\endgroup\@href {#1}{\urlprefix }}%
\providecommand \urlprefix  [0]{URL }%
\providecommand \Eprint [0]{\href }%
\providecommand \doibase [0]{https://doi.org/}%
\providecommand \selectlanguage [0]{\@gobble}%
\providecommand \bibinfo  [0]{\@secondoftwo}%
\providecommand \bibfield  [0]{\@secondoftwo}%
\providecommand \translation [1]{[#1]}%
\providecommand \BibitemOpen [0]{}%
\providecommand \bibitemStop [0]{}%
\providecommand \bibitemNoStop [0]{.\EOS\space}%
\providecommand \EOS [0]{\spacefactor3000\relax}%
\providecommand \BibitemShut  [1]{\csname bibitem#1\endcsname}%
\let\auto@bib@innerbib\@empty
\bibitem [{\citenamefont {Steane}(1999)}]{Steane1999}%
  \BibitemOpen
  \bibfield  {author} {\bibinfo {author} {\bibfnamefont {A.~M.}\ \bibnamefont
  {Steane}},\ }\bibfield  {title} {\bibinfo {title} {Efficient fault-tolerant
  quantum computing},\ }\href@noop {} {\bibfield  {journal} {\bibinfo
  {journal} {Nature}\ }\textbf {\bibinfo {volume} {399}},\ \bibinfo {pages}
  {124} (\bibinfo {year} {1999})}\BibitemShut {NoStop}%
\bibitem [{\citenamefont {Martinis}(2015)}]{Martinis2015}%
  \BibitemOpen
  \bibfield  {author} {\bibinfo {author} {\bibfnamefont {J.~M.}\ \bibnamefont
  {Martinis}},\ }\bibfield  {title} {\bibinfo {title} {Qubit metrology for
  building a fault-tolerant quantum computer},\ }\href@noop {} {\bibfield
  {journal} {\bibinfo  {journal} {npj Quantum Information}\ }\textbf {\bibinfo
  {volume} {1}},\ \bibinfo {pages} {1} (\bibinfo {year} {2015})}\BibitemShut
  {NoStop}%
\bibitem [{\citenamefont {Postler}\ \emph {et~al.}(2022)\citenamefont
  {Postler}, \citenamefont {Heu$\beta$en}, \citenamefont {Pogorelov},
  \citenamefont {Rispler}, \citenamefont {Feldker}, \citenamefont {Meth},
  \citenamefont {Marciniak}, \citenamefont {Stricker}, \citenamefont
  {Ringbauer}, \citenamefont {Blatt} \emph {et~al.}}]{Postler2022}%
  \BibitemOpen
  \bibfield  {author} {\bibinfo {author} {\bibfnamefont {L.}~\bibnamefont
  {Postler}}, \bibinfo {author} {\bibfnamefont {S.}~\bibnamefont
  {Heu$\beta$en}}, \bibinfo {author} {\bibfnamefont {I.}~\bibnamefont
  {Pogorelov}}, \bibinfo {author} {\bibfnamefont {M.}~\bibnamefont {Rispler}},
  \bibinfo {author} {\bibfnamefont {T.}~\bibnamefont {Feldker}}, \bibinfo
  {author} {\bibfnamefont {M.}~\bibnamefont {Meth}}, \bibinfo {author}
  {\bibfnamefont {C.~D.}\ \bibnamefont {Marciniak}}, \bibinfo {author}
  {\bibfnamefont {R.}~\bibnamefont {Stricker}}, \bibinfo {author}
  {\bibfnamefont {M.}~\bibnamefont {Ringbauer}}, \bibinfo {author}
  {\bibfnamefont {R.}~\bibnamefont {Blatt}}, \emph {et~al.},\ }\bibfield
  {title} {\bibinfo {title} {Demonstration of fault-tolerant universal quantum
  gate operations},\ }\href@noop {} {\bibfield  {journal} {\bibinfo  {journal}
  {Nature}\ }\textbf {\bibinfo {volume} {605}},\ \bibinfo {pages} {675}
  (\bibinfo {year} {2022})}\BibitemShut {NoStop}%
\bibitem [{\citenamefont {Zhou}\ \emph {et~al.}(2025)\citenamefont {Zhou},
  \citenamefont {Cain},\ and\ \citenamefont {Lukin}}]{Zhou2025}%
  \BibitemOpen
  \bibfield  {author} {\bibinfo {author} {\bibfnamefont {H.}~\bibnamefont
  {Zhou}}, \bibinfo {author} {\bibfnamefont {M.}~\bibnamefont {Cain}},\ and\
  \bibinfo {author} {\bibfnamefont {M.~D.}\ \bibnamefont {Lukin}},\ }\bibfield
  {title} {\bibinfo {title} {Opportunities in full-stack design of low-overhead
  fault-tolerant quantum computation},\ }\href@noop {} {\bibfield  {journal}
  {\bibinfo  {journal} {Nature Computational Science}\ }\textbf {\bibinfo
  {volume} {5}},\ \bibinfo {pages} {1110} (\bibinfo {year} {2025})}\BibitemShut
  {NoStop}%
\bibitem [{\citenamefont {Kitaev}\ \emph {et~al.}(2002)\citenamefont {Kitaev},
  \citenamefont {Shen},\ and\ \citenamefont {Vyalyi}}]{Kitaev2002}%
  \BibitemOpen
  \bibfield  {author} {\bibinfo {author} {\bibfnamefont {A.~Y.}\ \bibnamefont
  {Kitaev}}, \bibinfo {author} {\bibfnamefont {A.}~\bibnamefont {Shen}},\ and\
  \bibinfo {author} {\bibfnamefont {M.~N.}\ \bibnamefont {Vyalyi}},\
  }\href@noop {} {\emph {\bibinfo {title} {Classical and quantum
  computation}}},\ \bibinfo {number} {47}\ (\bibinfo  {publisher} {American
  Mathematical Soc.},\ \bibinfo {year} {2002})\BibitemShut {NoStop}%
\bibitem [{\citenamefont {Gilchrist}\ \emph {et~al.}(2005)\citenamefont
  {Gilchrist}, \citenamefont {Langford},\ and\ \citenamefont
  {Nielsen}}]{Gilchrist2005}%
  \BibitemOpen
  \bibfield  {author} {\bibinfo {author} {\bibfnamefont {A.}~\bibnamefont
  {Gilchrist}}, \bibinfo {author} {\bibfnamefont {N.~K.}\ \bibnamefont
  {Langford}},\ and\ \bibinfo {author} {\bibfnamefont {M.~A.}\ \bibnamefont
  {Nielsen}},\ }\bibfield  {title} {\bibinfo {title} {Distance measures to
  compare real and ideal quantum processes},\ }\href@noop {} {\bibfield
  {journal} {\bibinfo  {journal} {Physical Review A}\ }\textbf {\bibinfo
  {volume} {71}},\ \bibinfo {pages} {062310} (\bibinfo {year}
  {2005})}\BibitemShut {NoStop}%
\bibitem [{\citenamefont {Nielsen}(2002)}]{Nielsen2002}%
  \BibitemOpen
  \bibfield  {author} {\bibinfo {author} {\bibfnamefont {M.~A.}\ \bibnamefont
  {Nielsen}},\ }\bibfield  {title} {\bibinfo {title} {A simple formula for the
  average gate fidelity of a quantum dynamical operation},\ }\href@noop {}
  {\bibfield  {journal} {\bibinfo  {journal} {Physics Letters A}\ }\textbf
  {\bibinfo {volume} {303}},\ \bibinfo {pages} {249} (\bibinfo {year}
  {2002})}\BibitemShut {NoStop}%
\bibitem [{\citenamefont {Knill}\ \emph {et~al.}(2008)\citenamefont {Knill},
  \citenamefont {Leibfried}, \citenamefont {Reichle}, \citenamefont {Britton},
  \citenamefont {Blakestad}, \citenamefont {Jost}, \citenamefont {Langer},
  \citenamefont {Ozeri}, \citenamefont {Seidelin},\ and\ \citenamefont
  {Wineland}}]{Knill2008}%
  \BibitemOpen
  \bibfield  {author} {\bibinfo {author} {\bibfnamefont {E.}~\bibnamefont
  {Knill}}, \bibinfo {author} {\bibfnamefont {D.}~\bibnamefont {Leibfried}},
  \bibinfo {author} {\bibfnamefont {R.}~\bibnamefont {Reichle}}, \bibinfo
  {author} {\bibfnamefont {J.}~\bibnamefont {Britton}}, \bibinfo {author}
  {\bibfnamefont {R.~B.}\ \bibnamefont {Blakestad}}, \bibinfo {author}
  {\bibfnamefont {J.~D.}\ \bibnamefont {Jost}}, \bibinfo {author}
  {\bibfnamefont {C.}~\bibnamefont {Langer}}, \bibinfo {author} {\bibfnamefont
  {R.}~\bibnamefont {Ozeri}}, \bibinfo {author} {\bibfnamefont
  {S.}~\bibnamefont {Seidelin}},\ and\ \bibinfo {author} {\bibfnamefont
  {D.~J.}\ \bibnamefont {Wineland}},\ }\bibfield  {title} {\bibinfo {title}
  {Randomized benchmarking of quantum gates},\ }\href@noop {} {\bibfield
  {journal} {\bibinfo  {journal} {Physical Review A}\ }\textbf {\bibinfo
  {volume} {77}},\ \bibinfo {pages} {012307} (\bibinfo {year}
  {2008})}\BibitemShut {NoStop}%
\bibitem [{\citenamefont {Magesan}\ \emph {et~al.}(2011)\citenamefont
  {Magesan}, \citenamefont {Gambetta},\ and\ \citenamefont
  {Emerson}}]{Magesan2011}%
  \BibitemOpen
  \bibfield  {author} {\bibinfo {author} {\bibfnamefont {E.}~\bibnamefont
  {Magesan}}, \bibinfo {author} {\bibfnamefont {J.~M.}\ \bibnamefont
  {Gambetta}},\ and\ \bibinfo {author} {\bibfnamefont {J.}~\bibnamefont
  {Emerson}},\ }\bibfield  {title} {\bibinfo {title} {Scalable and robust
  randomized benchmarking of quantum processes},\ }\href@noop {} {\bibfield
  {journal} {\bibinfo  {journal} {Physical Review Letters}\ }\textbf {\bibinfo
  {volume} {106}},\ \bibinfo {pages} {180504} (\bibinfo {year}
  {2011})}\BibitemShut {NoStop}%
\bibitem [{\citenamefont {Sanders}\ \emph {et~al.}(2016)\citenamefont
  {Sanders}, \citenamefont {Wallman},\ and\ \citenamefont
  {Sanders}}]{Sanders2016}%
  \BibitemOpen
  \bibfield  {author} {\bibinfo {author} {\bibfnamefont {Y.~R.}\ \bibnamefont
  {Sanders}}, \bibinfo {author} {\bibfnamefont {J.~J.}\ \bibnamefont
  {Wallman}},\ and\ \bibinfo {author} {\bibfnamefont {B.~C.}\ \bibnamefont
  {Sanders}},\ }\bibfield  {title} {\bibinfo {title} {Bounding quantum gate
  error rate based on reported average fidelity},\ }\href@noop {} {\bibfield
  {journal} {\bibinfo  {journal} {New Journal of Physics}\ }\textbf {\bibinfo
  {volume} {18}},\ \bibinfo {pages} {012002} (\bibinfo {year}
  {2016})}\BibitemShut {NoStop}%
\bibitem [{\citenamefont {Kueng}\ \emph {et~al.}(2016)\citenamefont {Kueng},
  \citenamefont {Long}, \citenamefont {Doherty},\ and\ \citenamefont
  {Flammia}}]{Kueng2016}%
  \BibitemOpen
  \bibfield  {author} {\bibinfo {author} {\bibfnamefont {R.}~\bibnamefont
  {Kueng}}, \bibinfo {author} {\bibfnamefont {D.~M.}\ \bibnamefont {Long}},
  \bibinfo {author} {\bibfnamefont {A.~C.}\ \bibnamefont {Doherty}},\ and\
  \bibinfo {author} {\bibfnamefont {S.~T.}\ \bibnamefont {Flammia}},\
  }\bibfield  {title} {\bibinfo {title} {Comparing experiments to the
  fault-tolerance threshold},\ }\href@noop {} {\bibfield  {journal} {\bibinfo
  {journal} {Physical Review Letters}\ }\textbf {\bibinfo {volume} {117}},\
  \bibinfo {pages} {170502} (\bibinfo {year} {2016})}\BibitemShut {NoStop}%
\bibitem [{\citenamefont {O'Brien}\ \emph {et~al.}(2004)\citenamefont
  {O'Brien}, \citenamefont {Pryde}, \citenamefont {Gilchrist}, \citenamefont
  {James}, \citenamefont {Langford}, \citenamefont {Ralph},\ and\ \citenamefont
  {White}}]{OBrien2004}%
  \BibitemOpen
  \bibfield  {author} {\bibinfo {author} {\bibfnamefont {J.~L.}\ \bibnamefont
  {O'Brien}}, \bibinfo {author} {\bibfnamefont {G.~J.}\ \bibnamefont {Pryde}},
  \bibinfo {author} {\bibfnamefont {A.}~\bibnamefont {Gilchrist}}, \bibinfo
  {author} {\bibfnamefont {D.~F.}\ \bibnamefont {James}}, \bibinfo {author}
  {\bibfnamefont {N.~K.}\ \bibnamefont {Langford}}, \bibinfo {author}
  {\bibfnamefont {T.~C.}\ \bibnamefont {Ralph}},\ and\ \bibinfo {author}
  {\bibfnamefont {A.~G.}\ \bibnamefont {White}},\ }\bibfield  {title} {\bibinfo
  {title} {Quantum process tomography of a controlled-not gate},\ }\href@noop
  {} {\bibfield  {journal} {\bibinfo  {journal} {Physical Review Letters}\
  }\textbf {\bibinfo {volume} {93}},\ \bibinfo {pages} {080502} (\bibinfo
  {year} {2004})}\BibitemShut {NoStop}%
\bibitem [{\citenamefont {Blume-Kohout}\ \emph {et~al.}(2017)\citenamefont
  {Blume-Kohout}, \citenamefont {Gamble}, \citenamefont {Nielsen},
  \citenamefont {Rudinger}, \citenamefont {Mizrahi}, \citenamefont {Fortier},\
  and\ \citenamefont {Maunz}}]{Blume2017}%
  \BibitemOpen
  \bibfield  {author} {\bibinfo {author} {\bibfnamefont {R.}~\bibnamefont
  {Blume-Kohout}}, \bibinfo {author} {\bibfnamefont {J.~K.}\ \bibnamefont
  {Gamble}}, \bibinfo {author} {\bibfnamefont {E.}~\bibnamefont {Nielsen}},
  \bibinfo {author} {\bibfnamefont {K.}~\bibnamefont {Rudinger}}, \bibinfo
  {author} {\bibfnamefont {J.}~\bibnamefont {Mizrahi}}, \bibinfo {author}
  {\bibfnamefont {K.}~\bibnamefont {Fortier}},\ and\ \bibinfo {author}
  {\bibfnamefont {P.}~\bibnamefont {Maunz}},\ }\bibfield  {title} {\bibinfo
  {title} {Demonstration of qubit operations below a rigorous fault tolerance
  threshold with gate set tomography},\ }\href@noop {} {\bibfield  {journal}
  {\bibinfo  {journal} {Nature communications}\ }\textbf {\bibinfo {volume}
  {8}},\ \bibinfo {pages} {14485} (\bibinfo {year} {2017})}\BibitemShut
  {NoStop}%
\bibitem [{\citenamefont {Dankert}\ \emph {et~al.}(2009)\citenamefont
  {Dankert}, \citenamefont {Cleve}, \citenamefont {Emerson},\ and\
  \citenamefont {Livine}}]{Dankert2009}%
  \BibitemOpen
  \bibfield  {author} {\bibinfo {author} {\bibfnamefont {C.}~\bibnamefont
  {Dankert}}, \bibinfo {author} {\bibfnamefont {R.}~\bibnamefont {Cleve}},
  \bibinfo {author} {\bibfnamefont {J.}~\bibnamefont {Emerson}},\ and\ \bibinfo
  {author} {\bibfnamefont {E.}~\bibnamefont {Livine}},\ }\bibfield  {title}
  {\bibinfo {title} {Exact and approximate unitary 2-designs and their
  application to fidelity estimation},\ }\href@noop {} {\bibfield  {journal}
  {\bibinfo  {journal} {Physical Review A}\ }\textbf {\bibinfo {volume} {80}},\
  \bibinfo {pages} {012304} (\bibinfo {year} {2009})}\BibitemShut {NoStop}%
\bibitem [{\citenamefont {Harrow}\ and\ \citenamefont
  {Low}(2009)}]{Harrow2009}%
  \BibitemOpen
  \bibfield  {author} {\bibinfo {author} {\bibfnamefont {A.~W.}\ \bibnamefont
  {Harrow}}\ and\ \bibinfo {author} {\bibfnamefont {R.~A.}\ \bibnamefont
  {Low}},\ }\bibfield  {title} {\bibinfo {title} {Random quantum circuits are
  approximate 2-designs},\ }\href@noop {} {\bibfield  {journal} {\bibinfo
  {journal} {Communications in Mathematical Physics}\ }\textbf {\bibinfo
  {volume} {291}},\ \bibinfo {pages} {257} (\bibinfo {year}
  {2009})}\BibitemShut {NoStop}%
\bibitem [{\citenamefont {Brandao}\ \emph {et~al.}(2016)\citenamefont
  {Brandao}, \citenamefont {Harrow},\ and\ \citenamefont
  {Horodecki}}]{Brandao2016}%
  \BibitemOpen
  \bibfield  {author} {\bibinfo {author} {\bibfnamefont {F.~G.}\ \bibnamefont
  {Brandao}}, \bibinfo {author} {\bibfnamefont {A.~W.}\ \bibnamefont
  {Harrow}},\ and\ \bibinfo {author} {\bibfnamefont {M.}~\bibnamefont
  {Horodecki}},\ }\bibfield  {title} {\bibinfo {title} {Local random quantum
  circuits are approximate polynomial-designs},\ }\href@noop {} {\bibfield
  {journal} {\bibinfo  {journal} {Communications in Mathematical Physics}\
  }\textbf {\bibinfo {volume} {346}},\ \bibinfo {pages} {397} (\bibinfo {year}
  {2016})}\BibitemShut {NoStop}%
\bibitem [{\citenamefont {Selinger}(2013)}]{Selinger2013}%
  \BibitemOpen
  \bibfield  {author} {\bibinfo {author} {\bibfnamefont {P.}~\bibnamefont
  {Selinger}},\ }\bibfield  {title} {\bibinfo {title} {Quantum circuits of
  t-depth one},\ }\href@noop {} {\bibfield  {journal} {\bibinfo  {journal}
  {Physical Review A}\ }\textbf {\bibinfo {volume} {87}},\ \bibinfo {pages}
  {042302} (\bibinfo {year} {2013})}\BibitemShut {NoStop}%
\bibitem [{\citenamefont {Wallman}\ \emph {et~al.}(2015)\citenamefont
  {Wallman}, \citenamefont {Granade}, \citenamefont {Harper},\ and\
  \citenamefont {Flammia}}]{Wallman2015}%
  \BibitemOpen
  \bibfield  {author} {\bibinfo {author} {\bibfnamefont {J.}~\bibnamefont
  {Wallman}}, \bibinfo {author} {\bibfnamefont {C.}~\bibnamefont {Granade}},
  \bibinfo {author} {\bibfnamefont {R.}~\bibnamefont {Harper}},\ and\ \bibinfo
  {author} {\bibfnamefont {S.~T.}\ \bibnamefont {Flammia}},\ }\bibfield
  {title} {\bibinfo {title} {Estimating the coherence of noise},\ }\href@noop
  {} {\bibfield  {journal} {\bibinfo  {journal} {New Journal of Physics}\
  }\textbf {\bibinfo {volume} {17}},\ \bibinfo {pages} {113020} (\bibinfo
  {year} {2015})}\BibitemShut {NoStop}%
\bibitem [{\citenamefont {Collins}\ and\ \citenamefont
  {{\'S}niady}(2006)}]{Collins2006}%
  \BibitemOpen
  \bibfield  {author} {\bibinfo {author} {\bibfnamefont {B.}~\bibnamefont
  {Collins}}\ and\ \bibinfo {author} {\bibfnamefont {P.}~\bibnamefont
  {{\'S}niady}},\ }\bibfield  {title} {\bibinfo {title} {{Integration with
  Respect to the Haar Measure on Unitary, Orthogonal and Symplectic Group}},\
  }\href@noop {} {\bibfield  {journal} {\bibinfo  {journal} {Communications in
  Mathematical Physics}\ }\textbf {\bibinfo {volume} {264}},\ \bibinfo {pages}
  {773} (\bibinfo {year} {2006})}\BibitemShut {NoStop}%
\bibitem [{\citenamefont {Mele}(2024)}]{Mele2024}%
  \BibitemOpen
  \bibfield  {author} {\bibinfo {author} {\bibfnamefont {A.~A.}\ \bibnamefont
  {Mele}},\ }\bibfield  {title} {\bibinfo {title} {Introduction to {H}aar
  {M}easure {T}ools in {Q}uantum {I}nformation: {A} {B}eginner's {T}utorial},\
  }\href@noop {} {\bibfield  {journal} {\bibinfo  {journal} {{Quantum}}\
  }\textbf {\bibinfo {volume} {8}},\ \bibinfo {pages} {1340} (\bibinfo {year}
  {2024})}\BibitemShut {NoStop}%
\bibitem [{\citenamefont {Cho}\ and\ \citenamefont {Bang}(2026)}]{Cho2026}%
  \BibitemOpen
  \bibfield  {author} {\bibinfo {author} {\bibfnamefont {K.}~\bibnamefont
  {Cho}}\ and\ \bibinfo {author} {\bibfnamefont {J.}~\bibnamefont {Bang}},\
  }\bibfield  {title} {\bibinfo {title} {Entangling power and its deviation: A
  quantitative analysis on input-state dependence and variability in
  entanglement generation},\ }\href@noop {} {\bibfield  {journal} {\bibinfo
  {journal} {Physical Review A}\ }\textbf {\bibinfo {volume} {113}},\ \bibinfo
  {pages} {012442} (\bibinfo {year} {2026})}\BibitemShut {NoStop}%
\bibitem [{\citenamefont {Hoeffding}(1992)}]{Hoeffding1992}%
  \BibitemOpen
  \bibfield  {author} {\bibinfo {author} {\bibfnamefont {W.}~\bibnamefont
  {Hoeffding}},\ }\bibfield  {title} {\bibinfo {title} {A class of statistics
  with asymptotically normal distribution},\ }in\ \href@noop {} {\emph
  {\bibinfo {booktitle} {Breakthroughs in statistics: Foundations and basic
  theory}}}\ (\bibinfo  {publisher} {Springer},\ \bibinfo {year} {1992})\ pp.\
  \bibinfo {pages} {308--334}\BibitemShut {NoStop}%
\end{thebibliography}

\begin{thebibliography}{28}%
\makeatletter
\providecommand \@ifxundefined [1]{%
 \@ifx{#1\undefined}
}%
\providecommand \@ifnum [1]{%
 \ifnum #1\expandafter \@firstoftwo
 \else \expandafter \@secondoftwo
 \fi
}%
\providecommand \@ifx [1]{%
 \ifx #1\expandafter \@firstoftwo
 \else \expandafter \@secondoftwo
 \fi
}%
\providecommand \natexlab [1]{#1}%
\providecommand \enquote  [1]{``#1''}%
\providecommand \bibnamefont  [1]{#1}%
\providecommand \bibfnamefont [1]{#1}%
\providecommand \citenamefont [1]{#1}%
\providecommand \href@noop [0]{\@secondoftwo}%
\providecommand \href [0]{\begingroup \@sanitize@url \@href}%
\providecommand \@href[1]{\@@startlink{#1}\@@href}%
\providecommand \@@href[1]{\endgroup#1\@@endlink}%
\providecommand \@sanitize@url [0]{\catcode `\\12\catcode `\$12\catcode
  `\&12\catcode `\#12\catcode `\^12\catcode `\_12\catcode `\%12\relax}%
\providecommand \@@startlink[1]{}%
\providecommand \@@endlink[0]{}%
\providecommand \url  [0]{\begingroup\@sanitize@url \@url }%
\providecommand \@url [1]{\endgroup\@href {#1}{\urlprefix }}%
\providecommand \urlprefix  [0]{URL }%
\providecommand \Eprint [0]{\href }%
\providecommand \doibase [0]{http://dx.doi.org/}%
\providecommand \selectlanguage [0]{\@gobble}%
\providecommand \bibinfo  [0]{\@secondoftwo}%
\providecommand \bibfield  [0]{\@secondoftwo}%
\providecommand \translation [1]{[#1]}%
\providecommand \BibitemOpen [0]{}%
\providecommand \bibitemStop [0]{}%
\providecommand \bibitemNoStop [0]{.\EOS\space}%
\providecommand \EOS [0]{\spacefactor3000\relax}%
\providecommand \BibitemShut  [1]{\csname bibitem#1\endcsname}%
\let\auto@bib@innerbib\@empty
\bibitem [{\citenamefont {Fowler}(2004)}]{fowler2004constructing}%
  \BibitemOpen
  \bibfield  {author} {\bibinfo {author} {\bibfnamefont {A.~G.}\ \bibnamefont
  {Fowler}},\ }\href@noop {} {\bibfield  {journal} {\bibinfo  {journal} {arXiv
  preprint quant-ph/0411206}\ } (\bibinfo {year} {2004})}\BibitemShut {NoStop}%
\bibitem [{\citenamefont {Dorit~Aharonov}(2008)}]{aharonov2008}%
  \BibitemOpen
  \bibfield  {author} {\bibinfo {author} {\bibfnamefont {M.~B.-O.}\
  \bibnamefont {Dorit~Aharonov}},\ }\href@noop {} {\bibfield  {journal}
  {\bibinfo  {journal} {SIAM Journal of Computing}\ }\textbf {\bibinfo {volume}
  {38}},\ \bibinfo {pages} {1207} (\bibinfo {year} {2008})}\BibitemShut
  {NoStop}%
\bibitem [{\citenamefont {Gilchrist}\ \emph {et~al.}(2005)\citenamefont
  {Gilchrist}, \citenamefont {Langford},\ and\ \citenamefont
  {Nielsen}}]{Gilchrist2005}%
  \BibitemOpen
  \bibfield  {author} {\bibinfo {author} {\bibfnamefont {A.}~\bibnamefont
  {Gilchrist}}, \bibinfo {author} {\bibfnamefont {N.~K.}\ \bibnamefont
  {Langford}}, \ and\ \bibinfo {author} {\bibfnamefont {M.~A.}\ \bibnamefont
  {Nielsen}},\ }\href@noop {} {\bibfield  {journal} {\bibinfo  {journal}
  {Physical Review A}\ }\textbf {\bibinfo {volume} {71}},\ \bibinfo {pages}
  {062310} (\bibinfo {year} {2005})}\BibitemShut {NoStop}%
\bibitem [{\citenamefont {Panos~Aliferis}(2006)}]{aliferis2006}%
  \BibitemOpen
  \bibfield  {author} {\bibinfo {author} {\bibfnamefont {J.~P.}\ \bibnamefont
  {Panos~Aliferis}, \bibfnamefont {Daniel~Gottesman}},\ }\href@noop {}
  {\bibfield  {journal} {\bibinfo  {journal} {Quantum Information \&
  Computation}\ }\textbf {\bibinfo {volume} {6}},\ \bibinfo {pages} {97}
  (\bibinfo {year} {2006})}\BibitemShut {NoStop}%
\bibitem [{\citenamefont {Kueng}\ \emph {et~al.}(2016)\citenamefont {Kueng},
  \citenamefont {Long}, \citenamefont {Doherty},\ and\ \citenamefont
  {Flammia}}]{Kueng2016}%
  \BibitemOpen
  \bibfield  {author} {\bibinfo {author} {\bibfnamefont {R.}~\bibnamefont
  {Kueng}}, \bibinfo {author} {\bibfnamefont {D.~M.}\ \bibnamefont {Long}},
  \bibinfo {author} {\bibfnamefont {A.~C.}\ \bibnamefont {Doherty}}, \ and\
  \bibinfo {author} {\bibfnamefont {S.~T.}\ \bibnamefont {Flammia}},\
  }\href@noop {} {\bibfield  {journal} {\bibinfo  {journal} {Physical Review
  Letters}\ }\textbf {\bibinfo {volume} {117}},\ \bibinfo {pages} {170502}
  (\bibinfo {year} {2016})}\BibitemShut {NoStop}%
\bibitem [{\citenamefont {Nielsen}(2002)}]{nielsen2002simple}%
  \BibitemOpen
  \bibfield  {author} {\bibinfo {author} {\bibfnamefont {M.~A.}\ \bibnamefont
  {Nielsen}},\ }\href@noop {} {\bibfield  {journal} {\bibinfo  {journal}
  {Physics Letters A}\ }\textbf {\bibinfo {volume} {303}},\ \bibinfo {pages}
  {249} (\bibinfo {year} {2002})}\BibitemShut {NoStop}%
\bibitem [{\citenamefont {Knill}\ \emph {et~al.}(2008)\citenamefont {Knill},
  \citenamefont {Leibfried}, \citenamefont {Reichle}, \citenamefont {Britton},
  \citenamefont {Blakestad}, \citenamefont {Jost}, \citenamefont {Langer},
  \citenamefont {Ozeri}, \citenamefont {Seidelin},\ and\ \citenamefont
  {Wineland}}]{PhysRevA.77.012307}%
  \BibitemOpen
  \bibfield  {author} {\bibinfo {author} {\bibfnamefont {E.}~\bibnamefont
  {Knill}}, \bibinfo {author} {\bibfnamefont {D.}~\bibnamefont {Leibfried}},
  \bibinfo {author} {\bibfnamefont {R.}~\bibnamefont {Reichle}}, \bibinfo
  {author} {\bibfnamefont {J.}~\bibnamefont {Britton}}, \bibinfo {author}
  {\bibfnamefont {R.~B.}\ \bibnamefont {Blakestad}}, \bibinfo {author}
  {\bibfnamefont {J.~D.}\ \bibnamefont {Jost}}, \bibinfo {author}
  {\bibfnamefont {C.}~\bibnamefont {Langer}}, \bibinfo {author} {\bibfnamefont
  {R.}~\bibnamefont {Ozeri}}, \bibinfo {author} {\bibfnamefont
  {S.}~\bibnamefont {Seidelin}}, \ and\ \bibinfo {author} {\bibfnamefont
  {D.~J.}\ \bibnamefont {Wineland}},\ }\href {\doibase
  10.1103/PhysRevA.77.012307} {\bibfield  {journal} {\bibinfo  {journal} {Phys.
  Rev. A}\ }\textbf {\bibinfo {volume} {77}},\ \bibinfo {pages} {012307}
  (\bibinfo {year} {2008})}\BibitemShut {NoStop}%
\bibitem [{\citenamefont {Magesan}\ \emph {et~al.}(2011)\citenamefont
  {Magesan}, \citenamefont {Gambetta},\ and\ \citenamefont
  {Emerson}}]{magesan2011scalable}%
  \BibitemOpen
  \bibfield  {author} {\bibinfo {author} {\bibfnamefont {E.}~\bibnamefont
  {Magesan}}, \bibinfo {author} {\bibfnamefont {J.~M.}\ \bibnamefont
  {Gambetta}}, \ and\ \bibinfo {author} {\bibfnamefont {J.}~\bibnamefont
  {Emerson}},\ }\href@noop {} {\bibfield  {journal} {\bibinfo  {journal}
  {Physical review letters}\ }\textbf {\bibinfo {volume} {106}},\ \bibinfo
  {pages} {180504} (\bibinfo {year} {2011})}\BibitemShut {NoStop}%
\bibitem [{\citenamefont {Kitaev}(1997)}]{kitaev1997}%
  \BibitemOpen
  \bibfield  {author} {\bibinfo {author} {\bibfnamefont {A.~Y.}\ \bibnamefont
  {Kitaev}},\ }\href@noop {} {\bibfield  {journal} {\bibinfo  {journal}
  {Russian Mathematical Surveys}\ }\textbf {\bibinfo {volume} {52}},\ \bibinfo
  {pages} {1191} (\bibinfo {year} {1997})}\BibitemShut {NoStop}%
\bibitem [{\citenamefont {Sanders}\ \emph {et~al.}(2015)\citenamefont
  {Sanders}, \citenamefont {Wallman},\ and\ \citenamefont
  {Sanders}}]{sanders2015bounding}%
  \BibitemOpen
  \bibfield  {author} {\bibinfo {author} {\bibfnamefont {Y.~R.}\ \bibnamefont
  {Sanders}}, \bibinfo {author} {\bibfnamefont {J.~J.}\ \bibnamefont
  {Wallman}}, \ and\ \bibinfo {author} {\bibfnamefont {B.~C.}\ \bibnamefont
  {Sanders}},\ }\href@noop {} {\bibfield  {journal} {\bibinfo  {journal} {New
  Journal of Physics}\ }\textbf {\bibinfo {volume} {18}},\ \bibinfo {pages}
  {012002} (\bibinfo {year} {2015})}\BibitemShut {NoStop}%
\bibitem [{Note1()}]{Note1}%
  \BibitemOpen
  \bibinfo {note} {One way to see this is via the SDP characterization of the
  diamond norm and the fact that the maximally entangled input captures the
  Choi operator, together with a dimension factor converting between
  input-state optimizations and the Choi representation.}\BibitemShut {Stop}%
\bibitem [{\citenamefont {Fuchs}\ and\ \citenamefont {Van
  De~Graaf}(2002)}]{fuchs2002cryptographic}%
  \BibitemOpen
  \bibfield  {author} {\bibinfo {author} {\bibfnamefont {C.~A.}\ \bibnamefont
  {Fuchs}}\ and\ \bibinfo {author} {\bibfnamefont {J.}~\bibnamefont {Van
  De~Graaf}},\ }\href@noop {} {\bibfield  {journal} {\bibinfo  {journal} {IEEE
  Transactions on Information Theory}\ }\textbf {\bibinfo {volume} {45}},\
  \bibinfo {pages} {1216} (\bibinfo {year} {2002})}\BibitemShut {NoStop}%
\bibitem [{\citenamefont {Harrow}\ and\ \citenamefont
  {Low}(2009)}]{harrow2009random}%
  \BibitemOpen
  \bibfield  {author} {\bibinfo {author} {\bibfnamefont {A.~W.}\ \bibnamefont
  {Harrow}}\ and\ \bibinfo {author} {\bibfnamefont {R.~A.}\ \bibnamefont
  {Low}},\ }\href@noop {} {\bibfield  {journal} {\bibinfo  {journal}
  {Communications in Mathematical Physics}\ }\textbf {\bibinfo {volume}
  {291}},\ \bibinfo {pages} {257} (\bibinfo {year} {2009})}\BibitemShut
  {NoStop}%
\bibitem [{\citenamefont {Brandao}\ \emph {et~al.}(2016)\citenamefont
  {Brandao}, \citenamefont {Harrow},\ and\ \citenamefont
  {Horodecki}}]{brandao2016local}%
  \BibitemOpen
  \bibfield  {author} {\bibinfo {author} {\bibfnamefont {F.~G.}\ \bibnamefont
  {Brandao}}, \bibinfo {author} {\bibfnamefont {A.~W.}\ \bibnamefont {Harrow}},
  \ and\ \bibinfo {author} {\bibfnamefont {M.}~\bibnamefont {Horodecki}},\
  }\href@noop {} {\bibfield  {journal} {\bibinfo  {journal} {Communications in
  Mathematical Physics}\ }\textbf {\bibinfo {volume} {346}},\ \bibinfo {pages}
  {397} (\bibinfo {year} {2016})}\BibitemShut {NoStop}%
\bibitem [{\citenamefont {Dankert}\ \emph {et~al.}(2009)\citenamefont
  {Dankert}, \citenamefont {Cleve}, \citenamefont {Emerson},\ and\
  \citenamefont {Livine}}]{dankert2009exact}%
  \BibitemOpen
  \bibfield  {author} {\bibinfo {author} {\bibfnamefont {C.}~\bibnamefont
  {Dankert}}, \bibinfo {author} {\bibfnamefont {R.}~\bibnamefont {Cleve}},
  \bibinfo {author} {\bibfnamefont {J.}~\bibnamefont {Emerson}}, \ and\
  \bibinfo {author} {\bibfnamefont {E.}~\bibnamefont {Livine}},\ }\href@noop {}
  {\bibfield  {journal} {\bibinfo  {journal} {Physical Review A---Atomic,
  Molecular, and Optical Physics}\ }\textbf {\bibinfo {volume} {80}},\ \bibinfo
  {pages} {012304} (\bibinfo {year} {2009})}\BibitemShut {NoStop}%
\bibitem [{\citenamefont {Hoeffding}(1992)}]{hoeffding1992class}%
  \BibitemOpen
  \bibfield  {author} {\bibinfo {author} {\bibfnamefont {W.}~\bibnamefont
  {Hoeffding}},\ }in\ \href@noop {} {\emph {\bibinfo {booktitle} {Breakthroughs
  in statistics: Foundations and basic theory}}}\ (\bibinfo  {publisher}
  {Springer},\ \bibinfo {year} {1992})\ pp.\ \bibinfo {pages}
  {308--334}\BibitemShut {NoStop}%
\bibitem [{\citenamefont {Serfling}(2009)}]{serfling2009approximation}%
  \BibitemOpen
  \bibfield  {author} {\bibinfo {author} {\bibfnamefont {R.~J.}\ \bibnamefont
  {Serfling}},\ }\href@noop {} {\emph {\bibinfo {title} {Approximation theorems
  of mathematical statistics}}}\ (\bibinfo  {publisher} {John Wiley \& Sons},\
  \bibinfo {year} {2009})\BibitemShut {NoStop}%
\bibitem [{\citenamefont {Alsmeyer}(2014)}]{Alsmeyer2014Chebyshev}%
  \BibitemOpen
  \bibfield  {author} {\bibinfo {author} {\bibfnamefont {G.}~\bibnamefont
  {Alsmeyer}},\ }in\ \href@noop {} {\emph {\bibinfo {booktitle} {International
  Encyclopedia of Statistical Science}}}\ (\bibinfo  {publisher} {Springer},\
  \bibinfo {year} {2014})\ pp.\ \bibinfo {pages} {239--240}\BibitemShut
  {NoStop}%
\bibitem [{\citenamefont {Toeplitz}(1918)}]{toeplitz1918algebraische}%
  \BibitemOpen
  \bibfield  {author} {\bibinfo {author} {\bibfnamefont {O.}~\bibnamefont
  {Toeplitz}},\ }\href@noop {} {\bibfield  {journal} {\bibinfo  {journal}
  {Mathematische Zeitschrift}\ }\textbf {\bibinfo {volume} {2}},\ \bibinfo
  {pages} {187} (\bibinfo {year} {1918})}\BibitemShut {NoStop}%
\bibitem [{\citenamefont {Halmos}(2012)}]{halmos2012hilbert}%
  \BibitemOpen
  \bibfield  {author} {\bibinfo {author} {\bibfnamefont {P.~R.}\ \bibnamefont
  {Halmos}},\ }\href@noop {} {\emph {\bibinfo {title} {A Hilbert space problem
  book}}},\ Vol.~\bibinfo {volume} {19}\ (\bibinfo  {publisher} {Springer
  Science \& Business Media},\ \bibinfo {year} {2012})\BibitemShut {NoStop}%
\bibitem [{\citenamefont {Selinger}(2013)}]{Selinger2013}%
  \BibitemOpen
  \bibfield  {author} {\bibinfo {author} {\bibfnamefont {P.}~\bibnamefont
  {Selinger}},\ }\href@noop {} {\bibfield  {journal} {\bibinfo  {journal}
  {Physical Review A}\ }\textbf {\bibinfo {volume} {87}},\ \bibinfo {pages}
  {042302} (\bibinfo {year} {2013})}\BibitemShut {NoStop}%
\bibitem [{\citenamefont {Amy}\ \emph {et~al.}(2013)\citenamefont {Amy},
  \citenamefont {Maslov}, \citenamefont {Mosca},\ and\ \citenamefont
  {Roetteler}}]{amy2013meet}%
  \BibitemOpen
  \bibfield  {author} {\bibinfo {author} {\bibfnamefont {M.}~\bibnamefont
  {Amy}}, \bibinfo {author} {\bibfnamefont {D.}~\bibnamefont {Maslov}},
  \bibinfo {author} {\bibfnamefont {M.}~\bibnamefont {Mosca}}, \ and\ \bibinfo
  {author} {\bibfnamefont {M.}~\bibnamefont {Roetteler}},\ }\href@noop {}
  {\bibfield  {journal} {\bibinfo  {journal} {IEEE Transactions on
  Computer-Aided Design of Integrated Circuits and Systems}\ }\textbf {\bibinfo
  {volume} {32}},\ \bibinfo {pages} {818} (\bibinfo {year} {2013})}\BibitemShut
  {NoStop}%
\bibitem [{\citenamefont {Wallman}\ \emph {et~al.}(2015)\citenamefont
  {Wallman}, \citenamefont {Granade}, \citenamefont {Harper},\ and\
  \citenamefont {Flammia}}]{wallman2015estimating}%
  \BibitemOpen
  \bibfield  {author} {\bibinfo {author} {\bibfnamefont {J.}~\bibnamefont
  {Wallman}}, \bibinfo {author} {\bibfnamefont {C.}~\bibnamefont {Granade}},
  \bibinfo {author} {\bibfnamefont {R.}~\bibnamefont {Harper}}, \ and\ \bibinfo
  {author} {\bibfnamefont {S.~T.}\ \bibnamefont {Flammia}},\ }\href@noop {}
  {\bibfield  {journal} {\bibinfo  {journal} {New Journal of Physics}\ }\textbf
  {\bibinfo {volume} {17}},\ \bibinfo {pages} {113020} (\bibinfo {year}
  {2015})}\BibitemShut {NoStop}%
\bibitem [{\citenamefont {Cho}\ and\ \citenamefont
  {Bang}(2026)}]{cho2025entangling}%
  \BibitemOpen
  \bibfield  {author} {\bibinfo {author} {\bibfnamefont {K.}~\bibnamefont
  {Cho}}\ and\ \bibinfo {author} {\bibfnamefont {J.}~\bibnamefont {Bang}},\
  }\href@noop {} {\bibfield  {journal} {\bibinfo  {journal} {Physical Review
  A}\ }\textbf {\bibinfo {volume} {113}},\ \bibinfo {pages} {012442} (\bibinfo
  {year} {2026})}\BibitemShut {NoStop}%
\bibitem [{\citenamefont {Cho}\ and\ \citenamefont
  {Bang}(2025{\natexlab{a}})}]{cho2025fundamental}%
  \BibitemOpen
  \bibfield  {author} {\bibinfo {author} {\bibfnamefont {K.}~\bibnamefont
  {Cho}}\ and\ \bibinfo {author} {\bibfnamefont {J.}~\bibnamefont {Bang}},\
  }\href@noop {} {\bibfield  {journal} {\bibinfo  {journal} {arXiv preprint
  arXiv:2509.26373}\ } (\bibinfo {year} {2025}{\natexlab{a}})}\BibitemShut
  {NoStop}%
\bibitem [{\citenamefont {Cho}\ and\ \citenamefont
  {Bang}(2025{\natexlab{b}})}]{cho2025witness}%
  \BibitemOpen
  \bibfield  {author} {\bibinfo {author} {\bibfnamefont {K.}~\bibnamefont
  {Cho}}\ and\ \bibinfo {author} {\bibfnamefont {J.}~\bibnamefont {Bang}},\
  }\href@noop {} {\bibfield  {journal} {\bibinfo  {journal} {arXiv preprint
  arXiv:2511.21079}\ } (\bibinfo {year} {2025}{\natexlab{b}})}\BibitemShut
  {NoStop}%
\bibitem [{\citenamefont {Mele}(2024)}]{Mele2024introductiontohaar}%
  \BibitemOpen
  \bibfield  {author} {\bibinfo {author} {\bibfnamefont {A.~A.}\ \bibnamefont
  {Mele}},\ }\href {\doibase 10.22331/q-2024-05-08-1340} {\bibfield  {journal}
  {\bibinfo  {journal} {{Quantum}}\ }\textbf {\bibinfo {volume} {8}},\ \bibinfo
  {pages} {1340} (\bibinfo {year} {2024})}\BibitemShut {NoStop}%
\bibitem [{\citenamefont {Collins}\ and\ \citenamefont
  {{\'S}niady}(2006)}]{Collins:2006jgn}%
  \BibitemOpen
  \bibfield  {author} {\bibinfo {author} {\bibfnamefont {B.}~\bibnamefont
  {Collins}}\ and\ \bibinfo {author} {\bibfnamefont {P.}~\bibnamefont
  {{\'S}niady}},\ }\href {\doibase 10.1007/s00220-006-1554-3} {\bibfield
  {journal} {\bibinfo  {journal} {Commun. Math. Phys.}\ }\textbf {\bibinfo
  {volume} {264}},\ \bibinfo {pages} {773} (\bibinfo {year}
  {2006})}\BibitemShut {NoStop}%
\end{thebibliography}

%

\end{document}